\documentclass[letterpaper,11pt]{article}
\usepackage[margin=1in]{geometry}
\usepackage{amsmath,amsthm,amssymb,amsfonts,mathtools}
\usepackage{color,url,xspace,ascmac}
\usepackage{algorithm,algorithmic}
\usepackage[normalem]{ulem}
\usepackage{todonotes,comment,tikz}
\usepackage{thmtools,thm-restate}
\usepackage{enumitem}
\usepackage{physics}
\usepackage{tikz}
\usepackage{lscape}
\usepackage{subcaption}

\usepackage{fancybox}

\usetikzlibrary{positioning,calc}

\tikzset{every picture/.style={font issue=\footnotesize},
            font issue/.style={execute at begin picture={#1\selectfont}}
           }

\usepackage[hypertexnames=false,unicode,colorlinks=true,citecolor=blue,linkcolor=blue,pdfusetitle]{hyperref}
\usepackage[style=alphabetic,natbib=true,natbib=true,maxnames=99,maxalphanames=99,isbn=false,url=false,backref=true,backend=biber,giveninits=true]{biblatex}

\AtEveryBibitem{
\ifentrytype{article}{
  \clearlist{publisher}
  \clearfield{editor}
}{}
\ifentrytype{book}{
}{}
\ifentrytype{inproceedings}{
  \clearfield{volume}
  \clearfield{number}
  \clearfield{editor}
  \clearlist{publisher}
  \clearfield{location}
  \clearfield{series}
}{}
\clearfield{day}
\clearfield{month}
\clearfield{isbn}
\clearfield{url}
}

\usepackage{doi}
\usepackage[capitalise]{cleveref}
\newtheorem{theorem}{Theorem}[section]
\newtheorem{lemma}[theorem]{Lemma}
\newtheorem{corollary}[theorem]{Corollary}
\newtheorem{proposition}[theorem]{Proposition}

\newtheorem{remark}[theorem]{Remark}

\newtheorem{definition}[theorem]{Definition}

\def\defeq{\mathrel{\mathop:}=}

\DeclarePairedDelimiter{\rbra}{\lparen}{\rparen} 
\DeclarePairedDelimiter{\cbra}{\lbrace}{\rbrace} 
\DeclarePairedDelimiter{\sbra}{\lbrack}{\rbrack} 

\newcommand{\indicator}{\mathbf{1}}
\DeclareMathOperator*{\E}{\mathbb{E}}
\DeclareMathOperator*{\Var}{\mathbf{Var}}
\newcommand{\condition}{\;\middle\vert\;}
\renewcommand{\epsilon}{\varepsilon}

\newcommand{\e}{\mathrm{e}}

\newcommand{\Nat}{\mathbb{N}}
\newcommand{\Real}{\mathbb{R}}
\newcommand{\Int}{\mathbb{Z}}

\newcommand{\binset}{\{0,1 \}}

\newcommand{\Otilde}{\widetilde{O}}

\newcommand{\constref}[1]{C_{\mbox{\tiny\ref{#1}}}}
\newcommand{\constrefs}[2]{C_{\mbox{\tiny\ref{#1}(\ref{#2})}}}
\newcommand{\varianceref}[1]{\variance_{\mbox{\tiny\ref{#1}}}}

\newcommand{\calF}{\mathcal{F}}

\newcommand{\taucons}{\tau_{\mathrm{cons}}}
\newcommand{\alphanorm}{\gamma}

\newcommand{\taudeltaplus}{\tau^{+}_\delta}
\newcommand{\tauetaplus}{\tau^{+}_\eta}

\newcommand{\cweak}{c^{\mathrm{weak}}}
\newcommand{\calphaup}{c^{\uparrow}_\alpha}
\newcommand{\calphadown}{c^{\downarrow}_\alpha}
\newcommand{\cdeltaplus}{c^{+}_\delta}

\newcommand{\opinion}{\mathsf{opn}}
\newcommand{\bounded}{D}
\newcommand{\variance}{s}

\newcommand{\taunormup}{\tau^{\uparrow}_\alphanorm}
\newcommand{\taunormdown}{\tau^{\downarrow}_\alphanorm}
\newcommand{\taunormplus}{\tau^{+}_{\alphanorm}}
\newcommand{\taudisappeari}{\tau^*_i}
\newcommand{\tauincreasei}{\tau^{\mathrm{active}}_i}
\newcommand{\calE}{\mathcal{E}}

\newcommand{\cnormup}{c^{\uparrow}_\alphanorm}
\newcommand{\cnormdown}{c^{\downarrow}_\alphanorm}

\newcommand{\alphain}{\alpha^{\mathrm{in}}}
\newcommand{\alphaout}{\alpha^{\mathrm{out}}}
\newcommand{\Bin}{\mathrm{Bin}}

\newcommand{\taudeltaup}{\tau^{\uparrow}_\delta}
\newcommand{\cdeltaup}{c^{\uparrow}_\delta}
\newcommand{\taudeltadown}{\tau^{\downarrow}_\delta}
\newcommand{\cdeltadown}{c^{\downarrow}_\delta}

\newcommand{\tauiup}{\tau^{\uparrow}_i}
\newcommand{\ciup}{c_\alpha^{\uparrow}}
\newcommand{\tauidown}{\tau^{\downarrow}_i}
\newcommand{\cidown}{c_\alpha^{\downarrow}}
\newcommand{\tauiweak}{\tau^{\mathrm{weak}}_i}
\newcommand{\taujweak}{\tau^{\mathrm{weak}}_j}

\newcommand{\cactive}{c^{\mathrm{active}}}

\newcommand{\taujup}{\tau^{\uparrow}_j}
\newcommand{\taujdown}{\tau^{\downarrow}_j}

\newcommand{\drift}{R}
\newcommand{\tinyexists}{{}^{\exists}}
\newcommand{\tinyforall}{{}^{\forall}}

\newcommand{\taux}{\tau_X}
\newcommand{\tauxplus}{\tau_X^+}
\newcommand{\tauxminus}{\tau_X^-}

\newcommand{\deltaratio}{\eta}
\newcommand{\tauetaup}{\tau^{\uparrow}_\eta}
\newcommand{\cetaup}{c^\uparrow_\eta}
\newcommand{\tauivanish}{\tau^{\mathrm{vanish}}_i}

\newcommand{\xdelta}{x_{\delta}}
\newcommand{\xeta}{x_{\eta}}

\newcommand{\Cdeltavarplus}{C_{\delta}}

\newcommand{\tauiactive}{\tau^{\mathrm{active}}_i}
\newcommand{\taustar}{\tau^*}

\newcommand{\xnorm}{x_{\alphanorm}}
\newcommand{\driftnorm}{\drift_{\alphanorm}}

\newcommand{\tauphiplus}{\tau_\varphi^+}
\newcommand{\tauphiup}{\tau^\uparrow_\varphi}

\newcommand{\cphiup}{c_\varphi^{\uparrow}}

\crefname{figure}{Figure}{Figures}
\crefname{remark}{Remark}{Remarks}
\crefname{equation}{}{}

\makeatother

\title{3-Majority and 2-Choices with Many Opinions}

\author{
    Nobutaka Shimizu\\
    \small{Institute of Science Tokyo}\\
    \small{\texttt{\href{shimizu.n.ah@m.titech.ac.jp}{shimizu.n.ah@m.titech.ac.jp}}}
   \and
    Takeharu Shiraga\\
    \small{Chuo University}\\
   \small{\texttt{\href{shiraga.076@g.chuo-u.ac.jp}{shiraga.076@g.chuo-u.ac.jp}}}
}

\date{\today}
\begin{document}
\maketitle

\begin{abstract}
  We present the first nearly‑optimal bounds on the consensus time for the well‑known synchronous consensus dynamics, specifically 3‑Majority and 2‑Choices, for an \emph{arbitrary} number of opinions. 
  In synchronous consensus dynamics, we consider an $n$-vertex complete graph with self‑loops, where each vertex holds an opinion from $\{1,\dots,k\}$.
  At each discrete-time round, all vertices update their opinions simultaneously according to a given protocol. 
  The goal is to reach a consensus, where all vertices support the same opinion.
  In 3-Majority, each vertex chooses three random neighbors with replacement and updates its opinion to match the majority, with ties broken randomly.
  In 2‑Choices, each vertex chooses two random neighbors with replacement. 
  If the selected vertices hold the same opinion, the vertex adopts that opinion. 
  Otherwise, it retains its current opinion for that round.

  Improving upon a line of work [Becchetti et al., SPAA'14], [Becchetti et al., SODA'16], [Berenbrink et al., PODC'17], [Ghaffari and Lengler, PODC'18], we prove that, for every $2\le k \le n$, 
  3-Majority (resp.\ 2-Choices) reaches consensus within $\widetilde{\Theta}(\min\{k,\sqrt{n}\})$ (resp.\ $\widetilde{\Theta}(k)$) rounds with high probability. 
  Prior to this work, the best known upper bound on the consensus time of 3-Majority was $\widetilde{O}(k)$ if $k \ll n^{1/3}$ and $\widetilde{O}(n^{2/3})$ otherwise, and for 2-Choices, the consensus time was known to be $\widetilde{O}(k)$ for $k\ll \sqrt{n}$.
\end{abstract}

\tableofcontents
\section{Introduction}
We present nearly tight bounds on the convergence time of two well-known consensus dynamics: 3-Majority and 2-Choices. These bounds apply to any number of opinions under the synchronous update rule on a complete graph with self-loops. Specifically, we provide upper and lower bounds that differ by at most polylogarithmic factors.

In synchronous consensus dynamics, we consider a distributed system consisting of an $n$-vertex graph where each vertex holds an element from a finite set $[k]=\{1,\dots,k\}$, referred to as an \emph{opinion}. At each discrete-time round, all vertices simultaneously update their opinions according to a protocol. The goal is to reach a consensus, where all vertices support the same opinion, which must be initially supported by at least one vertex (validity condition). Additionally, the protocol should satisfy the plurality condition: if the most popular initial opinion has a sufficiently large margin, consensus will favor this opinion. The main quantity of interest is the consensus time, the number of rounds required to reach consensus. For background and applications of consensus dynamics, see \cite{consensus_dynamics_becchetti20} and references therein.

3-Majority and 2-Choices are simple probabilistic protocols that satisfy both validity and plurality conditions while achieving a small consensus time with high probability. In 3-Majority, each vertex $u$ chooses three random neighbors with replacement and updates its opinion to match the majority, with ties broken randomly. In 2-Choices, each vertex $u$ chooses two random neighbors with replacement. If the selected vertices hold the same opinion $\sigma$, $u$ updates its opinion to $\sigma$. Otherwise, $u$ does not change its opinion in that round.

Throughout this paper, unless otherwise noted, the underlying graph is the $n$-vertex complete graph with self-loops (thus, choosing a random neighbor corresponds to choosing a vertex uniformly at random). The main result of this paper is as follows:
\begin{theorem}[Main] \label{thm:main theorem}
The consensus time of 3-Majority is $\widetilde{\Theta}(\min\{\sqrt{n},k\})$\footnote{$\widetilde{\Theta}(\cdot)$ and $\widetilde O(\cdot)$ hide polylogarithmic factors.} with high probability\footnote{The term ``with high probability'' means that the event holds with probability $1-O(n^{-c})$ for some constant $c>0$.} for all $2\le k\le n$.
Moreover, if $k = o(\sqrt{n}/\log n)$ and the most popular opinion is supported by $\omega(\sqrt{n\log n})$ more vertices than any other opinion, then 3-Majority reaches consensus on the most popular opinion with high probability.

Similarly, the consensus time of 2-Choices is $\tilde{\Theta}(k)$ with high probability for all $2\le k \le n$. Moreover, if $k=o(n/(\log n)^2)$ and 
the most popular opinion is supported by $\omega(\sqrt{\alpha_1 n\log n})$ more vertices than any other opinion where $\alpha_1$ is the fraction of vertices supporting the most popular opinion, then 2-Choices reaches consensus on the most popular opinion.
\end{theorem}

Prior to this work, the best known upper bound for the consensus time of 3-Majority was $\Otilde(k)$ if $k= O(n^{1/3}/\sqrt{\log n})$ and $\Otilde(n^{2/3})$ otherwise \cite{nearly_tight_analysis,ignore_or_comply}.
For 2-Choices, the consensus time was known to be $\Otilde(k)$ for $k=O(\sqrt{n/\log n})$ \cite{nearly_tight_analysis}.
\Cref{thm:main theorem} improves both of these bounds, as summarized in \cref{fig:comparison}.
In particular, for 2-Choices, \cref{thm:main theorem} provides the first upper bound that holds for \emph{any} $k$.
For more detailed results that take the logarithmic terms into account, see \cref{sec:proof_outline} (\cref{thm:consensus time large alphanorm,thm:growth of alphanorm,thm:plurality,thm:lowerbound}).

\begin{figure}[t]
  \begin{center}
    \begin{minipage}{0.45\textwidth}
      \includegraphics[width=\textwidth]{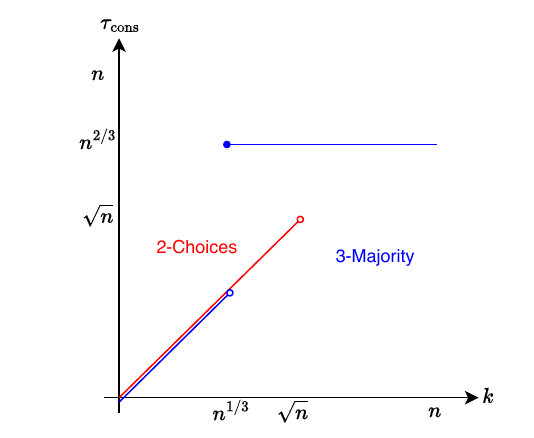}
      \subcaption{Prior to this work}
     \end{minipage}
    \hfill
    \begin{minipage}{0.45\textwidth}
      \includegraphics[width=\textwidth]{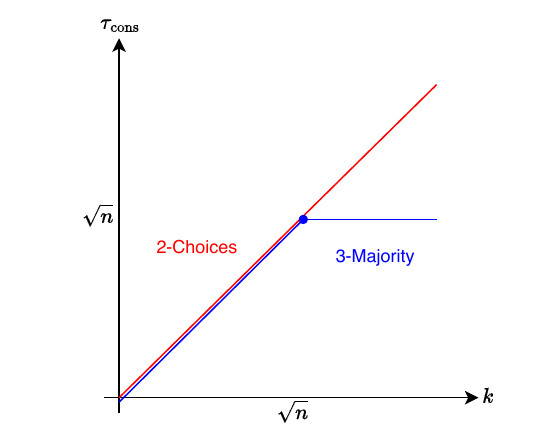}
      \subcaption{This work (\cref{thm:main theorem})}
    \end{minipage}
  \end{center}
  \caption{Upper bounds on the consensus time of 3-Majority (blue) and 2-Choices (red). Here, we ignore polylogarithmic factors. \label{fig:comparison}}
\end{figure}

\subsection{Related Results}
3-Majority on a complete graph with multiple opinions was initially studied by \cite{simple_dynamics} and subsequently by \cite{stabilizing_consensus,nearly_tight_analysis,ignore_or_comply}.
\citet{simple_dynamics} showed that the consensus time is $O(k\log n)$ with high probability for $k = O((n/\log n)^{1/3})$, assuming the most popular opinion has a significant margin.
\citet{stabilizing_consensus} removed this margin condition, proving a consensus time of $\Otilde(k^3)$ with high probability for $k = O(n^{1/3-\varepsilon})$ for any $\varepsilon>0$.
\citet{nearly_tight_analysis} improved this to $O(k\log n)$ for $k \le O(n^{1/3} / \sqrt{\log n})$.
For larger $k$, \citet{ignore_or_comply} showed that after $T$ steps, the number of remaining opinions is at most $O(n\log n / T)$ with high probability.
Combining this with \cite{nearly_tight_analysis}, the consensus time is $O(k\log n)=\Otilde(k)$ for $k=O(n^{1/3}/\sqrt{\log n})$ and $O(n^{2/3}(\log n)^{3/2})=\Otilde(n^{2/3})$ otherwise.

2-Choices was first implicitly studied in \cite{Doerr11}.
In their protocol, each vertex $u$ takes the median of its own opinion and those of two randomly chosen neighbors. For $k=2$, this coincides with 2-Choices, and they proved a consensus time of $O(\log n)$ with high probability. This proof technique also applies to 3-Majority, yielding the same upper bound.
Subsequent works \cite{twochoice_ICALP14,twochoice_expander_DISC15,twochoice_expander_DISC17,CNS19,SS19,quasi-majority,CNN+18} focused on 2-Choices and 3-Majority for $k=2$ on various graph classes (e.g., expander, stochastic block model, core-periphery graph).
For $k\ge2$, \citet{ignore_or_comply} proved a general lower bound.
For example, starting with the balanced initial configuration, the consensus time is $\Omega\qty(\min\{k,n/\log n\})$ with high probability for any $2 \le k\le n$.
This matches the lower bound of \cref{thm:main theorem} for 2-Choices.
\citet{nearly_tight_analysis} proved a consensus time of $O(k\log n)$ with high probability for $k=O(\sqrt{n/\log n})$.

In the asynchronous model, where a uniformly random vertex updates its opinion each round, the consensus time of 3-Majority was studied by \citet{hierarchy_Berenbrink} (for $k=2$) and \citet{async_3-majority_soda25} (for general $k$).
\citet{async_3-majority_soda25} showed that the consensus time is $\Otilde(\min(kn,n^{3/2}))$ with high probability for all $k\le n$ and any initial opinion configuration.
Considering that one round of synchronous dynamics equates to $n$ rounds of asynchronous dynamics, their result implies a consensus time of $\Otilde(\min(\sqrt{n},k))$ for synchronous 3-Majority.
However, their proof technique does not directly apply to synchronous dynamics, leaving the consensus time for synchronous dynamics as an open problem.
In \cref{sec:make rigorous}, we discuss the main obstacles in applying their technique to synchronous dynamics and how we overcome them.

\section{Proof Outline} \label{sec:proof_outline}
In this section, we outline the essential ideas underlying the proof of \cref{thm:main theorem}, focusing on the upper bound. 
First, we introduce two general results (\cref{thm:consensus time large alphanorm,thm:growth of alphanorm}) that form the upper bound of \cref{thm:main theorem}. 
Next, we present a heuristic argument for their proofs, focusing on 3‑Majority, and explain how to make this argument rigorous using Freedman's inequality. 
Lastly, we offer additional remarks on the general results regarding plurality consensus and lower bounds (\cref{thm:plurality,thm:lowerbound}). 
We conclude this section by listing some open problems.

We begin by introducing some notation.
For a given opinion $i\in [k]$, we define $\alpha_t(i)$ as the fraction of vertices that support opinion $i$ at round $t$. 
The key quantity of interest is the $\ell^2$-norm $$\alphanorm_t := \sum_{i\in [k]} \alpha_t(i)^2.$$
Note that $\alphanorm_t\geq 1/k$ holds for any $t$
since $1=(\sum_{i\in [k]}\alpha_t(i))^2\leq \sum_{i\in [k]}\alpha_t(i)^2\sum_{i\in [k]}1^2=\alphanorm_t k$ from the Cauchy-Schwarz inequality.
\subsection{General Results on Upper Bounds}
We introduce two general results that lead to the upper bound results of \cref{thm:main theorem}. 
The first shows that if the initial value of the $\ell^2$-norm $\alphanorm_0$ is sufficiently large, 
then the consensus times of 3-Majority and 2-Choices are $O\qty(\frac{\log n}{\alphanorm_0})$ with high probability.

\begin{theorem}[Starting from large $\alphanorm_0$] \label{thm:consensus time large alphanorm}
  The consensus time of 3-Majority starting from any initial configuration provided that $\alphanorm_0 \ge \frac{C\log n}{\sqrt{n}}$ for a sufficiently large constant $C>0$ is $O\qty(\frac{\log n}{\alphanorm_0})$ with high probability.

  Similarly, the consensus time of 2-Choices starting from any initial configuration provided that $\alphanorm_0 \ge \frac{C(\log n)^2}{n}$ for a sufficiently large constant $C>0$ is $O\qty(\frac{\log n}{\alphanorm_0})$ with high probability.
\end{theorem}
Since $\alphanorm_0 \geq 1/k$, \cref{thm:consensus time large alphanorm} implies that the consensus time is $O(k \log n)$ with high probability for 3-Majority when $k =o( \sqrt{n}/\log n)$ and for 2-Choices when $k =o( n/(\log n)^2)$. These bounds match those of \cref{thm:main theorem} for such small $k$. Notably, these ranges of $k$ improve upon the previously best-known results~\cite{nearly_tight_analysis}, where the $O(k \log n)$ consensus time was shown for 3-Majority with $k = O(n^{1/3}/\sqrt{\log n})$ and for 2-Choices with $k = O(\sqrt{n/\log n})$.

The second general result guarantees that even when the $\ell^2$-norm $\alphanorm_0$ is initially small, it rapidly increases to a regime where \cref{thm:consensus time large alphanorm} becomes applicable.

\begin{theorem}[Growth of $\alphanorm_t$] \label{thm:growth of alphanorm}
  Let $c_*>0$ be any constant.
  For 3-Majority starting from any initial configuration, with high probability, we have $\alphanorm_T \ge \frac{c_*\log n}{\sqrt{n}}$ for some $T=O\qty(\sqrt{n}(\log n)^{2})$.

  Similarly, for 2-Choices starting from any initial configuration, with high probability, we have $\alphanorm_T \ge \frac{c_*(\log n)^2}{n}$ for some $T=O(n(\log n)^3)$.
\end{theorem}

Combining \cref{thm:consensus time large alphanorm,thm:growth of alphanorm}, 
we immediately obtain the following upper bounds that hold for any initial configuration.
For 3-Majority, the consensus time is $O\qty(\sqrt{n}(\log n)^2+ \frac{\log n}{\log n/\sqrt{n}})=O(\sqrt{n}(\log n)^2)$, which improves upon the bound of $O(n^{2/3}(\log n)^{3/2})$ \cite{nearly_tight_analysis,ignore_or_comply}.
For 2-Choices, the consensus time is $O\qty(n(\log n)^3+ \frac{\log n}{ (\log n)^2/n})=O(n(\log n)^3)$, which further extends the range of $k$ in \cref{thm:consensus time large alphanorm} and is the first bound that holds for any $k$.
These bounds complement the upper bounds of \cref{thm:main theorem} for large $k$.

\subsection{Heuristic Argument for 3-Majority} \label{sec:heuristic argument}
Now, we present a heuristic argument for the proof of \cref{thm:consensus time large alphanorm,thm:growth of alphanorm}. 
We often use $\E_{t-1}[\cdot]$ to denote the expectation conditioned on the configuration at round $t-1$ (see \cref{sec:consensus dynamics} for details).
For example, a straightforward calculation (also used in previous works~\cite{simple_dynamics,stabilizing_consensus,nearly_tight_analysis,ignore_or_comply}; see \cref{lem:basic inequality} for details) shows that the expectation of $\alpha_t(i)$ conditioned on the configuration at round $t-1$ satisfies
\begin{align}
  \E_{t-1}[\alpha_t(i)] = \alpha_{t-1}(i)(1+\alpha_{t-1}(i) - \alphanorm_{t-1}) \label{eq:expectation of alpha intro}.
\end{align}
In view of \cref{eq:expectation of alpha intro}, one might expect that $\alpha_t(i)$ is likely to decrease if $\alpha_{t-1}(i) \ll \alphanorm_{t-1}$.
With this in mind, we say that an opinion is \emph{weak} at round $t$ if $\alpha_t(i) < (1-c)\alphanorm_t$, where $0<c<1/2$ is some suitable constant. Otherwise, we say that $i$ is \emph{strong}.
Observe that the most popular opinion is always strong in every round since $\max_i \alpha_t(i) \ge \alphanorm_t$.

\paragraph*{Weak Opinion Vanishing.}
We first show that within $O\left(\frac{\log n}{\alphanorm_0}\right)$ rounds, any weak opinion $i$ is likely to vanish.
\begin{lemma}[Weak Opinion Vanishing; see also \cref{lem:weakvanish}] \label{lem:weak opinion vanishing intro}
  Consider 3-Majority starting from an initial configuration with $\alphanorm_0 \ge \frac{C\log n}{\sqrt{n}}$ for a sufficiently large constant $C>0$.
  If an opinion $i$ is weak at round $0$, then $\alpha_T(i) = 0$ with probability $1-O(n^{-3})$ for some $T = O\qty(\frac{\log n}{\alphanorm_0})$.
\end{lemma}
Although our formal proof of \cref{lem:weak opinion vanishing intro} is more involved, the intuition behind it is based on the following heuristic argument.
For any weak opinion $i$, from \cref{eq:expectation of alpha intro}, we have $\E_{t-1}[\alpha_t(i)] = \alpha_{t-1}(i)(1+\alpha_{t-1}(i) - \alphanorm_{t-1}) \le (1-c\alphanorm_{t-1})\alpha_{t-1}(i)$.
Therefore, $\alpha_t(i)$ decreases by a factor of $1-c\alphanorm_{t-1}$ in every round in expectation.
To prove that $\alpha_t(i)$ vanishes quickly, we need to keep track of the value of $\alphanorm_t$.
Indeed, by a somewhat involved calculation (see \cref{lem:basic inequality} for details), we can show that
\begin{align}
  \E_{t-1}[\alphanorm_t] \ge \alphanorm_{t-1} + \frac{1-\alphanorm_{t-1}}{n} \ge \alphanorm_{t-1}. \label{eq:expectation of alphanorm intro}
\end{align}
In particular, $\alphanorm_t$ does not decrease in expectation during the dynamics (i.e., $\alphanorm_t$ is a submartingale).
This yields that $\alphanorm_t \gtrsim \alphanorm_0$ and thus $\alpha_t(i)$ vanishes within $O\qty(\frac{\log n}{\alphanorm_0})$ rounds.

\paragraph{Strong Opinion Weakening.}
Next, consider two distinct strong opinions, $i$ and $j$.
We claim that at least one of them becomes weak within $O\qty(\frac{\log n}{\alphanorm_0})$ rounds.
\begin{lemma}[Strong Opinion Weakening; see also \cref{lem:gap amplification,lem:initial bias weak}] \label{lem:strong opinion weakening intro}
  Consider 3-Majority starting with any initial configuration satisfying $\alphanorm_0 \ge C\sqrt{\frac{\log n}{n}}$ for a sufficiently large constant $C>0$.
  Then, there exists some $T = O\qty(\frac{\log n}{\alphanorm_0})$ such that,
  for any two distinct strong opinions $i$ and $j$, either $i$ or $j$ becomes weak within $T$ rounds with probability $1-O(n^{-3})$.
\end{lemma}

Here, the condition $\alphanorm_0 \gg \sqrt{\frac{\log n}{n}}$ is slightly weaker than the condition $\alphanorm_0\gg\frac{\log n}{\sqrt{n}}$ of \cref{lem:weak opinion vanishing intro}.

The intuition behind \cref{lem:weak opinion vanishing intro} is as follows:
Fix two strong opinions $i,j$ and let $\delta_t = \alpha_t(i) - \alpha_t(j)$.
We may assume that $\delta_0 \geq  0$ without loss of generality.
From \cref{eq:expectation of alpha intro}, we have
\begin{align}
  \E_{t-1}[\delta_t] = (1+\alpha_{t-1}(i) + \alpha_{t-1}(j) - \alphanorm_{t-1}) \delta_{t-1} \label{eq:expectation of delta intro}.
\end{align}
Since $i,j$ are strong and $c<1/2$, we have $\E_{t-1}[\delta_t] \ge (1+(1-2c)\alphanorm_{t-1}) \delta_{t-1} \ge (1+\Omega(\alphanorm_{t-1}))\delta_{t-1}$.
Since $\alphanorm_t \gtrsim \alphanorm_0$, we have that $\delta_t$ increases by a factor of $(1+\Omega(\alphanorm_0))$ at every round in expectation unless either $i$ or $j$ become weak (see \cref{lem:deltaupweak} for details).
Moreover, even if the bias is initially zero, we can show that $\abs{\delta_T}$ grows to $\Omega(\sqrt{T\alphanorm_0/n})$ for a suitable choice of $T$.
The key insight is that the squared bias $\delta_t^2$ for two strong opinions $i,j$ exhibits an additive drift: by considering the variance of $\delta_t$, we establish that $\E_{t-1}[\delta_t^2]\geq \delta_{t-1}^2+\Omega(\alphanorm_0/n)$ (see \cref{,lem:additive drift} for details).
%
%
%
Combining them, we can conclude that either $i$ or $j$ becomes weak within $O\qty(\frac{\log n}{\alphanorm_0})$ rounds (otherwise, $\abs{\delta_t}$ becomes too large).

\paragraph*{Putting Them Together.}
Combining \cref{lem:weak opinion vanishing intro,lem:strong opinion weakening intro},
we can conclude that for any pair of distinct opinions $i,j$, at least one of them vanishes within \(O\left(\frac{\log n}{\alphanorm_0}\right)\) rounds with probability \(1-O(n^{-3})\).
By the union bound over $i,j$, we obtain \cref{thm:consensus time large alphanorm}.

On the other hand, from \cref{eq:expectation of alphanorm intro}, we know that $\alphanorm_t$ increases by $\Omega(1/n)$ at every round in expectation unless $\alphanorm_t \le 1/2$.
In particular, by our concentration technique explained in \cref{sec:make rigorous}, we can prove that $\alphanorm_T \approx \log n/\sqrt{n}$ for some $T=\Otilde(\sqrt{n})$, which yields \cref{thm:growth of alphanorm}.

\begin{remark}
  While \cref{thm:consensus time large alphanorm} bounds the consensus time for any $1\le k\ll \sqrt{n}$,
  the case of $k \ge \sqrt{n}$ can be handled by the result of \cite{ignore_or_comply}:
  They proved that the number of remaining opinions after $T$ rounds of 3-Majority is at most $O(n\log n / T)$ with high probability.
  Combined \cref{thm:consensus time large alphanorm} with their result for $T=\sqrt{n}\log n$, we can conclude that the consensus time is $\Otilde(\sqrt{n})$ with high probability for all $2\le k\le n$.
  However, their result does not hold for 2-Choices, whereas our argument based on the increasing of $\alphanorm_t$ (\cref{thm:growth of alphanorm}) can be applied to 2-Choices.
\end{remark}

These arguments can be extended to the 2-Choices dynamics, yielding a similar consensus time bound.
Specifically, \cref{lem:weak opinion vanishing intro,lem:strong opinion weakening intro} hold for 2-Choices as well if $\alphanorm_0 \ge \frac{(\log n)^2}{n}$.
The main difference is that the additive drift of $\alphanorm_t$ in 2-Choices is $\Omega\qty(\frac{1}{n^2})$, which is much smaller than that of 3-Majority.
This yields that $\alphanorm_t \ge \frac{(\log n)^2}{n}$ within $\Otilde(n)$ rounds in expectation.

We note that our argument seemingly simplifies the analysis of \cite{nearly_tight_analysis}, which classifies the opinions into three classes, divides time into epochs which consist of several consecutive rounds, and each epoch is further divided into two phases.

Interestingly, our argument can also be extended to the asynchronous 3-Majority dynamics, providing an alternative proof of the result of \cite{async_3-majority_soda25}.
We believe that our argument is simpler than the original proof of \cite{async_3-majority_soda25}.
In particular, \cite{async_3-majority_soda25} extended the proof technique of \cite{ignore_or_comply} to the asynchronous setting with a complicated coupling argument from Majorization Theory \cite{MOA11}.
We avoid this complication by directly analyzing the growth of $\alphanorm_t$.

\subsection{Making the Heuristic Argument Rigorous} \label{sec:make rigorous}
In \cref{sec:heuristic argument}, we presented a heuristic argument for the consensus time of the 3-Majority dynamics based on the expected behavior of $\alpha_t(i),\alphanorm_t$, and $\delta_t$.
To make it rigorous, we need concentration inequalities to show that the actual behavior of $\alpha_t(i), \alphanorm_t$, and $\delta_t$ are close to their expected values.

\paragraph*{Na\"ive Approach: One-Step Concentration via the Chernoff Bound.}
The most straightforward way to make the heuristic argument rigorous is to apply the Chernoff bound to argue that $\alpha_t(i) \approx \E_{t-1}[\alpha_t(i)]$ since $\alpha_t(i)$ can be written as the sum of $n$ independent random variables.
This approach was used in many previous works \cite{simple_dynamics,stabilizing_consensus,nearly_tight_analysis} in the range of $k \ll n^{1/3}$.

Unfortunately, this approach is not sufficient for the case of $k \gg n^{1/3}$.
In the balanced configuration where $\alpha_{t-1}(i) \approx 1/k$, we have that the variance $\Var_{t-1}[\alpha_t(i)]$ is roughly $\Theta(1/k)$.
Therefore, by the central limit theorem, we can argue that $\alpha_t(i) \approx \E_{t-1}[\alpha_t(i)] \pm \Theta(1/\sqrt{nk})$ at every round.
On the other hand, in the proof of \cref{lem:weak opinion vanishing intro}, we used the fact that $\alpha_t(i)$ for a weak opinion $i$ drops by a multiplicative factor of $1-\Omega(\alphanorm_0) = 1-\Omega(1/k)$.
In summary, the one-step concentration yields that 
\begin{align*}
  \alpha_t(i) \approx \qty(1-\Omega\qty(\frac{1}{k}))\alpha_{t-1}(i) \pm \Theta\qty(\frac{1}{\sqrt{nk}}) \approx \alpha_{t-1}(i) - \Omega\qty(\frac{1}{k^2}) \pm \Theta\qty(\frac{1}{\sqrt{nk}}).
\end{align*}
To ensure that $\alpha_t(i)$ keeps decreasing, we need to have $1/k^2 \gg 1/\sqrt{nk}$, which is equivalent to $k \ll n^{1/3}$.
In other words, the na\"ive approach can only handle the case of $k \ll n^{1/3}$ due to the standard deviation at every round.
This is the main obstruction to extending the proof of \cite{nearly_tight_analysis} to the case of $k \gg n^{1/3}$.

\paragraph*{Our Approach: Multi-Step Concentration via Freedman's Inequality (\cref{sec:drift analysis intro}).}
To remedy the above issue, we track the amortized change of $\alpha_t(i)$ during $T$ rounds.
Recall that the one-step concentration yields that $\alpha_t(i)$ differs from its expectation $\E_{t-1}[\alpha_t(i)]$ by $\Theta(1/\sqrt{nk})$.
Summing up $t=1,\dots,T$, the total gap between $\alpha_t(i)$ and its expectation is roughly $\Theta(T/\sqrt{nk})$.
In contrast, using our multi-step concentration technique described later, we can show that the total gap is indeed $\Theta(\sqrt{T/nk})$, which is much smaller than the na\"ive bound.
This suffices to our purpose since if we set $T\approx k$, then
\begin{align*}
  \alpha_T(i) \approx \alpha_0(i) - \Omega\qty(\frac{T}{k^2}) \pm \Theta\qty(\sqrt{\frac{T}{nk}}) \approx \alpha_0(i) - \Omega\qty(\frac{1}{k}) \pm \Theta\qty(\frac{1}{\sqrt{n}}).
\end{align*}
That is, we can show that $\alpha_T(i)$ is likely to decrease if $k\ll \sqrt{n}$.

The idea of multi-step concentration above appeared in \cite{nearly_tight_analysis} implicitly and was made explicit in \cite{async_3-majority_soda25} for the asynchronous 3-Majority dynamics.

Our multi-step concentration builds upon the \emph{Freedman's inequality}, which is a Bernstein-type concentration inequality for martingales \cite[Theorem 4.1]{Fre75}.
Recall that a sequence of random variables $(X_t)_{t\ge 0}$ is a submartingale if $\E_{t-1}[X_t] \ge X_{t-1}$.
The Freedman's inequality states that for a submartingale $(X_t)_{t\ge 0}$ such that $|X_t - X_{t-1}| \le D$ and $\Var_{t-1}[X_t - X_{t-1}] \le s$ for all $t$, we have 
\begin{align}
  \Pr\qty[\tinyexists t\le T,X_t \le \E[X_t] - h] \le \exp\qty(-\frac{h^2/2}{Ts + hD/3}). \label{eq:Freedman_ineq_intro}
\end{align}
\citet{async_3-majority_soda25} applied the Freedman's inequality to $\alpha_t(i)$ and other quantities to deduce multi-step concentration results in the asynchronous 3-Majority dynamics.
Here, they crucially used the fact that the one-step difference $\alpha_t(i) - \alpha_{t-1}(i)$ is at most $1/n$, which enables to set $D=1/n$ in the Freedman's inequality.
However, in the synchronous dynamics, $\alpha_t(i) - \alpha_{t-1}(i)$ can be $1$, which prevents us from applying the Freedman's inequality directly.
This is one of the main reason why the proof of \cite{async_3-majority_soda25} does not directly apply to the synchronous dynamics.

It is worth noting that, other than \cite{async_3-majority_soda25}, there are some works that (implicitly) consider the multi-step concentration analysis for asynchronous consensus dynamics including undecided dynamics \cite{fast_convergence_undecided} and chemical reaction network \cite{Condon2020-CRN}, where the authors regard the amortized change of quantities of interest as the outcome of a biased random walk.
These analysis compare the probabilities of increase and decrease of the quantity of interest at each step and then apply Gambler's ruin to deduce the concentration result.
Since this approach crucially relies on the boundedness of the one-step difference, it is not directly applicable to the synchronous dynamics.

\paragraph*{Bernstein Condition (\cref{sec:bernstein condition}).}

To apply the Freedman's inequality to $\alpha_t(i)$ in the synchronous dynamics,
  we relax the bounded jump condition that $\abs{X_t - X_{t-1}} \le D$ of the Freedman's inequality.
Specifically, we say that a real-valued random variable $X$ satisfies the \emph{$(D,s)$-Bernstein condition} if for any $-\frac{3}{D}<\lambda<\frac{3}{D}$, we have
\[
  \E\qty[\e^{\lambda X}] \le \exp\qty(\frac{\lambda^2 s^2/2}{1-\abs{\lambda} D/3}).
\]
The intuition behind this condition is that, if $\abs{\lambda X}$ is small enough and $\E[X]=0$, then the Taylor expansion yields
\begin{align*}
  \E\qty[\e^{\lambda X}] \approx \E\qty[ 1 + \lambda X + \frac{\lambda^2 X^2}{2} ] = 1 + \frac{\lambda^2 \Var[X]}{2} \le \exp\qty( \frac{\lambda^2 \Var[X]}{2} ).
\end{align*}
For example, if $\abs{X}\le D$ and $\Var[X]\le s$, then $X$ satisfies the $(D,s)$-Bernstein condition.
It is not hard to see that we can recover the Freedman's inequality \cref{eq:Freedman_ineq_intro} if each one-step difference $X_t - X_{t-1}$ satisfies the Bernstein condition (see \cite{FGL15} and \cref{cor:Freedman} details).

Our key observation is that, if $X$ can be written as the sum of independent random variables $X=Y_1+\dots+Y_m$ and each $Y_j$ satisfies $(D,s)$-Bernstein condition, then $X$ also satisfies the $(D,ms)$-Bernstein condition (see \cref{lem:Bernstein condition} for details).
Since the quantity $\alpha_t(i) - \alpha_{t-1}(i)$ conditioned on round $t-1$
can be written as the sum of $n$ independent random variables each of those satisfying $\qty(\frac{1}{n},s)$-Bernstein condition for some small $s$, we can apply the Freedman's inequality to $\alpha_t(i)$.

\subsection{Results on Plurality Consensus and Lower Bounds}\label{sec:plurality and lower}
In the proof of \cref{thm:main theorem}, we introduce two results that follow the approach of the above proofs.
First, as a byproduct of the proofs of \cref{lem:strong opinion weakening intro,lem:weak opinion vanishing intro} (specifically, \cref{lem:initial bias weak,lem:weakvanish}), we establish the following theorem on plurality consensus.
\begin{theorem}[Plurality consensus]\label{thm:plurality}
    Let $C>0$ be a sufficiently large constant.
    Consider 3-Majority starting with any initial configuration such that
    $\alphanorm_0\geq \frac{C\log n}{\sqrt{n}}$ and $\alpha_0(1)-\alpha_0(j)\geq C\sqrt{\frac{\log n}{n}}$ for all $j\neq 1$.
    Then, for some $T=O\qty(\frac{\log n}{\alphanorm_0})$, we have $\alpha_T(1)=1$ with high probability.

    Similarly, consider 2-Choices starting with any initial configuration such that $\alphanorm_0\geq \frac{C(\log n)^2}{n}$ and $\alpha_0(1)-\alpha_0(j)\geq C\sqrt{\frac{\alpha_0(1)\log n}{n}}$ for all $j\neq 1$.
    Then, for some $T=O\qty(\frac{\log n}{\alphanorm_0})$, we have $\alpha_T(1)=1$ with high probability.
\end{theorem}
\Cref{thm:plurality} presents new results regarding the initial bias required for plurality consensus.
For 3-Majority, under the same assumption on the initial bias (i.e., $\alpha_0(1) - \alpha_0(j) \ge \Omega(\sqrt{\log n / n})$ for all $j \neq 1$), previous work~\cite{simple_dynamics} requires $\alpha_0(1) =\Omega(1)$ (i.e., $\max_{i\in [k]}\alpha_0(i) = \Theta(1)$ and hence $\alphanorm_0 = \Theta(1)$) to achieve the plurality consensus.
This is a much stricter condition than our necessary condition that $\alphanorm_0\geq \Omega(\log n/\sqrt{n})$.
For 2-Choices, previous work by \citet{plurality_2Choices}
requires $\alpha_0(1) - \alpha_0(j) \ge \Omega(\sqrt{\log n/n})$ for all $j \neq 1$ in order to achieve plurality consensus.

Second, we introduce the following lower bound on the consensus time, which is an immediate consequence of the multi-step concentration technique (specifically, \cref{lem:drift analysis for basic}).
\begin{theorem}[Lower bound]
  \label{thm:lowerbound}
  Consider 3-Majority with $k\leq c\sqrt{n/\log n}$ for a sufficiently small constant $c>0$.
  Then, there exists an initial configuration such that the consensus time is $\Omega(k)$ with high probability.

  Similarly, consider 2-Choices with $k\leq cn/\log n$ for a sufficiently small constant $c>0$. 
  Then, there exists an initial configuration such that the consensus time is $\Omega(k)$ with high probability.
\end{theorem}
\Cref{thm:lowerbound} guarantees the tightness of our upper bound results.
For 2-Choices, our lower bound coincides with \cite[Theorem 4.1]{ignore_or_comply}.
For 3-Majority, \cref{thm:lowerbound} establishes the first $\Omega(k)$ lower bounds for $k = \omega((n/\log n)^{1/4})$,
whereas the best previously known results~\cite{simple_dynamics} demonstrated a lower bound of $\Omega(k\log n)$ that holds for $k\le (n/\log n)^{1/4}$. 

Combining the earlier stated upper bound results (\cref{thm:consensus time large alphanorm,thm:growth of alphanorm}) with \cref{thm:plurality,thm:lowerbound}, we obtain \cref{thm:main theorem}.

\subsection{Open Question}
In this paper, we introduce two new technical tools: multi-step concentration via the Bernstein condition and drift analysis of the $\ell^2$-norm.
These tools allow us to derive nearly tight bounds for the consensus time of 3-Majority and 2-Choices across all ranges of $k$.
Additionally, these techniques open up several interesting research directions.

One direction is to apply our techniques to other consensus dynamics.
For instance, the $h$-Majority dynamics \cite{simple_dynamics,hierarchy_Berenbrink} generalizes the 3-Majority dynamics by having each vertex update its opinion to the majority opinion among $h$ randomly chosen neighbors (with ties broken randomly).
Another interesting dynamic is the undecided dynamics, which has been extensively studied in distributed computing \cite{undecided_Angluin,undecided_MFCS,fast_convergence_undecided,Becchetti2015-undecided,undecided-soda22}.
In particular, the consensus time for the undecided dynamics with arbitrary $2 \le k \le n$ opinions remains an open question (for both synchronous and asynchronous settings).

Another promising direction is to study convergence in the presence of an adversary.
In this scenario, an adversary can corrupt the opinions of $F$ vertices each round, where $F = o(n)$.
Previous work \cite{nearly_tight_analysis} shows that the consensus time bound for 3-Majority holds even if $F =O( \sqrt{n}/k^{1.5})$ and $k =O( n^{1/3}/\sqrt{\log n})$.

Finally, it would be interesting to analyze 3-Majority or 2-Choices with many opinions on graphs other than the complete graph.
While the problem on general graphs has been well studied, far less is known for the case of $k \geq 3$ opinions.
For example, the behavior on expander graphs with $k \geq 3$ opinions for any initial configuration remains open and warrants further research.

\subsection{Organization}
In \cref{sec:preliminaries}, we present formal definitions of the 3‑Majority and 2‑Choices processes and introduce the Bernstein condition, which is a key element in our proof. 
In \cref{sec:drift analysis}, we prove concentration results for 3‑Majority and 2‑Choices using Freedman's inequality and the Bernstein condition. 
Finally, using the techniques developed in \cref{sec:preliminaries,sec:drift analysis}, we prove our main result, \cref{thm:main theorem}, in \cref{sec:proof}.

\section{Preliminaries} \label{sec:preliminaries}
For $n\in\Nat$, let $[n]=\{1,\dots,n\}$.
Let $\Nat_0 = \{0\}\cup\Nat$ denote the set of non-negative integers.
Unless otherwise specified, $\log$ denotes the natural logarithm.
For $a,b\in\Real$, let $a \land b = \min\{a,b\}$.
For $x\in \mathbb{R}^n$ and $p\in \mathbb{R}$, let $\norm{x}_p=\qty(\sum_{i\in [n]}x_i^p)^{1/p}$.
For $x\in \mathbb{R}^n$, let $\norm{x}_\infty=\max_{i\in [n]}x_i$.

\subsection{Consensus Dynamics} \label{sec:consensus dynamics}
First, we state the formal definition of the 3-Majority and 2-Choices dynamics as follows.
\begin{definition}[3-Majority and 2-Choices] \label{def:3-Majority and 2-Choices}
  Let $n,k\in \Nat$ be such that $1\le k\le n$.
  The \emph{3-Majority} (or \emph{2-Choices}) 
  is a discrete-time Markov chain $(\opinion_t)_{t\in\Nat_0}$ over the state space $[k]^V$ for a finite set $V$ with $|V|=n$, where $(\opinion_t)_{t\in\Nat_0}$ is defined as follows:

  In 3-Majority, for every $t\ge 1$, $\opinion_t\in[k]^V$ is obtained from $\opinion_{t-1}\in[k]^V$ by the following procedure:
  \begin{enumerate}
      \item For each vertex $v\in V$, select uniformly random $w_1,w_2,w_3\in V$, independent and with replacement.
      \item Define $\opinion_{t}(v)\in [k]$ by
      \[
          \opinion_t(v) = 
          \begin{cases}
              \opinion_{t-1}(w_1) & \text{if $\opinion_{t-1}(w_1)=\opinion_{t-1}(w_2)$}, \\
              \opinion_{t-1}(w_3) & \text{otherwise}.
          \end{cases}
      \]
  \end{enumerate}
  
  In 2-Choices, for every $t\ge 1$, $\opinion_t\in[k]^V$ is obtained from $\opinion_{t-1}\in[k]^V$ by the following procedure:
    \begin{enumerate}
        \item For each vertex $v\in V$, select uniformly random $w_1,w_2\in V$, independently and with replacement.
        \item Define $\opinion_{t}(v)\in [k]$ by
        \[
            \opinion_t(v) = 
            \begin{cases}
                \opinion_{t-1}(w_1) & \text{if $\opinion_{t-1}(w_1)=\opinion_{t-1}(w_2)$}, \\
                \opinion_{t-1}(v) & \text{otherwise}.
            \end{cases}
        \]
    \end{enumerate}

  For both dynamics, the \emph{consensus time} $\tau_{\mathrm{cons}}$ is the stopping time defined by
  \[
      \taucons = \inf\cbra*{t\geq 0 \colon \text{for some $i\in[k]$ and all $v\in V$, $\opinion_t(v)=i$}}.
  \]
  \end{definition}

  Throughout this paper, we are interested in the following quantities.
  \begin{definition}[Basic quantities] \label{def:basic quantities}
    Let $(\opinion_t)_{t\in \Nat_0}$ be 3-Majority or 2-Choices.
    We define the following quantities.
\begin{enumerate}
  \renewcommand{\labelenumi}{(\roman{enumi})}
    \item  The \emph{fractional population} is the sequence of random vectors $(\alpha_t)_{t\in\Nat_0}$ where each $\alpha_t\in[0,1]^k$ is defined by
  \[
      \alpha_t(i) = \frac{\abs*{\{v\in V\colon \opinion_t(v)=i\}}}{n}.
  \]
  \item For $t\ge 0$ and $i,j\in[k]$, the \emph{bias} $\delta_t(i,j)$ is defined as $\delta_t(i,j) \defeq \alpha_t(i) - \alpha_t(j)$. If opinions $i,j$ are clear from context, we use the abbreviated notation $\delta_t=\delta_t(i,j)$.
  \item Let $\gamma_t=\norm{\alpha_t}^2_2=\sum_{i\in [k]}\alpha_t(i)^2$ denote the squared $\ell^2$-norm of $\alpha_t$.
\end{enumerate}
  \end{definition}

  We sometimes use $(\calF_t)_{t\in\Nat_0}$ as a natural filtration of a sequence of random variables of interest to state general results (e.g., \cref{sec:drift analysis intro}), but in our context, we think of $\calF_t$ as the history of configurations up to round $t$, i.e., $\calF_t$ is the natural filtration generated by $(\opinion_s)_{s\le t}$.
  We use $\E_{t-1}[\cdot]=\E[\cdot | \mathcal{F}_{t-1}],\Pr_{t-1}[\cdot]=\Pr[\cdot | \mathcal{F}_{t-1}]$, and $\Var_{t-1}[\cdot]=\Var[\cdot | \mathcal{F}_{t-1}]$ to denote the conditional expectation, probability, and variance with respect to the history up to round $t-1$ (respectively).

  It is easy to see that, for 3-Majority, for any $v\in V$, $i\in [k]$ and $t\geq 1$, 
  \begin{align}
      \Pr_{t-1}\sbra*{\opinion_t(v)=i}
      =\alpha_{t-1}(i)^2+\left(1-\alphanorm_{t-1}\right)\alpha_{t-1}(i)
      =\alpha_{t-1}(i)\rbra*{1+\alpha_{t-1}(i)-\alphanorm_{t-1}}.
      \label{eq:updating probability for 3Majority}
  \end{align}
  Similarly, for 2-Choices, 
  For any $v\in V$, $i\in [k]$ and $t\geq 1$, we have
  \begin{align}
      \label{eq:updating probability for 2 choices}
      \Pr_{t-1}\sbra*{\opinion_t(v)=i}
      &=\begin{cases}
          1-\alphanorm_{t-1}+\alpha_{t-1}(i)^2 & (\textrm{if }\opinion_{t-1}(v)=i)\\
          \alpha_{t-1}(i)^2 & (\textrm{if }\opinion_{t-1}(v)\neq i)
      \end{cases}.
  \end{align}

\subsection{Bernstein Condition} \label{sec:bernstein condition}
  A key component of our concentration bounds is the \emph{Bernstein condition}, which is defined as follows.
    \begin{definition}[Bernstein condition and one-sided Bernstein condition]
    \label{def:Bernstein condition}
    Let $\bounded,\variance\ge 0$ be parameters.
    A random variable $X$ satisfies \emph{$(\bounded,\variance)$-Bernstein condition} if,
    for any $\lambda\in \mathbb{R}$ such that $\abs{\lambda} \bounded < 3$, 
    $\E\qty[ \e^{\lambda X} ]  \le \exp\qty( \frac{\lambda^2 \variance / 2}{1- (\abs{\lambda}\bounded)/3})$.
    We say that $X$ satisfies \emph{one-sided} $(\bounded,\variance)$-Bernstein condition if,
    for any $\lambda\geq 0$ such that $\lambda \bounded < 3$, 
    $\E\qty[ \e^{\lambda X} ]  \le \exp\qty( \frac{\lambda^2 \variance / 2}{1- (\lambda\bounded)/3})$.
    \end{definition}
  The above definition implies that $X$ satisfies $(\bounded,\variance)$-Bernstein condition if both $X$ and $-X$ satisfy one-sided $(\bounded,\variance)$-Bernstein condition.

  There are several related concepts concerning conditions on moment generating functions (see, e.g., \cite{HDS19}).
  In our analysis, the following properties derived from the Bernstein condition are crucial. 
  For instance, the Bernstein condition for sums of independent random variables (\cref{item:BC for independent rvs}) is consistently important in our analysis of the synchronous process.
  Additionally, the Bernstein condition for negatively associated random variables (\cref{item:BC for NA rvs}) helps us analyze the concentration of the norm $\alphanorm_t$.
  \begin{lemma}[Properties for Bernstein condition]
        \label{lem:Bernstein condition}
        Let $X,Y$ be random variables. We have the following:
        \begin{enumerate}
            \renewcommand{\labelenumi}{(\roman{enumi})}
            \item \label{item:Bernstein condition for bounded random variables}
            If $\E[X]=0$ and $\abs{X}\leq \bounded$ for some $\bounded$, then $X$ satisfies $\qty(\bounded,\Var\qty[X])$-Bernstein condition.
            \item \label{item:Bernstein condition 1}
            If $X$ satisfies $\qty(\bounded,\variance)$-Bernstein condition, then
            $X$ satisfies $\qty(\bounded',\variance')$-Bernstein condition for any 
            $\bounded' \ge \bounded$ and $\variance' \ge \variance$.
            Similarly, if $X$ satisfies one-sided $\qty(\bounded,\variance)$-Bernstein condition, then
            $X$ satisfies one-sided $\qty(\bounded',\variance')$-Bernstein condition for any 
            $\bounded' \ge \bounded$ and $\variance' \ge \variance$.
            \item \label{item:Bernstein condition 2}
            If $X$ satisfies $\qty(\bounded,\variance)$-Bernstein condition, then $aX$ satisfies $(\abs{a}\bounded,a^2\variance)$-Bernstein condition for any $a \in \Real$.
            If $X$ satisfies one-sided $\qty(\bounded,\variance)$-Bernstein condition, then $aX$ satisfies one-sided $(a \bounded,a^2\variance)$-Bernstein condition for any $a \geq 0$.
            \item \label{item:dominated Bernstein condition}
            If $X$ satisfies one-sided $\qty(\bounded,\variance)$-Bernstein condition and $Y\preceq X$, 
            then $Y$ satisfies one-sided $\qty(\bounded, \variance)$-Bernstein condition, where $\preceq$ denotes the stochastic domination (see \cref{def:Stochastic domination}).
            In particular,
              if $X$ satisfies one-sided $\qty(\bounded,\variance)$-Bernstein condition and $Y\leq  X$,  
              then $Y$ satisfies one-sided $\qty(\bounded, \variance)$-Bernstein condition.
            \item \label{item:BC for independent rvs} 
            If a sequence of $n$ random variables $X_1,\ldots,X_n$ are independent and $X_i$ satisfies $\qty(D,s_i)$-Bernstein condition for $i\in [n]$, 
            then $\sum_{i\in [n]}X_i$ satisfies $\qty(D,\sum_{i\in [n]}s_i)$-Bernstein condition. 
            \item \label{item:BC for NA rvs} 
            If a sequence of $n$ random variables $X_1,\ldots,X_n$ are negatively associated and $X_i$ satisfies one-sided $\qty(D,s_i)$-Bernstein condition for $i\in [n]$, 
            then $\sum_{i\in [n]}X_i$ satisfies one-sided $\qty(D,\sum_{i\in [n]}s_i)$-Bernstein condition. 
        \end{enumerate}
    \end{lemma}
    \begin{proof}[Proof of \ref{item:Bernstein condition for bounded random variables}]
    Consider $\lambda\in \mathbb{R}$ such that $\abs{\lambda} \bounded<3$.
    For such $\lambda$, we have $\abs{\lambda X}\leq \abs{\lambda}\abs{X}\leq \abs{\lambda} \bounded<3$.
    Hence, applying \cref{lem:Bernstein lemma} to $\lambda X$, 
    \begin{align*}
    \E\qty[\e^{\lambda X}]
    &\leq \E\qty[1+\lambda X+\frac{\lambda^2X^2/2}{1-\abs{\lambda X}/3}] &&(\textrm{\cref{lem:Bernstein lemma}})\\
    &\leq 1+\E\qty[\frac{\lambda^2X^2/2}{1-(\abs{\lambda} D)/3}] &&(\textrm{$\because$ $\E[X]=0$})\\
    &\leq \exp\qty(\frac{\lambda^2\E[X^2]/2}{1-(\abs{\lambda} D)/3}) &&(\textrm{$\because$ $1+x\leq \e^x$})\\
    &= \exp\qty(\frac{\lambda^2\Var[X]/2}{1-(\abs{\lambda} D)/3}) &&(\textrm{$\because$ $\Var[X]=\E[X^2]-\E[X]^2=\E[X^2]$})
    \end{align*}
    and we obtain the claim.
  \end{proof}
    \begin{proof}[Proof of \cref{item:Bernstein condition 1}]
    For the first claim, consider $\lambda\in \mathbb{R}$ such that $\abs{\lambda} \bounded'<3$.
    For such $\lambda$, we also have $\abs{\lambda} \bounded \leq \abs{\lambda} \bounded'<3$.
    Since $X$ satisfies $(\bounded,\variance)$-Bernstein condition, 
        \begin{align*}
            \E\sbra*{\e^{\lambda X}}
            \leq \exp\rbra*{\frac{\lambda^2\variance/2}{1-(\abs{\lambda}\bounded )/3}}
            \leq \exp\rbra*{\frac{\lambda^2\variance'/2}{1-(\abs{\lambda}\bounded' )/3}}
        \end{align*}
        holds, and we obtain the first claim.
        Similarly, for the second claim, consider $\lambda \geq 0$ such that $\lambda \bounded'<3$.
        For such $\lambda$, we also have $\lambda \bounded \leq \lambda \bounded'<3$.
        Since $X$ satisfies one-sided $(\bounded,\variance)$-Bernstein condition,
        \begin{align*}
            \E\sbra*{\e^{\lambda X}}
            \leq \exp\rbra*{\frac{\lambda^2\variance/2}{1-(\lambda\bounded )/3}}
            \leq \exp\rbra*{\frac{\lambda^2\variance'/2}{1-(\lambda\bounded' )/3}}
        \end{align*}
        holds, and we obtain the second claim.
    \end{proof}
    \begin{proof}[Proof of \cref{item:Bernstein condition 2}]
    For the first claim, consider $\lambda\in \mathbb{R}$ such that $\abs{\lambda} \qty(\abs{a} \bounded)<3$.
    For such $\lambda$, we also have $\abs{\lambda a} \bounded \leq \abs{\lambda} \abs{a} \bounded<3$.
    Since $X$ satisfies $(\bounded,\variance)$-Bernstein condition, we have
        \begin{align*}
            \E\sbra*{\e^{\lambda (aX)}}
            =\E\sbra*{\e^{(\lambda a)X}}
            \leq \exp\rbra*{\frac{(\lambda a)^2\variance/2}{1-(\abs{\lambda a}\bounded )/3}}
            \leq \exp\rbra*{\frac{\lambda^2 (a^2\variance)/2}{1-(\abs{\lambda} (\abs{a}\bounded) )/3}}
        \end{align*}
    and we obtain the first claim.
    For the second claim, consider $\lambda\geq 0$ such that $\lambda \qty(a \bounded)<3$.
    Since $X$ satisfies one-sided $(\bounded,\variance)$-Bernstein condition, we have
        \begin{align*}
            \E\sbra*{\e^{\lambda (aX)}}
            =\E\sbra*{\e^{(\lambda a)X}}
            \leq \exp\rbra*{\frac{(\lambda a)^2\variance/2}{1-(\lambda a\bounded )/3}}
            = \exp\rbra*{\frac{\lambda^2 (a^2\variance)/2}{1-(\lambda a \bounded )/3}}
        \end{align*}
    and we obtain the second claim.
    \end{proof}
    \begin{proof}[Proof of \cref{item:dominated Bernstein condition}]
    Consider $\lambda \geq 0$ such that $\lambda \bounded <3$.
    Then, $f(x) = \e^{\lambda x}$ is non-decreasing and hence $\E[\e^{\lambda X}]\leq \E[\e^{\lambda Y}]$ holds for random variables such that $X\preceq Y$ (see \cref{lem:Stochastic domination}).
    Hence, 
        \begin{align*}
            \E\sbra*{\e^{\lambda Y}}
            \leq \E\sbra*{\e^{\lambda X}}
            \leq \exp\rbra*{\frac{\lambda^2\variance/2}{1-(\lambda \bounded )/3}}
        \end{align*}
        holds and we obtain the claim.

      In particular, if $Y\leq X$, then $Y\preceq X$ (see \cref{lem:Stochastic domination}), which proves the claim.
    \end{proof}
    \begin{proof}[Proof of \cref{item:BC for independent rvs}]
    Consider $\lambda \in \mathbb{R}$ such that $\abs{\lambda}\bounded <3$.
   Since $X_1,\ldots, X_n$ are independent, we obtain
    \begin{align*}
    \E\qty[\mathrm{e}^{\lambda X}]=\E\qty[\prod_{i\in [n]}\mathrm{e}^{\lambda X_i}]
    =\prod_{i\in [n]}\E\qty[\mathrm{e}^{\lambda X_i}]
    \leq \prod_{i\in [n]}\exp\qty(\frac{\lambda^2s_i/2}{1-(\abs{\lambda}\bounded/3)})
    =\exp\qty(\frac{\lambda^2\sum_{i\in[n]}s_i/2}{1-(\abs{\lambda}\bounded/3)}).
    \end{align*}
  \end{proof}  
  \begin{proof}[Proof of \cref{item:BC for NA rvs}]
    Consider $\lambda \geq 0$ such that $\lambda \bounded <3$.
    Then, $f(x)=\e^{\lambda x}$ is non-decreasing and hence
    $\E\qty[\prod_{i\in [n]}\mathrm{e}^{\lambda X_i}]
    \leq \prod_{i\in [n]}\E\qty[\mathrm{e}^{\lambda X_i}]$
    holds for negatively associated random variables $X_1,\ldots, X_n$ (\cref{lem:prodict of NA}).
    Hence, we obtain
    \begin{align*}
    \E\qty[\mathrm{e}^{\lambda X}]
    =\E\qty[\prod_{i\in [n]}\mathrm{e}^{\lambda X_i}] 
    \leq \prod_{i\in [n]}\E\qty[\mathrm{e}^{\lambda X_i}] 
    \leq \prod_{i\in [n]}\exp\qty(\frac{\lambda^2s_i/2}{1-(\lambda\bounded/3)})
    =\exp\qty(\frac{\lambda^2\sum_{i\in[n]}s_i/2}{1-(\lambda\bounded/3)}).
    \end{align*}
  \end{proof} 


  
 \subsection{Drift Analysis based on Bernstein Condition} \label{sec:drift analysis intro}
  In this paper, we use the drift analysis based on the Bernstein condition.
   Consider a sequence of random variables $(X_t)_{t\in \Nat_0}$ such that (i) $\E_{t-1}\qty[ X_t ] \le X_{t-1} - \drift$ whenever $X_t$ satisfies some “good" condition $\calE$, and (ii) the difference $X_t - X_{t-1}$ conditioned on the $(t-1)$-th configuration satisfies the Bernstein condition.
    One can expect that such $(X_t)_{t\in \Nat_0}$ behaves like $X_t \lesssim X_0 - \drift\cdot t$ while $X_t$ keeps satisfying $\calE$.
    In the following, we prove this intuition using Freedman's inequality combined with martingale techniques.

\begin{lemma}[Additive drift lemma] \label{lem:Freedman stopping time additive}
 Let $(X_t)_{t\in \Nat_0}$ be a sequence of random variables and
 let $(\calF_t)_{t\in\Nat_0}$ be a filtration such that $ X_t $ is $ \calF_t $-measurable for all $t\ge 0$.
 Let $\tau$ be a stopping time with respect to $(\calF_t)_{t\in \Nat_0}$.
 Let $\bounded,\variance \ge 0$ and $\drift\in\Real$ be parameters.
 Suppose the following condition holds for any $t\geq 1$: conditioned on $\mathcal{F}_{t-1}$, 
      \begin{enumerate}[label=$(C\arabic*)$]
        \item $\indicator_{\tau>t-1}\rbra*{\E_{t-1}[X_t]-X_{t-1}-\drift}\leq 0$, \label{item:C1}
        \item $\indicator_{\tau>t-1}\rbra*{X_t-X_{t-1}-\drift}$ satisfies one-sided $\qty(\bounded,\variance)$-Bernstein condition. \label{item:C2}
      \end{enumerate}
    For a parameter $h>0$, define stopping times
    \begin{align*} 
      &\tauxplus\defeq \inf\qty{t\geq 0\colon X_t\geq X_0+h}, \\
      &\tauxminus \defeq \inf\qty{t\geq 0\colon X_t\leq X_0-h}.
    \end{align*}
    Then, we have the following:
      \begin{enumerate}
        \renewcommand{\labelenumi}{(\roman{enumi})}
        \item \label{item:positive drift} Suppose $\drift\geq 0$. Then, for any $h,T>0$ such that $z \defeq h-\drift\cdot T > 0$, we have
        \[
        \Pr\qty[\tauxplus \le \min\{T,\tau\}] 
         \le \exp\qty(- \frac{z^2/2}{\variance T + (z \bounded)/3}).
        \]
        \item \label{item:negative drift} Suppose $\drift< 0$. Then, for any $h,T>0$ such that $ z \defeq (-\drift)\cdot T - h > 0$, we have
        \[
        \Pr\qty[\min\cbra{\tauxminus,\tau} >T] 
         \le \exp\qty(- \frac{z^2/2}{\variance T + (z \bounded)/3}).
        \]
      \end{enumerate}
  \end{lemma}
    \begin{remark}
      Since \cref{item:C1} implies $\indicator_{\tau>t-1}\rbra*{X_t-X_{t-1}-R}\leq \indicator_{\tau>t-1}\rbra*{X_t-\E_{t-1}[X_t]}$, we can use the following \cref{item:C2'} instead of \cref{item:C2}:
      \begin{enumerate}[label=$(C\arabic*')$]
        \setcounter{enumi}{1}
        \item $\indicator_{\tau>t-1}\rbra*{X_t-\E_{t-1}[X_t]}$ satisfies one-sided $\qty(\bounded,\variance)$-Bernstein condition. \label{item:C2'}
      \end{enumerate}
    \end{remark}

Very intuitively, \cref{lem:Freedman stopping time additive} implies tha following:
If $\E_{t-1}[X_t]\leq X_{t-1}+R$ for positive $R$, then the probability that $X_t$ exceeds $X_0+h$ within fewer than $h/R$ steps is exponentially small (\cref{item:positive drift}).
If $\E_{t-1}[X_t]\leq X_{t-1}-\bar{R}$ for positive $\bar{R}$, then the probability that $X_t$ has not reached $X_0-h$ after more than $h/\bar{R}$ steps is exponentially small (\cref{item:negative drift}).

The key component of proof of \cref{lem:Freedman stopping time additive} is Freedman's inequality.
Recall that a sequence of random variables $(X_t)_{t\ge 0}$ is a \emph{supermartingale} (resp.\ \emph{submartingale}) if $\E_{t-1}[X_t]\leq X_{t-1}$ (resp.\ $\E_{t-1}[X_t]\geq X_{t-1}$) holds for all $t\ge 1$.
  In this paper, we use the following general version.
  \begin{theorem}[Theorem 2.6 of \cite{FGL15}]
    \label{thm:general Freedman}
        Let $(X_t)_{t\in \Nat_0}$ be a real-valued supermartingale associated with the natural filtration $(\mathcal{F}_t)_{t\in \Nat_0}$.
        Assume that $V_{t-1}$, $t\in [T]$ are positive and $\mathcal{F}_{t-1}$-measurable random variables.
        Suppose 
        $
        \E_{t-1}\sbra*{\exp\rbra*{\lambda (X_t-X_{t-1})}}\leq \exp\rbra*{f(\lambda)V_{t-1}}
        $
        for all $t\in [T]$ and for a positive function $f(\lambda)$ for some $\lambda\in (0,\infty)$.
        Then, for all $h, W>0$,
        \begin{align*}
            \Pr\sbra*{\tinyexists t \le T,\,X_t-X_0\geq h \textrm{ and }\sum_{i=1}^tV_{i-1}\leq W} \leq \exp\rbra*{-\lambda h+f(\lambda)W }.
        \end{align*}
  \end{theorem}
  If a supermartingale $(X_t)_{t\in \Nat_0}$ satisfies one-sided Bernstein condition, then we can obtain the following concentration inequality from \cref{thm:general Freedman}.
    \begin{corollary}[Freedman-type inequality under one-sided Bernstein condition] 
      \label{cor:Freedman}
        Let $(X_t)_{t\in \Nat_0}$ be a supermartingale associated with the natural filtration $(\calF_t)_{t\in \Nat_0}$.
        Suppose that, for every $t\ge 1$, the difference $X_t - X_{t-1}$ conditioned on $\calF_{t-1}$ satisfies one-sided $(\bounded,\variance)$-Bernstein condition.
        Then, for any $h > 0$, we have
      \begin{align*}
         \Pr\sbra*{{}^{\exists} t\le T,\,X_t-X_0\geq h} \leq \exp\rbra*{-\frac{h^2/2}{T\variance + (h\bounded )/3}}.
      \end{align*}
    \end{corollary}
    \begin{proof}
      We apply \cref{thm:general Freedman} for $\lambda=\frac{h}{T\variance+(h\bounded )/3} > 0$, 
      $f(\lambda)=\frac{\lambda^2/2}{1-(\lambda\bounded /3)}$, $W=T \variance$, and $V_t = s$ (for all $t$).
      Due to the one-sided Bernstein condition of $Y_t - Y_{t-1}$, we have $\E_{t-1}\qty[ \exp\qty( \lambda (Y_t - Y_{t-1}) ) ] \le \exp\qty( f(\lambda) V_{t-1} )$.
      Note that $\lambda\bounded =\frac{h \bounded}{T\variance+(h \bounded)/3}<3$.
      Therefore, from \cref{thm:general Freedman}, we obtain
      \begin{align*}
        \Pr\sbra*{{}^{\exists}t\le T,\,Y_t-Y_0\geq h}
        &=\Pr\sbra*{{}^{\exists}t\le T,\,Y_t-Y_0\geq h \text{ and }\sum_{i=1}^t V_{i-1}\leq \variance T}\\
        &\leq \exp\rbra*{-\lambda h+\frac{\lambda^2/2}{1-(\lambda\bounded/3)}\cdot T\variance }\\
        &=\exp\rbra*{-\frac{h^2/2}{T\variance+(h\bounded)/3}}.
      \end{align*}
    \end{proof}
  \begin{proof}[Proof of \cref{lem:Freedman stopping time additive}]
    For the parameter $z>0$ (either $ z=h-\drift\cdot T $ or $ z = (-\drift)\cdot T - h $), consider the following stopping time:
    \begin{align*} 
       \taux\defeq \inf\cbra*{t\geq 0\colon X_t\geq X_0+\drift\cdot t+z}.
    \end{align*}
      Let $Y_t=X_t-\drift\cdot t$ and $Z_t=Y_{t\wedge \tau}$.
      Then, $Z_t-Z_{t-1}$ conditioned on $\mathcal{F}_{t-1}$ satisfies one-sided $\qty(\bounded,\variance)$-Bernstein condition and
      $(Z_t)_{t\in \Nat_0}$ is a supermartingale since
      \begin{align*}
      Z_t-Z_{t-1}
      =\indicator_{\tau>t-1}\rbra*{Y_t-Y_{t-1}}
      =\indicator_{\tau>t-1}\rbra*{X_t-X_{t-1}-\drift}
      \end{align*}
      and
      \begin{align*}
        \E_{t-1}\sbra*{Z_t-Z_{t-1}}
        =\indicator_{\tau>t-1}\rbra*{\E_{t-1}[X_t]-X_{t-1}-\drift}
        \leq 0.
      \end{align*}
      Then, 
      we obtain
      \begin{align*}
        \Pr\qty[\tau_X \le \min\{T,\tau\}] 
        &=\Pr\qty[\tinyexists t\le T\land \tau,\,X_t\geq X_0+\drift\cdot t+z ] \\
        &=\Pr\qty[\tinyexists t \le T \land \tau,\,Y_t\geq Y_0+z] &&(\because Y_t=X_t-Rt)\\
        &=\Pr\qty[\tinyexists t \le T\land \tau,\,Z_t\geq Z_0+z] &&(\because t\leq T\wedge \tau \leq \tau)\\
        &\le \Pr\qty[\tinyexists t \le T,\,Z_t\geq Z_0+z]\\
        &\le \exp\qty(- \frac{z^2/2}{\variance T + (z \bounded)/3}) &&(\because \textrm{\cref{cor:Freedman}}).
      \end{align*}
      For the first claim, it suffices to show that $\Pr[\tauxplus \le \min\{T,\tau\}] 
      \leq \Pr\qty[\tau_X \le \min\{T,\tau\}]$.
      Since $\bounded\geq 0$ and $\drift T\leq h-z$, we have
      \begin{align*}
        \Pr\qty[\tauxplus \le \min\{T,\tau\}] 
        &=\Pr\qty[\tinyexists t \le T,\,X_t\geq X_0+h]\\
        &\leq \Pr\qty[\tinyexists t \le T\land \tau,\,X_t\geq X_0+\drift\cdot t+z]&&(\because h\geq \drift\cdot T+z\geq \drift\cdot t+z)\\
        &=\Pr\qty[\tau_X \le \min\{T,\tau\}].
      \end{align*}
    
    Now, we apply \cref{thm:general Freedman} for $\lambda=\frac{h}{T\variance+(h\bounded )/3} > 0$, 
      $f(\lambda)=\frac{\lambda^2/2}{1-(\lambda\bounded /3)}$, $W=T \variance$, and $V_t = s$ (for all $t$).
      Due to the one-sided Bernstein condition of $Y_t - Y_{t-1}$, we have $\E_{t-1}\qty[ \exp\qty( \lambda (Y_t - Y_{t-1}) ) ] \le \exp\qty( f(\lambda) V_{t-1} )$.
      Note that $\lambda\bounded =\frac{h \bounded}{T\variance+(h \bounded)/3}<3$.
      Therefore, from \cref{thm:general Freedman}, we obtain
      \begin{align*}
        \Pr\sbra*{{}^{\exists}t\le T,\,Y_t-Y_0\geq h}
        &=\Pr\sbra*{{}^{\exists}t\le T,\,Y_t-Y_0\geq h \text{ and }\sum_{i=1}^t V_{i-1}\leq \variance T}\\
        &\leq \exp\rbra*{-\lambda h+\frac{\lambda^2/2}{1-(\lambda\bounded/3)}\cdot T\variance }\\
        &=\exp\rbra*{-\frac{h^2/2}{T\variance+(h\bounded)/3}}.
      \end{align*}
      
      For the second claim, it suffices to show that $\Pr[\min\{\tauxminus,\tau\} >T] 
        \leq \Pr\qty[\tau_X \le \min\{T,\tau\}]$.
      Since $\drift < 0$ (i.e., $-\drift>0$) and $-\drift T\geq h+z$, we have
      \begin{align*}
        \Pr\qty[\min\cbra{\tauxminus,\tau} >T]
        &=\Pr\qty[X_T> X_0-h\textrm{ and }\tauxminus >T\textrm{ and }\tau>T] \\
        &\leq \Pr\qty[X_T> X_0+\drift\cdot T+z\textrm{ and }\tau>T] &&(\because -h\geq \drift\cdot T+z)\\
        &=\Pr\qty[X_{T\wedge \tau}> X_0+\drift\cdot (T\wedge \tau)+z\textrm{ and }\tau>T] \\
        &\leq \Pr\qty[X_{T\wedge \tau}\geq  X_0+\drift\cdot (T\wedge \tau)+z] \\
        &\leq \Pr\qty[\tinyexists t \le T \land \tau,\,X_t\geq X_0+\drift\cdot t+z ]\\
        &=\Pr\qty[\tau_X \le \min\{T,\tau\}].
      \end{align*}
    \end{proof}

Finally, we introduce the following simple lemma, which follows immediately from \cref{lem:Bernstein condition} and is frequently used in our proof.
   \begin{lemma}
      \label{lem:BC for SP}
      Let $(X_t)_{t\in \Nat_0}$ be a sequence of random variables and
      let $(\calF_t)_{t\in\Nat_0}$ be a filtration such that $ X_t $ is $ \calF_t $-measurable for all $t\ge 0$.
      Let $(\bounded_{t})_{t\in \Nat_0}$ and $(\variance_{t})_{t\in \Nat_0}$ are $ \calF_t $-measurable random variables.
      Let $\tau$ be a stopping time with respect to $(\calF_t)_{t\in \Nat_0}$.
      Suppose that $X_t$ conditioned on $\mathcal{F}_{t-1}$ satisfies $(\bounded_{t-1},\variance_{t-1})$-Bernstein condition, $\indicator_{\tau>t-1} \bounded_{t-1}\leq D$, and $\indicator_{\tau>t-1} \variance_{t-1}\leq \variance$ for some non-negative parameters $\bounded$ and $\variance$.
      Then, 
      both $\indicator_{\tau>t-1}X_t$ and $-\indicator_{\tau>t-1}X_t$ conditioned on $\mathcal{F}_{t-1}$ satisfy $(\bounded,\variance)$-Bernstein condition.
    \end{lemma}
    \begin{proof}
      From \cref{item:Bernstein condition 2} of \cref{lem:Bernstein condition}, it suffices to show that $\indicator_{\tau>t-1}X_t$ conditioned on $\mathcal{F}_{t-1}$ satisfies $(\bounded,\variance)$-Bernstein condition.
      First, from \cref{item:Bernstein condition 2} of \cref{lem:Bernstein condition} and our assumption, $\indicator_{\tau>t-1}X_t$ satisfies $(\indicator_{\tau>t-1}\bounded_{t-1},\indicator_{\tau>t-1}^2\variance_{t-1})$-Bernstein condition.
      Then, from $\indicator_{\tau>t-1}\bounded_{t-1}\leq \bounded$, $\indicator_{\tau>t-1}^2 \variance_{t-1}=\indicator_{\tau>t-1} \variance_{t-1}\leq \variance$, and  \cref{item:Bernstein condition 1} of \cref{lem:Bernstein condition}, 
      $\indicator_{\tau>t-1}X_t$ conditioned on $\mathcal{F}_{t-1}$ satisfies $(\bounded,\variance)$-Bernstein condition.
      Thus, we obtain the claim.
    \end{proof}

\section{Drift Analysis for 3-Majority and 2-Choices} \label{sec:drift analysis}
In this section, we check that for several quantities (e.g., $\alpha_t(i)$, $\alphanorm_t$, $\delta_t$) of 3-Majority or 2-Choices, their one-step difference satisfies the Bernstein condition (\cref{sec:bernstein condition for sync processes}).
This enables us to apply our drift analysis (\cref{lem:Freedman stopping time additive}) to these quantities (\cref{sec:drift analysus for sync processes}).
\subsection{Bernstein Condition for 3-Majority and 2-Choices} \label{sec:bernstein condition for sync processes}
First, we present bounds of expectations and variances of quantities of interest for each model.
Some of them are already known in the literature, but we provide a unified proof for all quantities in \cref{sec:proof of basic inequalities}.
\begin{lemma}[Basic inequalities for $\alpha_t, \delta_t$ and $\alphanorm_t$]\label{lem:basic inequality}
  Consider the quantities defined in \cref{def:basic quantities} for 3-Majority or 2-Choices.
  Then, we have the following for any $t\geq 1$:
\begin{enumerate}
            \renewcommand{\labelenumi}{(\roman{enumi})}
\item \label{item:alpha} 
For any opinion $i \in [k]$, we have
  \begin{align*} 
     &\E_{t-1} [\alpha_t(i)] = 
      \alpha_{t-1}(i) \qty( 1 + \alpha_{t-1}(i) - \alphanorm_{t-1} ), \\ 
    &\Var_{t-1}[\alpha_t(i)] \le \begin{cases}
      \frac{\alpha_{t-1}(i)}{n}	& \text{for 3-Majority},\\
      \frac{\alpha_{t-1}(i)(\alpha_{t-1}(i) + \alphanorm_{t-1})}{n}	& \text{for 2-Choices}.
    \end{cases} 
  \end{align*}
\item \label{item:delta} 
For any two distinct opinions $i,j\in [k]$, we have
  \begin{align*} 
     &\E_{t-1} [\delta_t(i,j)] =
      \delta_{t-1}(i,j) \qty( 1 + \alpha_{t-1}(i) + \alpha_{t-1}(j) - \alphanorm_{t-1} ), \\
    &\Var_{t-1}[\delta_t(i,j)] \le \begin{cases}
      \frac{2}{n}(\alpha_{t-1}(i) + \alpha_{t-1}(j))	& \text{for 3-Majority},\\[10pt]
      \frac{1}{n}(\alpha_{t-1}(i)+\alpha_{t-1}(j))(\alpha_{t-1}(i)+\alpha_{t-1}(j)+\alphanorm_{t-1})	& \text{for 2-Choices}.
    \end{cases} 
  \end{align*}
\item \label{item:expectation of alphanorm} 
It holds that
\begin{align*} 
\E_{t-1}[\alphanorm_{t}] \ge \begin{cases}
   \alphanorm_{t-1}+\frac{1-\alphanorm_{t-1}}{n}	& \text{for 3-Majority},\\[10pt]
   \alphanorm_{t-1}+\frac{(1-\sqrt{\alphanorm_{t-1}})(1-\alphanorm_{t-1})\alphanorm_{t-1}}{n}	& \text{for 2-Choices}.
 \end{cases} 
\end{align*}
In particular, $\E_{t-1}[\alphanorm_{t}] \geq \alphanorm_{t-1}$.
\end{enumerate}
\end{lemma}

Now, we introduce the Bernstein condition for the quantities $\alpha_t$, $\delta_t$ and $\alphanorm_t$.
Essentially, we apply the Bernstein condition for the sum of independent random variables (\cref{item:BC for independent rvs} of \cref{lem:Bernstein condition}).
The most technical part of our analysis is the study of $\alphanorm_t$, for which we additionally use the Bernstein condition for negatively associated random variables (see \cref{item:BC for NA rvs} in \cref{lem:Bernstein condition}).
\begin{lemma}[Bernstein condition for $\alpha_t$, $\delta_t$ and $\alphanorm_t$]
\label{lem:Bernstein condition for sync processes}
Consider the quantities defined in \cref{def:basic quantities} for 3-Majority or 2-Choices.
Then, we have the following for any $t\geq 1$:
\begin{enumerate}
            \renewcommand{\labelenumi}{(\roman{enumi})}
    \item \label{item:BC for alpha}
    For any opinion $i\in [k]$, 
    $\alpha_{t}(i)-\E_{t-1}\sbra*{\alpha_t(i)}$ conditioned on round $t-1$ satisfies $\qty(\frac{1}{n},\variance)$-Bernstein condition, where
    \begin{align*} 
       \variance = \begin{cases}
        \frac{\alpha_{t-1}(i)}{n}	& \text{for 3-Majority},\\
        \frac{\alpha_{t-1}(i)(\alpha_{t-1}(i)+\alphanorm_{t-1})}{n} & \text{for 2-Choices}.
       \end{cases}
    \end{align*}
    \item \label{item:BC for delta}
     For any two distinct opinions $i,j\in [k]$,
     $\delta_{t}(i,j)-\E_{t-1}\sbra*{\delta_t(i,j)}$ conditioned on round $t-1$ satisfies $ \qty(\frac{2}{n}, \variance) $-Bernstein condition, 
     where
    \begin{align*} 
       \variance = \begin{cases}
          \frac{2}{n}(\alpha_{t-1}(i)+\alpha_{t-1}(j))	& \text{for 3-Majority},\\[10pt]
          \frac{1}{n}\qty(\alpha_{t-1}(i) + \alpha_{t-1}(j))\qty(\alpha_{t-1}(i)+\alpha_{t-1}(j)+\alphanorm_{t-1}) & \text{for 2-Choices}.
       \end{cases} 
    \end{align*}
    \item \label{item:BC for alphanorm}
     $\alphanorm_{t-1}-\alphanorm_t$ conditioned on round $t-1$ satisfies one-sided $ \qty(\frac{2\sqrt{\alphanorm_{t-1}}}{n},\variance)$-Bernstein condition, where
    \begin{align*}
       \variance = \begin{cases}
          \frac{4}{n}\alphanorm_{t-1}^{1.5}	& \text{for 3-Majority},\\[10pt]
          \frac{8}{n}\alphanorm_{t-1}^2 & \text{for 2-Choices}.
       \end{cases}
    \end{align*} 
\end{enumerate}
\end{lemma}

\begin{proof}[Proof of \cref{item:BC for alpha}]
    We show that the random variable $\alpha_{t}(i)-\E_{t-1}\qty[\alpha_{t}(i)]$ conditioned on round $t-1$ satisfies $\qty(\frac{1}{n},\Var_{t-1}\qty[\alpha_{t}(i)])$-Bernstein condition. 
    Once this is established, the claim follows from \cref{item:alpha} of \cref{lem:basic inequality} and \cref{item:Bernstein condition 1} of \cref{lem:Bernstein condition}.
    
      From definition, $n\alpha_{t}(i)=\sum_{v\in V}\indicator_{\opinion_t(v)=i}$.
      Hence, $\alpha_{t}(i)-\E_{t-1}\qty[\alpha_{t}(i)]=\sum_{v\in V}X_t(v)$ where $X_t(v)\defeq \frac{\indicator_{\opinion_t(v)=i}-\E_{t-1}\qty[\indicator_{\opinion_t(v)=i}]}{n}$.
      Since $\abs{X_t(v)}\leq 1/n$ for all $v\in V$, $X_t(v)$ conditioned on round $t-1$ satisfies $\qty(\frac{1}{n},\Var_{t-1}[X_t(v)])$-Bernstein condition (\cref{item:Bernstein condition for bounded random variables} of \cref{lem:Bernstein condition}).
      Furthermore, since $(X_t(v))_{v\in V}$ conditioned on round $t-1$ are $n$ mean-zero independent random variables, 
      $\alpha_{t}(i)-\E_{t-1}\qty[\alpha_{t}(i)]=\sum_{v\in V}X_t(v)$ satisfies $\qty(\frac{1}{n},\sum_{v\in V}\Var_{t-1}[X_v])$-Bernstein condition from \cref{item:BC for independent rvs} of \cref{lem:Bernstein condition}.
      Since 
      \[
      \sum_{v\in V}\Var_{t-1}[X_t(v)]=\Var_{t-1}\qty[\sum_{v\in V}X_t(v)]=\Var_{t-1}\qty[\alpha_{t}(i)-\E_{t-1}\qty[\alpha_{t}(i)]]=\Var_{t-1}\qty[\alpha_{t}(i)],
      \]
      we obtain the claim.
  \end{proof}

  \begin{proof}[Proof of \cref{item:BC for delta}]
  We show that $\delta_{t}-\E_{t-1}\qty[\delta_{t}]$ conditioned on round $t-1$ satisfies $\qty(\frac{1}{n},\Var_{t-1}\qty[\delta_t])$-Bernstein condition. 
    Once this is established, the claim follows from \cref{item:alpha} of \cref{lem:basic inequality} and \cref{item:Bernstein condition 1} of \cref{lem:Bernstein condition}.

      By definition, $n\delta_{t}=n\qty(\alpha_t(i)-\alpha_t(j))=\sum_{v\in V}\qty(\indicator_{\opinion_t(v)=i}-\indicator_{\opinion_t(v)=j})$.
      Let
      \begin{align*}
        X_t(v)\defeq \frac{1}{n}\qty(\indicator_{\opinion_t(v)=i}-\indicator_{\opinion_t(v)=j}-\E_{t-1}\qty[\indicator_{\opinion_t(v)=i}-\indicator_{\opinion_t(v)=j}]).
      \end{align*}
      Then, we have $\delta_{t}-\E_{t-1}\qty[\delta_{t}]=\sum_{v\in V}X_t(v)$.
      Since $\abs{X_t(v)}\leq \frac{2}{n}$ for all $v\in V$, $X_t(v)$ conditioned on round $t-1$ satisfies $\qty(\frac{2}{n},\Var_{t-1}[X_t(v)])$-Bernstein condition (\cref{item:Bernstein condition for bounded random variables} of \cref{lem:Bernstein condition}).
      Furthermore, since $(X_t(v))_{v\in V}$ conditioned on round $t-1$ are $n$ mean-zero independent random variables, 
      $\delta_{t}-\E_{t-1}\qty[\delta_{t}]=\sum_{v\in V}X_t(v)$ satisfies $\qty(\frac{1}{n},\sum_{v\in V}\Var_{t-1}[X_t(v)])$-Bernstein condition from \cref{item:BC for independent rvs} of \cref{lem:Bernstein condition}.
      We obtain the claim since 
      \[
      \sum_{v\in V}\Var_{t-1}[X_t(v)]=\Var_{t-1}\qty[\sum_{v\in V}X_t(v)]=\Var_{t-1}\qty[\delta_{t}-\E_{t-1}\qty[\delta_{t}]]=\Var_{t-1}\qty[\delta_{t}].
      \]
  \end{proof}

  \begin{proof}[Proof of \cref{item:BC for alphanorm}]
    From the Cauchy-Schwartz inequality, we have
    $\alphanorm_{t-1}^2=\left(\sum_{i\in [k]}\alpha_{t-1}(i)^2\right)^2\leq \sum_{i\in [k]}(\alpha_{t-1}(i)^{1/2})^2(\alpha_{t-1}(i)^{3/2})^2=\|\alpha_{t-1}\|_3^3$.
    Hence, 
    \begin{align}
        \alphanorm_{t-1}
        &\leq \alphanorm_{t-1}+\|\alpha_{t-1}\|_3^3-\alphanorm_{t-1}^2 & & (\text{$\because$ $\norm{\alpha_{t-1}}_3^3\geq \alphanorm_{t-1}^2$}) \nonumber \\
        &=\sum_{i\in [k]}\alpha_{t-1}(i)^2\qty(1+\alpha_{t-1}(i)+\alphanorm_{t-1}) \nonumber \\
        &=\sum_{i\in [k]}\alpha_{t-1}(i)\E_{t-1}[\alpha_t(i)] & & (\text{$\because$ \cref{item:alpha} of \cref{lem:basic inequality}})
        \label{eq:alphanorm_key_square 1}
    \end{align}
    holds.
    Then, we have
    \begin{align}
      \alphanorm_{t-1}-\alphanorm_t
        &=\sum_{i\in [k]}\rbra*{\alpha_{t-1}(i)^2-\alpha_{t}(i)^2} \nonumber\\
        &\leq \sum_{i\in [k]}2\alpha_{t-1}(i)\rbra*{\alpha_{t-1}(i)-\alpha_{t}(i)} & & \text{$\because$ $\tinyforall x,y\in\Real, x^2-y^2\le 2x(x-y)$}\nonumber\\
        &\leq \sum_{i\in [k]}2\alpha_{t-1}(i)\rbra*{\E_{t-1}\sbra*{\alpha_{t}(i)}-\alpha_{t}(i)} & & \text{$\because$ \cref{eq:alphanorm_key_square 1}}\nonumber\\
        &=\sum_{i\in [k]}Y_t(i), \nonumber 
    \end{align}
    where \[
    Y_t(i)\defeq 2\alpha_{t-1}(i)\rbra*{\E_{t-1}\sbra*{\alpha_{t}(i)}-\alpha_{t}(i)}=\sum_{v\in V}\frac{2\alpha_{t-1}(i)}{n}\qty(\E_{t-1}[\indicator_{\opinion_t(v)=i}]-\indicator_{\opinion_t(v)=i}).
    \]
    $Y_t(i)$ conditioned on round $t-1$ satisfies $\qty(\frac{2\alpha_{t-1}(i)}{n},4\alpha_{t-1}(i)^2\Var_{t-1}[\alpha_t(i)])$-Bernstein condition from \cref{item:Bernstein condition 2} of \cref{lem:Bernstein condition} and \cref{item:BC for alpha} of \cref{lem:Bernstein condition for sync processes}.
    Furthermore, from $\alpha_{t-1}(i)^2\leq \alphanorm_{t-1}$, \cref{item:Bernstein condition 1} of \cref{lem:Bernstein condition} implies that 
    $Y_t(i)$ conditioned on round $t-1$ satisfies $\qty(\frac{2\sqrt{\alphanorm_{t-1}}}{n},4\alpha_{t-1}(i)^2\Var_{t-1}[\alpha_t(i)])$-Bernstein condition.
    
    From \cref{lem:zero one lemma}, the random variables $(\indicator_{\opinion_t(v)=i})_{i\in [k]}$ are negatively associated for each $v\in V$.
    From \cref{lem:concatination of NA}, $\qty((\indicator_{\opinion_t(v)=i})_{i\in [k]})_{v\in V}$, a sequence of $kn$ random variables, are also negatively associated.
    Since $Y_t(i)=h_i\qty((\indicator_{\opinion_t(v)=i})_{v\in V})$, i.e., non-increasing functions of disjoint subsets of negatively associated random variables $\qty((\indicator_{\opinion_t(v)=i})_{i\in [k]})_{v\in V}$, $(Y_t(i))_{i\in [k]}$ are negatively associated (\cref{lem:concatination of NA}).    
    Thus, from \cref{item:BC for NA rvs} of \cref{lem:Bernstein condition}, $\sum_{i\in [k]}Y_t(i)$ conditioned on round $t-1$ satisfies 
    one-sided $\qty(\frac{2\sqrt{\alphanorm_{t-1}}}{n}, 4\sum_{i\in [k]}\alpha_{t-1}(i)^2\Var_{t-1}[\alpha_t(i)])$-Bernstein condition.
    From \cref{item:dominated Bernstein condition} of \cref{lem:Bernstein condition}, $\alphanorm_{t-1}-\alphanorm_t\leq \sum_{i\in [k]}Y_t(i)$ conditioned round $t-1$ also satisfies 
    one-sided $\qty(\frac{2\sqrt{\alphanorm_{t-1}}}{n},4\sum_{i\in [k]}\alpha_{t-1}(i)^2\Var_{t-1}[\alpha_t(i)])$-Bernstein condition.
    
    For specific bounds of $\Var_{t-1}\sbra*{\alpha_{t}(i)}$, we can apply \cref{item:alpha} of \cref{lem:Bernstein condition}: 
    $
      \sum_{i\in [k]}\Var_{t-1}\sbra*{\alpha_{t}(i)}\leq \frac{\|\alpha_{t-1}\|_3^3}{n}\leq \frac{\alphanorm_{t-1}^{1.5}}{n}
    $
    for 3-Majority and 
    $\Var_{t-1}\sbra*{\alpha_{t}(i)}\leq \frac{\|\alpha_{t-1}\|_4^4 + \|\alpha_{t-1}\|_3^3 \alphanorm_{t-1}}{n}
    \leq \frac{2\alphanorm_{t-1}^2}{n}$ for 2-Choices.
    Applying \cref{item:Bernstein condition 1} of \cref{lem:basic inequality}, we obtain the claim.
  \end{proof}

  Finally, we introduce the following lemma that provides a key bound in 2-Choices.
\begin{lemma}[Bernstein condition for $\alpha_t$: A special case of 2-Choices]
\label{lem:concentration of alpha difference for 2 choices}
    Consider the 2-Choices.
    Then, for any opinion $i\in [k]$ and $t\ge 1$, 
    $\alpha_t(i) - \alpha_{t-1}(i)$ conditioned on round $t-1$ satisfies one-sided $\qty(\frac{1}{n},\frac{2\alpha_{t-1}(i)^2}{n})$-Bernstein condition if 
    $\alpha_{t-1}(i)\leq \alphanorm_{t-1}$.
  \end{lemma}
  \begin{proof}
    Conditioned on the $(t-1)$-th round, the difference $\alpha_t(i) - \alpha_{t-1}(i)$ can be written as
    \[
      \alpha_t(i) - \alpha_{t-1}(i) = \alphain_t(i) - \alphaout_t(i),
    \]
    where $\alphaout_t(i) = \frac{1}{n}\cdot |\cbra{v \in V \colon \opinion_t(v) \neq i \text{ and } \opinion_{t-1}(v) = i}|$, and
     $\alphain_t(i) = \frac{1}{n}\cdot |\cbra{v \in V \colon \opinion_t(v) = i \text{ and }\opinion_{t-1}(v)\ne i}|$.
    By the update rule of 2-Choices, conditioned on the $(t-1)$-th round, the distributions of $n\cdot \alphain_t(i)$ and $n\cdot \alphaout_t(i)$ are
    \begin{align*}
      &n\cdot \alphain_t(i) \sim \Bin(n(1-\alpha_{t-1}(i)),\alpha_{t-1}(i)^2),\\
      &n\cdot \alphaout_t(i) \sim \Bin(n\alpha_{t-1}(i), \alphanorm_{t-1} - \alpha_{t-1}(i)^2),
    \end{align*}
    where $\Bin(m,p)$ denotes the binomial distribution with parameters $m\in\Nat$ and $p\in[0,1]$.
    Note that these random variables are independent conditioned on round $t-1$.
    Let $X$ and $Y$ be binomial random variables, where 
    \begin{align*}
        &X\sim\Bin(n(1-\alpha_{t-1}(i)), \alpha_{t-1}(i)^2),\\
        &Y\sim\Bin(n\alpha_{t-1}(i),\alpha_{t-1}(i)-\alpha_{t-1}(i)^2).
    \end{align*}
    Note that we can write 
    $X=\sum_{j=1}^{n(1-\alpha_{t-1}(i))}X_j$ and 
    $Y=\sum_{j=1}^{n\alpha_{t-1}(i)}Y_j$ for independent Bernoulli random variables $(X_j)_{j\in [n(1-\alpha_{t-1}(i))]}$ and $(Y_j)_{j\in [n\alpha_{t-1}(i)]}$.
    From the assumption of $\alpha_{t-1}(i)\leq \alphanorm_{t-1}$ and \cref{lem:binomial dominance}, we have $Y\preceq n\cdot \alphaout_t(i)$.
    Hence, 
    \begin{align*}
        \alpha_t(i) - \alpha_{t-1}(i) 
        &=\frac{1}{n}\qty( n\cdot \alphain_t(i)-n\cdot \alphaout_t(i))\\
        &\preceq \frac{1}{n}(X-Y) && (\because Y\preceq n\cdot \alphaout_t(i))\\
        &=\frac{1}{n}\qty(X-\E[X]+\E[Y]-Y) && (\because \E[X]=\E[Y])\\
        &=\sum_{j=1}^{n(1-\alpha_{t-1}(i))}\frac{X_j-\E[X_j]}{n}+\sum_{j=1}^{n\alpha_{t-1}(i)}\frac{\E[Y_j]-Y_j}{n}.
    \end{align*}
    From \cref{item:dominated Bernstein condition,item:Bernstein condition for bounded random variables,item:BC for independent rvs} of \cref{lem:Bernstein condition}, $\alpha_t(i) - \alpha_{t-1}(i)$ conditioned on round $t-1$ satisfies one-sided $\qty(\frac{1}{n},s)$-Bernstein condition for 
    \begin{align*}
    s&=\sum_{j=1}^{n(1-\alpha_{t-1}(i))}\Var\qty[\frac{X_j-\E[X_j]}{n}]+\sum_{j=1}^{n\alpha_{t-1}(i)}\Var\qty[\frac{\E[Y_j]-Y_j}{n}]\\
    &=\frac{\Var\qty[X-\E[X]]}{n^2}+\frac{\Var\qty[\E[Y]-Y]}{n^2}\\
    &=\frac{\Var\qty[X]+\Var\qty[Y]}{n^2}.
    \end{align*}
    Note that  both $\frac{X_j-\E[X_j]}{n}$ and $\frac{\E[Y_j]-Y_j}{n}$ are mean-zero and bounded by $1/n$. 
    Thus, from \cref{item:Bernstein condition 1} of \cref{lem:Bernstein condition} with $\frac{\Var\qty[X]+\Var\qty[Y]}{n^2}\leq \frac{2\alpha_{t-1}(i)^2}{n}$, we obtain the claim.
  \end{proof}

\subsection{Drift Analysis for Basic Quantities} \label{sec:drift analysus for sync processes}
Now, we introduce the drift analysis for the quantities $\alpha_t$, $\delta_t$ and $\alphanorm_t$ (\cref{lem:drift analysis for basic}).
We present an overview of the drift terms employed in \cref{lem:drift analysis for basic} in \cref{table:drift}.
To begin, we summarize the stopping times that we focus on. 

\begin{definition}[Stopping times for basic quantities] \label{def:stopping times}
  Consider the quantities defined in \cref{def:basic quantities} for 3-Majority or 2-Choices.
  Fix two distinct opinions $i,j\in [k]$.
\begin{enumerate}
  \renewcommand{\labelenumi}{(\roman{enumi})}
  \item  For constants $\ciup,\cidown>0$, define
    \begin{align*}
        &\tauiup = \inf\cbra*{ t \ge 0 \colon \alpha_t(i) \ge (1+\ciup)\alpha_0(i)},\\
        &\tauidown = \inf\cbra*{ t \ge 0 \colon \alpha_t(i) \le (1-\cidown)\alpha_0(i)}.
    \end{align*}

  \item For constants $ \cdeltaup,\cdeltadown >0$ and a parameter $ \xdelta = \xdelta(n) \in [1/n, 1]$, define
    \begin{align*}
        &\taudeltaup = \inf\cbra*{ t \ge 0 \colon \delta_t(i,j) \ge (1+\cdeltaup) \delta_0(i,j) },\\
        &\taudeltadown = \inf\cbra*{ t \ge 0 \colon \delta_t(i,j) \le (1-\cdeltadown) \delta_0(i,j) }, \\
        &\taudeltaplus= \inf\qty{t\geq 0: \abs{\delta_t(i,j)}\geq \xdelta}.
    \end{align*}

  \item For constants $\cnormup,\cnormdown>0$ and a parameter $\xnorm = \xnorm(n)\in[0,1]$, define
    \begin{align*}
        &\taunormup = \inf\{t\geq 0: \alphanorm_t\geq (1+\cnormup)\alphanorm_0\}, \\
        &\taunormdown = \inf\cbra*{t\geq 0: \alphanorm_t\leq (1-\cnormdown)\alphanorm_0}, \\
        &\taunormplus \defeq \inf\qty{t\geq 0: \alphanorm_t\geq \xnorm}. 
    \end{align*}

  \item For a constant $0\le \cweak<1/2$, we say that an opinion $i\in[k]$ is \emph{weak at round $t$} if $\alpha_t(i) \le (1-\cweak) \alphanorm_t$.
    We define
    \begin{align*} 
        &\tauiweak = \inf\cbra*{ t \ge 0 \colon \alpha_t(i) \le (1-\cweak) \alphanorm_t }.
    \end{align*}
  \item  For a constant $\cactive>0$ satisfying $\cnormdown<\cactive<\cweak$, we say that an opinion $ i \in [k]$ is \emph{active at round $ t $} if $ \alpha_t(i) \ge (1-\cactive)\cdot \alphanorm_0$.
  We define
  \begin{align*} 
    \tauiactive = \inf\qty{ t \ge 0 \colon \alpha_t(i) \ge (1-\cactive)\cdot \alphanorm_0 } .
  \end{align*}
\end{enumerate}

  The constants $\ciup,\cidown,\cdeltaup,\cdeltadown,\cnormup,\cnormdown,\cweak,\cactive$ are universal constants, e.g., 
  we can set $\calphaup=\calphadown=\cweak=1/10$, $\cdeltaup=\cdeltadown=\cactive=1/20$, and $\cnormup=\cnormdown=1/30$ for both 3-Majority and 2-Choices.
\end{definition}

\begin{table}[t]
  \centering
  \begin{tabular}{|c|c|}
  \hline
  Drift & Condition of $t$  \\ \hline \hline
  $\mathbb{E}_{t-1}[\alpha_t(i)-\alpha_{t-1}(i)]\leq C\alpha_0(i)^2$ & $t-1<\tauiup$  \\ \hline 
  $\mathbb{E}_{t-1}[\alpha_t(i)-\alpha_{t-1}(i)]\geq - C\alpha_0(i)^2$ & $t-1<\min\{\tauiweak,\tauiup\}$   \\ \hline
  $\mathbb{E}_{t-1}[\alpha_t(i)-\alpha_{t-1}(i)]\leq 0$ &$t-1<\min\{\tauiactive,\taunormdown\}$   \\ \hline \hline
  $\mathbb{E}_{t-1}[\delta_t(i,j)-\delta_{t-1}(i,j)]\geq 0$ & $t-1<\min\{\taujweak,\taudeltadown\}$   \\ \hline
  $\mathbb{E}_{t-1}[\delta_t(i,j)-\delta_{t-1}(i,j)]\geq C\alpha_{0}(i)\delta_{0}(i,j)$ & $t-1<\min\{\taujweak,\taudeltadown,\tauidown\}$   \\ \hline \hline
  $\mathbb{E}_{t-1}[\alphanorm_t-\alphanorm_{t-1}]\geq 0$ & $\forall t$  \\ \hline 
  \end{tabular}
  \caption{Summary of the drift terms of $\alpha_t, \delta_t$, and $\alphanorm_t$ used in \cref{lem:drift analysis for basic} (for both 3-Majority and 2-Choices).
  For each item, $C>0$ is a carefully chosen constant.}
  \label{table:drift}
  \end{table}

\begin{lemma}[Drift analysis for basic quantities]
  \label{lem:drift analysis for basic}
  Consider stopping times defined in \cref{def:stopping times}.
  Fix two distinct opinions $i,j\in [k]$.
  Then, we have the following:
  \begin{enumerate}
    \renewcommand{\labelenumi}{(\roman{enumi})}
\item \label{item:tauiup}
For any constant $\varepsilon\in (0,1)$, 
let $\constrefs{lem:drift analysis for basic}{item:tauiup}\defeq \frac{ (1 - \varepsilon) \ciup }{(1+\ciup)^2}$.
Then, we have
\begin{align*}
\Pr\qty[ \tauiup \le \frac{\constrefs{lem:drift analysis for basic}{item:tauiup}}{ \alpha_0(i)} ] \le 
    \begin{cases}
        \exp\qty(-\Omega\qty( n\alpha_0(i)^2 ))	& \text{for 3-Majority},\\[10pt]
        \exp\qty( -\Omega\qty( n\alpha_0(i) ) ) & \text{for 2-Choices}.
    \end{cases}
\end{align*}
\item \label{item:tauidown}
For any constant $\varepsilon \in (0,1)$, let
$\constrefs{lem:drift analysis for basic}{item:tauidown}\defeq \frac{(1-\cweak)(1-\varepsilon)\calphadown}{\cweak (1+\ciup)^2}$.
Then, we have
\begin{align*}
    \Pr\qty[ \tauidown \le \min\cbra*{\tauiweak,\tauiup,\frac{\constrefs{lem:drift analysis for basic}{item:tauidown}}{\alpha_0(i)}} ] \le 
    \begin{cases}
            \exp\qty(-\Omega\qty( n\alpha_0(i)^2 ))	& \text{for 3-Majority},\\[10pt]
            \exp\qty( -\Omega\qty( n\alpha_0(i) ) ) & \text{for 2-Choices}.
    \end{cases}
\end{align*}
\item \label{item:tauiactive}
For any $T>0$, we have
\begin{align*}
    \Pr\qty[\tauiactive\leq \min\{T,\taunormdown\}]
    \leq 
    \begin{cases}
      \exp\qty(-\Omega\qty(\frac{n\alphanorm_0}{T})) & \text{for 3-Majority}, \\[10pt]
      \exp\qty(-\Omega\qty(\frac{n\alphanorm_0}{T\alphanorm_0+1})) & \text{for 2-Choices}.
    \end{cases}
\end{align*}
\item \label{item:taudeltadown}
        For any $T> 0$, we have
    \begin{align*}
         \Pr\qty[ \taudeltadown \le \min\{\taujweak,\tauiup,T\} ] \le
         \begin{cases}
           \exp\qty( -\Omega\qty( \frac{n\delta_0(i,j)^2}{\alpha_0(i)T+\delta_0(i,j)} ) ) & \text{for 3-Majority}, \\[10pt]
           \exp\qty( -\Omega\qty(\frac{n \delta_0(i,j)^2}{\alpha_0(i)^2 T+\delta_0(i,j)}) )	&  \text{for 2-Choices}.
         \end{cases}
    \end{align*}
    \item \label{item:taudeltaup}
    For any constant $\varepsilon\in (0,1)$,
    let $\constrefs{lem:drift analysis for basic}{item:taudeltaup}\defeq \frac{(1-\cweak)(1+\varepsilon)\cdeltaup}{(1-2\cweak)(1-\calphadown)(1-\cdeltadown)}$.
    Then, we have
    \begin{align*} 
         \Pr\qty[ \min\cbra*{ \taudeltaup, \taujweak, \taudeltadown,  \tauiup, \tauidown } > \frac{\constrefs{lem:drift analysis for basic}{item:taudeltaup}}{\alpha_0(i)}]
         \leq 
         \begin{cases}
          \exp\qty(-\Omega(n\delta_0(i,j)^2)) & \text{for 3-Majority}, \\[10pt]
          \exp\qty(-\Omega\qty(\frac{n\delta_0(i,j)^2}{\alpha_0(i)}))	& \text{for 2-Choices}.
         \end{cases}
    \end{align*}
    \item \label{item:taunormdown}
    For any $T>0$, we have
    \begin{align*}
        \Pr\qty[\taunormdown\leq \min\cbra*{T,\taunormup}]\le
          \begin{cases}
          \exp\qty(-\Omega\qty(\frac{n\sqrt{\alphanorm_0}}{T})) & \text{for 3-Majority}, \\[10pt]
          \exp\qty(-\Omega\qty(\frac{n}{T+\alphanorm_0^{-1/2}})) & \text{for 2-Choices}.
        \end{cases}
    \end{align*}
\end{enumerate}
\end{lemma}

Before showing \cref{lem:drift analysis for basic}, we list the following inequalities that hold in special cases.
The proof is a straightforward calculation, and we put the proof in \cref{sec:special inequalities}.
\begin{lemma}[Inequalities for non-weak opinions]
    \label{lem:two strong opinions}
    Fix two distinct opinions $i,j\in [k]$.
    Consider the quantities $\alpha_t(i),\delta_t,\alphanorm_t$ (\cref{def:basic quantities}) and the
    stopping times $\tauiweak,\taujweak$ defined in \cref{def:stopping times}.
    Then, we have the following for $t-1<\min\{\tauiweak,\taujweak\}$:
    \begin{enumerate}
            \renewcommand{\labelenumi}{(\roman{enumi})}
        \item \label{item:lower bound for multipricative term} 
        $\alpha_{t-1}(i)+\alpha_{t-1}(j)-\alphanorm_{t-1}\geq \frac{1-2\cweak}{1-\cweak}\max\qty{\alpha_{t-1}(i),\alpha_{t-1}(j)}.$
        
        \item \label{item:lower bound for variance}
        For a positive constant $\constref{lem:two strong opinions}\defeq 1-\frac{1}{\sqrt{2(1-\cweak)}}>0$,
        \begin{align*} 
            \Var_{t-1}\qty[\delta_t]
            \geq \begin{cases}
                \constref{lem:two strong opinions}^3 \cdot \frac{\alpha_{t-1}(i)+\alpha_{t-1}(j)}{n}	& \text{for 3-Majority},\\[10pt]
               \constref{lem:two strong opinions}^2 \cdot \frac{\alpha_{t-1}(i)^2+\alpha_{t-1}(j)^2}{n}	& \text{for 2-Choices}.
            \end{cases}
       \end{align*}
    \end{enumerate}
\end{lemma}

\begin{proof}[Proof of \cref{item:tauiup} of \cref{lem:drift analysis for basic}]
    Let $X_t=\alpha_t(i)$, $\tau=\tauiup$, and $\drift = 
        (1+\ciup)^2\alpha_0(i)^2$.
    Suppose $\tau>t-1$.
    For both 2-Choices and 3-Majority, we have
    \begin{align*}
    \E_{t-1}\sbra{\alpha_t(i)}
    =\alpha_{t-1}(i)\qty(1+\alpha_{t-1}(i)-\alphanorm_{t-1})
    \leq \alpha_{t-1}(i)+\alpha_{t-1}(i)^2
    \leq \alpha_{t-1}(i)+R.
    \end{align*}
    Therefore, we have
    \begin{align*}
    \indicator_{\tau>t-1}\qty(\E_{t-1}\qty[X_t]-X_{t-1}-R)
    =\indicator_{\tau>t-1}\qty(\E_{t-1}\qty[\alpha_t(i)]-\alpha_{t-1}(i)-R)
    \leq 0.
    \end{align*}
    That is, we checked the condition \ref{item:C1} of \cref{lem:Freedman stopping time additive}.

    Now we check that 3-Majority and 2-Choices satisfy \ref{item:C2} or \ref{item:C2'} for $\bounded = \frac{1}{n}$ and
        \begin{align} 
       s = \begin{cases}
        O\qty( \frac{\alpha_0(i)}{n} )	& \text{ for 3-Majority},\\[10pt]
        O\qty( \frac{\alpha_0(i)^2}{n} ) & \text{ for 2-Choices}.
       \end{cases} \label{eq:variance tauiup}
    \end{align}

    \paragraph*{3-Majority.}
    We have $\alpha_{t-1}(i) \le (1+\ciup)\alpha_0(i)$ if $ t-1 < \tau $.
    Hence, from \cref{lem:BC for SP} and \cref{item:BC for alpha} of \cref{lem:Bernstein condition for sync processes}, 
    $\indicator_{\tau>t-1}\qty(\alpha_t(i)-\E_{t-1}\qty[ \alpha_{t-1}(i)])$ conditioned on round $t-1$ satisfies $\qty(\frac{1}{n},\variance)$-Bernstein condition (this verifies the condition \ref{item:C2'}).

    \paragraph*{2-Choices.}
    We deal with two cases: $\alpha_{t-1}(i)\geq \alphanorm_{t-1}$ or not.
    First, suppose $\alpha_{t-1}(i)\geq \alphanorm_{t-1}$.
    We have $\alphanorm_{t-1}\leq \alpha_{t-1}(i)$ and $\alpha_{t-1}(i) \le (1+\ciup)\alpha_0(i)$ if $ t-1 < \tau $.
    From \cref{lem:BC for SP} and \cref{item:BC for alpha} of \cref{lem:Bernstein condition for sync processes},
    $\indicator_{\tau>t-1}\qty(\alpha_t(i)-\E_{t-1}\qty[ \alpha_{t-1}(i)])$ conditioned on round $t-1$ satisfies  $\qty(\frac{1}{n},\variance)$-Bernstein condition.
    Hence, from \cref{item:dominated Bernstein condition} of \cref{lem:Bernstein condition}, the random variable
    $\indicator_{\tau>t-1}\qty(\alpha_t(i)-\alpha_{t-1}(i)-R)\leq \indicator_{\tau>t-1}\qty(\alpha_t(i)-\E_{t-1}\qty[ \alpha_{t-1}(i)])$ satisfies one-sided $\qty(\frac{1}{n},\variance)$-Bernstein condition.
    

    Second, consider the other case where $\alpha_{t-1}(i)\leq \alphanorm_{t-1}$.
    From \cref{lem:concentration of alpha difference for 2 choices} and \cref{item:Bernstein condition 2} of \cref{lem:Bernstein condition},
    the random variable $\indicator_{\tau>t-1}\qty(\alpha_t(i)-\alpha_{t-1}(i))$ satisfies one-sided $\qty(\frac{\indicator_{\tau>t-1}}{n},\frac{\indicator_{\tau>t-1}\alpha_{t-1}(i)^2}{n})$-Bernstein condition.
    Hence, $\indicator_{\tau>t-1}\qty(\alpha_t(i)-\alpha_{t-1}(i)-R)\leq \indicator_{\tau>t-1}\qty(\alpha_t(i)-\alpha_{t-1}(i))$ satisfies one-sided $\qty(\frac{1}{n}, \variance)$-Bernstein condition from \cref{item:dominated Bernstein condition} of \cref{lem:Bernstein condition}.
    
    Hence, in any case, $\indicator_{\tau>t-1}\qty(\alpha_t(i)-\alpha_{t-1}(i)-R)$ satisfies one-sided $\qty(\frac{1}{n}, \variance)$-Bernstein condition (this verifies the condition \ref{item:C2}).

    Applying \cref{item:positive drift} of \cref{lem:Freedman stopping time additive} for 
    $\bounded=\frac{1}{n}$, $h=\calphaup\alpha_0(i)$, and $T = \frac{\constrefs{lem:drift analysis for basic}{item:tauiup}}{ \alpha_0(i)}$ (then, $ z = h-\drift\cdot T = \varepsilon\calphaup\alpha_0(i) $), we have
    \begin{align*}
    \Pr\qty[\tau\leq T]
    &=\Pr\qty[\tauxplus\leq \min\{T,\tau\}]
    \le \exp\qty( -\Omega\qty(\frac{\alpha_0(i)^2}{ \variance T + \alpha_0(i)/n }) ),
    \end{align*}
    where $ \tauxplus = \inf\qty{ t\ge 0\colon X_t\ge X_0 + h } = \inf\{ t\ge 0\colon \alpha_t(i)\ge \alpha_0(i) + \calphaup\alpha_0(i) \}= \tauiup$.
    Substituting \cref{eq:variance tauiup}, we obtain the claim.
\end{proof}

\begin{proof}[Proof of \cref{item:tauidown} of \cref{lem:drift analysis for basic}]
    Let $X_t=-\alpha_t(i)$, $\tau = \min\{\tauiweak,\tauiup\}$, and
      \begin{align*}
       \drift = 
       \frac{\cweak (1+\ciup)^2}{1-\cweak}\cdot\alpha_0(i)^2. 
      \end{align*}
    Suppose $\tau>t-1$.
    For both models, we have
      \begin{align*} 
         \E_{t-1}\qty[ \alpha_t(i) ] &= \alpha_{t-1}(i) \qty( 1 + \alpha_{t-1}(i) - \alphanorm_{t-1} ) \\
         &\ge \alpha_{t-1}(i) \qty( 1 - \frac{\cweak}{1-\cweak}\cdot \alpha_{t-1}(i) ) & & \because\text{$t-1 < \tauiweak$; thus $\alphanorm_{t-1} \le \frac{\alpha_{t-1}(i)}{1-\cweak}$} \\
         &\ge \alpha_{t-1}(i) - \frac{\cweak (1+\ciup)^2}{1-\cweak}\cdot\alpha_0(i)^2. & & \because\text{$t-1 <\tauiup$; thus $\alpha_{t-1}(i) \le (1+\ciup)\alpha_0(i)$}
    \end{align*}
    From above, it is easy to check the condition \ref{item:C1} of \cref{lem:Freedman stopping time additive} as follows:
    \begin{align*}
        \indicator_{\tau>t-1}\qty(\E_{t-1}\qty[ X_t ]-X_{t-1}-R)
        &=\indicator_{\tau>t-1}\qty(\alpha_{t-1}(i)-R-\E_{t-1}\qty[ \alpha_t(i) ])
        \leq 0.
    \end{align*}

    Now we check that 3-Majority and 2-Choices satisfy \ref{item:C2} or \ref{item:C2'} for $\bounded = \frac{1}{n}$ and $\variance$ defined in~\cref{eq:variance tauiup}.

    \paragraph*{3-Majority.}
    We have $\alpha_{t-1}(i) \le (1+\ciup)\alpha_0(i)$ if $ t-1 < \tau $.
    Hence, from \cref{lem:BC for SP} and \cref{item:BC for alpha} of \cref{lem:Bernstein condition for sync processes}, 
    $\indicator_{\tau>t-1}\qty(\E_{t-1}\qty[ \alpha_{t-1}(i)]-\alpha_t(i))$ conditioned on round $t-1$ satisfies $ \qty( \frac{1}{n}, \variance ) $-Bernstein condition (this verifies the condition \ref{item:C2'}).
    

    \paragraph*{2-Choices.}
    We have $\alpha_{t-1}(i) \le (1+\ciup)\alpha_0(i)$ and $\alphanorm_{t-1}\leq \frac{\alpha_{t-1}(i)}{1-\cweak}\leq \frac{1+\calphaup}{1-\cweak}\alpha_{0}(i)$ if $ t-1 < \tau $.
Hence, from \cref{lem:BC for SP} and \cref{item:BC for alpha} of \cref{lem:Bernstein condition for sync processes}, $\indicator_{\tau>t-1}\qty(\E_{t-1}\qty[ \alpha_{t-1}(i)]-\alpha_t(i))$ satisfies $\qty(\frac{1}{n},\variance)$-Bernstein condition (this verifies the condition \ref{item:C2'}).

    %
    %

    Applying \cref{item:positive drift} of \cref{lem:Freedman stopping time additive} for $\bounded = \frac 1 n$, $h=\calphadown \alpha_0(i)$, and $T=\frac{\constrefs{lem:drift analysis for basic}{item:tauidown}}{\alpha_0(i)}$ (then, $z=h-\drift\cdot T=\varepsilon \calphadown \alpha_0(i)$),
    we obtain
    \begin{align*}
        \Pr\qty[\tauxplus\leq \min\{T,\tau\}]
        \leq\exp\qty( -\Omega\qty(\frac{\alpha_0(i)^2}{ \variance T + (\alpha_0(i)/n) }) ).
    \end{align*}
    Note that $\tauxplus=\inf\{t\geq 0: X_t\geq X_0+h\}=\inf\{t\geq 0: -\alpha_t(i)\geq -\alpha_0(i)+\calphadown \alpha_0(i)\}=\tauidown$.
    Substituting \cref{eq:variance tauiup}, we obtain the claim.

  \end{proof}

\begin{proof}[Proof of \cref{item:tauiactive} of \cref{lem:drift analysis for basic}]
    Let $X_t=\alpha_t(i)$,
    $\tau=\min\cbra{\tauincreasei,\taunormdown}$, and $\drift=0$.
    Suppose $\tau>t-1$.
    For both models, we have
    \begin{align*}
    \E_{t-1}\qty[\alpha_t(i)]
    =\alpha_{t-1}(i)\qty(1+\alpha_{t-1}(i)-\alphanorm_{t-1})
    \leq \alpha_{t-1}(i)\qty(1+(1-\cactive)\alphanorm_{0}-(1-\cnormdown)\alphanorm_{0})
    \leq \alpha_{t-1}(i).
    \end{align*}
    From above, it is easy to check the condition \ref{item:C1} of \cref{lem:Freedman stopping time additive} as follows:
    \begin{align*}
        \indicator_{\tau>t-1}\qty(\E_{t-1}\qty[ X_t ]-X_{t-1}-R)
        &=\indicator_{\tau>t-1}\qty(\E_{t-1}\qty[ \alpha_t(i) ]-\alpha_{t-1}(i))
        \leq 0.
    \end{align*}

    Now we check that 3-Majority and 2-Choices satisfy \ref{item:C2} or \ref{item:C2'} for $\bounded = \frac{1}{n}$ and
        \begin{align} 
       s = \begin{cases}
        O\qty( \frac{\alphanorm_0}{n} )	& \text{ for 3-Majority},\\[10pt]
        O\qty( \frac{\alphanorm_0^2}{n} ) & \text{ for 2-Choices}.
       \end{cases} \label{def:variance alphaactive}
    \end{align}

    \paragraph*{3-Majority.}
    We have $\alpha_{t-1}(i)\leq (1-\cactive)\alphanorm_0$ if $ t-1 < \tau $.
    Hence, from \cref{lem:BC for SP} and \cref{item:BC for alpha} of \cref{lem:Bernstein condition for sync processes}, 
    $\indicator_{\tau>t-1}\qty(\alpha_t(i)-\E_{t-1}\qty[ \alpha_{t-1}(i)])$ conditioned on round $t-1$ satisfies $\qty(\frac{1}{n},\variance)$-Bernstein condition (this verifies the condition \ref{item:C2'}).
    

    \paragraph*{2-Choices.}
    We deal with two cases: $\alpha_{t-1}(i)\geq \alphanorm_{t-1}$ or not.
    First, suppose $\alpha_{t-1}(i)\geq \alphanorm_{t-1}$.
    We have $\alpha_{t-1}(i)(\alpha_{t-1}(i)+\alphanorm_{t-1})\leq 2\alpha_{t-1}(i)^2\leq 2(1-\cactive)^2\alphanorm_0^2$ if $ t-1 < \tau $.
    From \cref{lem:BC for SP} and \cref{item:BC for alpha} of \cref{lem:Bernstein condition for sync processes}, 
    $\indicator_{\tau>t-1}\qty(\alpha_t(i)-\E_{t-1}\qty[ \alpha_{t-1}(i)])$ conditioned on round $t-1$ satisfies $\qty(\frac{1}{n},\variance)$-Bernstein condition.
    Furthermore, from \cref{item:dominated Bernstein condition} of \cref{lem:Bernstein condition}, $\indicator_{\tau>t-1}\qty(\alpha_t(i)-\alpha_{t-1}(i)-R)\leq \indicator_{\tau>t-1}\qty(\alpha_t(i)-\E_{t-1}\qty[ \alpha_{t-1}(i)])$ satisfies one-sided $\qty(\frac{1}{n},\variance)$-Bernstein condition.


    Second, consider the other case where $\alpha_{t-1}(i)\leq \alphanorm_{t-1}$.
    In this case, from \cref{lem:concentration of alpha difference for 2 choices} and \cref{item:Bernstein condition 2} of \cref{lem:Bernstein condition}, the random variable $\indicator_{\tau>t-1}\qty(\alpha_t(i)-\alpha_{t-1}(i))$ satisfies one-sided $\qty(\frac{\indicator_{\tau>t-1}}{n},\frac{\indicator_{\tau>t-1}\alpha_{t-1}(i)^2}{n})$-Bernstein condition.
    Hence, $\indicator_{\tau>t-1}\qty(\alpha_t(i)-\alpha_{t-1}(i)-R)\leq \indicator_{\tau>t-1}\qty(\alpha_t(i)-\alpha_{t-1}(i))$ satisfies one-sided $\qty(\frac{1}{n}, \variance)$-Bernstein condition from \cref{item:dominated Bernstein condition} of \cref{lem:Bernstein condition}.
    Note that $\alpha_{t-1}(i)\leq (1-\cactive)\alphanorm_0$ for $\tau>t-1$.

    Hence, in any case, $\indicator_{\tau>t-1}\qty(\alpha_t(i)-\alpha_{t-1}(i)-R)$ satisfies one-sided $\qty(\frac{1}{n}, \variance)$-Bernstein condition (This verifies the condition \ref{item:C2}).

    Applying \cref{item:positive drift} of \cref{lem:Freedman stopping time additive} for $\bounded=\frac{1}{n}$, $h=z=(1-\cactive)\alphanorm_0-\alpha_0(i)\geq \varepsilon (1-\cactive)\alphanorm_0$, we have
    \begin{align*}
    \Pr\qty[\tauxplus\leq \min\{T,\tau\}]
    \leq \exp\qty(-\Omega\qty(\frac{\alphanorm_0^2}{\variance T+\alphanorm_0/n})).
    \end{align*}
    Note that $\tauxplus=\inf\{t\geq 0:X_t\geq X_0+h\}=\inf\{t\geq 0:\alpha_t(i)\geq \alpha_0(i)+h\}=\tauincreasei$.
    Substituting \cref{def:variance alphaactive}, we obtain the claim.
\end{proof}

    \begin{proof}[Proof of \cref{item:taudeltadown} of \cref{lem:drift analysis for basic}]
    Let $X_t = -\delta_t$, $\tau = \min\{\taujweak, \taudeltadown, \tauiup\}$, and $ \drift = 0 $.
    
    Suppose $t-1 <  \min\{ \taujweak, \taudeltadown\}$.
    For both models, from 
    \cref{item:delta} of \cref{lem:basic inequality} and 
    \cref{item:lower bound for multipricative term} of \cref{lem:two strong opinions}, we have
    \begin{align} 
        \E_{t-1}\qty[ \delta_t ] &= \delta_{t-1}  + \delta_{t-1}\qty(\alpha_{t-1}(i) + \alpha_{t-1}(j) - \alphanorm_{t-1} ) 
        && (\text{$\because$ \cref{item:delta} of \cref{lem:basic inequality}})\nonumber\\
        &\geq \delta_{t-1} + \frac{1-2\cweak}{1-\cweak}\alpha_{t-1}(i)\delta_{t-1}.  
        && (\text{$\because$ $\delta_{t-1}\ge 0$ and \cref{item:lower bound for multipricative term} of \cref{lem:two strong opinions}}) \label{eq:Edelta sync}\\
        &\geq \delta_{t-1}. && (\text{$\because$ $\delta_{t-1}\ge 0$}) \nonumber
    \end{align}
    Hence, for both models, we have
      \begin{align*}
        \indicator_{\tau>t-1}\qty(\E_{t-1}\qty[ X_t ]-X_{t-1}-R)
        &=\indicator_{\tau>t-1}\qty(\delta_{t-1}-\E_{t-1}\qty[\delta_t])
        \leq 0.
    \end{align*}
  This verifies the condition \ref{item:C1} of \cref{lem:Freedman stopping time additive}.
    Now, we check that 3-Majority and 2-Choices satisfy the condition \ref{item:C2'} of \cref{lem:Freedman stopping time additive} 
      for $ \bounded = \frac{2}{n} $ and
    \begin{align}  \label{eq:variance deltaup}
       \variance = \begin{cases}
          O\qty( \frac{\alpha_{0}(i)}{n} )	& \text{for 3-Majority},\\[10pt]
          O\qty( \frac{\alpha_{0}(i)^2}{n} ) & \text{for 2-Choices}.
       \end{cases} 
    \end{align}

    \paragraph*{3-Majority.}
    We have $\alpha_{t-1}(j) \le \alpha_{t-1}(i) \le (1+\calphaup)\alpha_0(i)$ if $ t-1 < \tau $.
    Hence, from \cref{lem:BC for SP} and \cref{item:BC for delta} of \cref{lem:Bernstein condition for sync processes}, 
    $\indicator_{\tau>t-1}\qty(\E_{t-1}\qty[ \delta_{t-1}]-\delta_t)$ conditioned on round $t-1$ satisfies $\qty(\frac{2}{n},\variance)$-Bernstein condition (this verifies the condition \ref{item:C2'}).
    

    \paragraph*{2-Choices.}
    If $ t-1 < \tau $, we have $\alpha_{t-1}(j) \le \alpha_{t-1}(i) \le (1+\calphaup)\alpha_0(i)$ and $ \alphanorm_{t-1} \leq \frac{\alpha_{t-1}(j)}{1-\cweak} \leq \frac{(1+\calphaup)\alpha_{0}(i)}{1-\cweak} $.
    Hence, from \cref{lem:BC for SP} and \cref{item:BC for delta} of \cref{lem:Bernstein condition for sync processes}, 
    $\indicator_{\tau>t-1}\qty(\E_{t-1}\qty[ \delta_{t-1}]-\delta_t)$ conditioned on round $t-1$ satisfies $\qty(\frac{2}{n},\variance)$-Bernstein condition (this verifies the condition \ref{item:C2'}).


   Applying \cref{item:positive drift} of \cref{lem:Freedman stopping time additive} with $\bounded = \frac 2 n $, $h=z= \cdeltadown \delta_0$, we obtain
    \begin{align*}
        \Pr\qty[\tauxplus\leq \min\{T,\tau\}]
        \leq \exp\qty(-\Omega\qty(\frac{\delta_0^2}{ \variance T+\delta_0/n})).
    \end{align*}
    Note that $\tauxplus=\inf\{t\geq 0: X_t\geq X_0+h\}=\inf\{t\geq 0: -\delta_t\geq -\delta_0+\cdeltadown \delta_0\}=\taudeltadown$.
    Substituting concrete value \cref{eq:variance deltaup}, we obtain the claim.

  \end{proof}

\begin{proof}[\cref{item:taudeltaup} of \cref{lem:drift analysis for basic}]
Let $X_t = -\delta_t$, $\tau = \min\cbra*{ \taujweak, \taudeltadown, \tauiup, \tauidown}$, and
    \begin{align*}  
       \drift = 
        -\frac{(1-2\cweak)(1-\calphadown)(1-\cdeltadown)}{1-\cweak}\cdot \alpha_{0}(i)\delta_{0}.
    \end{align*}
Suppose $t-1 <  \min\cbra*{ \taujweak, \taudeltadown, \tauidown}$.
Then, from \cref{eq:Edelta sync}, we have
    \begin{align*}
        \E_{t-1}\qty[ \delta_t ] 
        &\ge \delta_{t-1} + \frac{(1-2\cweak)(1-\calphadown)(1-\cdeltadown)}{1-\cweak}\alpha_{0}(i)\delta_{0}
        &&\text{$\because$ \cref{eq:Edelta sync} and $t-1 <  \min\cbra*{  \taudeltadown, \tauidown}$}\\
        &=\delta_{t-1}-R.
    \end{align*}
    From above, it is easy to check the condition \ref{item:C1} of \cref{lem:Freedman stopping time additive} as follows:
    \begin{align*}
        \indicator_{\tau>t-1}\qty(\E_{t-1}\qty[ X_t ]-X_{t-1}-R)
        =\indicator_{\tau>t-1}\qty(\delta_{t-1}-\E_{t-1}\qty[ \delta_t]-R)
        \leq 0.
    \end{align*}

  Furthermore, from the same argument in the proof of \cref{item:taudeltaup}, 
  both models satisfy the condition \ref{item:C2'} of \cref{lem:Freedman stopping time additive} for 
  $ \bounded = \frac{2}{n} $ and $ \variance $ defined in~\cref{eq:variance deltaup}.

        Applying \cref{item:negative drift} of \cref{lem:Freedman stopping time additive} with 
    $\bounded = \frac{2}{n}$, $h=\cdeltaup \delta_0$, and $T=\frac{\constrefs{lem:drift analysis for basic}{item:taudeltaup}}{\alpha_0(i)}$ 
    (then, $z=(-\drift)\cdot T-h=\varepsilon \cdeltaup \delta_0$), 
    we obtain
    \begin{align*}
        \Pr\qty[\min\{\tauxminus,\tau\}>T]
        \leq \exp\qty(-\Omega\qty(\frac{\delta_0(i)^2}{\variance T+\delta_0/n})).
    \end{align*}
    Note that $\tauxminus=\inf\{t\geq 0: X_t\leq X_0-h\}=\inf\{t\geq 0: -\delta_t\leq -\delta_0-\cdeltaup \delta_0\}=\taudeltaup$.
    Substituting \cref{eq:variance deltaup}, we obtain the claim.
\end{proof}

\begin{proof}[Proof of \cref{item:taunormdown} of \cref{lem:drift analysis for basic}]
  Let $\tau = \taunormup$, $X_t=-\alphanorm_{t\wedge \tau}$, and $\drift = 0$.
  For both models, from \cref{item:expectation of alphanorm} of \cref{lem:basic inequality}, 
    \begin{align*}
        \indicator_{\tau>t-1}\qty(\E_{t-1}\qty[ X_t ]-X_{t-1}-R)
        &=\indicator_{\tau>t-1}\qty(\alphanorm_{t-1}-\E_{t-1}\qty[\alphanorm_{t}])
        \leq 0.
    \end{align*}
    Furthermore, from \cref{item:BC for alphanorm} of \cref{lem:Bernstein condition for sync processes}
    and \cref{item:Bernstein condition 1} of \cref{lem:Bernstein condition}, 
    the random variable
    \begin{align*}
        \indicator_{\tau>t-1}\qty( X_t - X_{t-1}-R)
        &=\indicator_{\tau>t-1}\qty(\alphanorm_{t-1}-\alphanorm_{t})
    \end{align*}
    satisfies one-sided $\qty(O\qty(\frac{\sqrt{\alphanorm_0}}{n}),\variance)$-Bernstein condition, where
    \begin{align*}
       \variance = \begin{cases}
          \frac{4(1+\cnormup)^{1.5}\alphanorm_{0}^{1.5}}{n}	& \text{for 3-Majority},\\[10pt]
          \frac{8(1+\cnormup)^{2}\alphanorm_{0}^2}{n} & \text{for 2-Choices}.
       \end{cases}
    \end{align*} 
    Here, we used $ \alphanorm_{t-1} \le (1+\cnormup)\alphanorm_0 $ for $ t-1 < \tau $.

    Applying \cref{lem:Freedman stopping time additive} with $\bounded = O\qty(\frac{\sqrt{\alphanorm_0}}{n})$, and $h=\varepsilon= \cnormdown \alphanorm_0$, 
    \begin{align*}
        \Pr\qty[\tauxplus\leq \min\{T,\tau\}]
        \leq \exp\qty(-\Omega\qty(\frac{\alphanorm_0^2}{\variance T+\alphanorm_0^{1.5}/n}))
    \end{align*}
    holds. Since $\tauxplus=\inf\{t\geq 0: X_t\geq X_0+h\}=\inf\{t\geq 0: -\alphanorm_t\geq -\alphanorm_0+\cnormdown \alphanorm_0\}=\taunormdown$, we obtain the claim.
\end{proof}

Finally, we introduce the following key lemma obtained from \cref{item:taunormdown} of \cref{lem:drift analysis for basic}. 
\begin{lemma}[Bounded decrease of $\alphanorm_t$]
\label{lem:taunormdown is large}
  Consider the stopping times $\taunormup,\taunormdown$ defined in \cref{def:stopping times}.
  Then, for any $T>0$, we have
    \begin{align*}
        \Pr\qty[\taunormdown\leq T]\le
          \begin{cases}
          T\cdot \exp\qty(-\Omega\qty(\frac{n\sqrt{\alphanorm_0}}{T})) & \text{for 3-Majority}, \\[10pt]
          T\cdot \exp\qty(-\Omega\qty(\frac{n}{T+\alphanorm_0^{-1/2}})) & \text{for 2-Choices}.
        \end{cases}
    \end{align*}
  Specifically, we have the following for a sufficiently large constant $C>0$:
  Suppose that $\alphanorm_0\geq \frac{C\log n}{\sqrt{n}}$ for 3-Majority and $\alphanorm_0\geq \frac{(C\log n)^2}{n}$ for 2-Choices. Then, $\Pr\qty[\taunormdown\leq \frac{C\log n}{\alphanorm_0}]\leq O(n^{-10})$.
\end{lemma}
\begin{proof}[Proof of \cref{lem:taunormdown is large}]
  For simplicity, we prove the claim for 3-Majority.
  The same proof works for 2-Choices.
    For each $0\le s \le T$,
    let $\sigma^{\downarrow}_s = \inf\{ t\ge s \colon \alphanorm_t \le (1-\cnormdown) \alphanorm_s \}$,
    $\sigma^{\uparrow}_s = \inf\{t \ge s \colon \alphanorm_t \ge 2\alphanorm_s\}$,
    and
    let $\calE^{(s)}$ be the event that $\alphanorm_s \ge \alphanorm_0$ and $\sigma^{\downarrow}_s \le \min\{T, \sigma^{\uparrow}_s\}$.
    Note that $\taunormdown = \sigma^{\downarrow}_0$ and $\taunormup = \sigma^{\uparrow}_0$ (for $\cnormup=1$).

    The key observation is that the partial process $(\alpha_t)_{t\ge s}$ is again a 3-Majority process and $\sigma_t^{\uparrow},\sigma_t^{\downarrow}$ can be seen as the stopping times of \cref{def:stopping times} for the partial process.
    Moreover, the event $\calE^{(s)}$ depends only on the partial process $(\alpha_t)_{t\ge s}$.
    Therefore, from \cref{item:taunormdown} of \cref{lem:drift analysis for basic}, we have
    \begin{align*}
        \Pr_{(\alpha_t)_{t\ge s}}\qty[ \calE^{(s)} ] &\le \Pr_{(\alpha_t)_{t\ge s}}\qty[ \sigma_s^{\downarrow} \le \min\{T, \sigma^{\uparrow}_s\} \middle| \alphanorm_s \ge \alphanorm_0 ] \\
        &\le \exp\qty( - \Omega\qty( \frac{n\sqrt{\alphanorm_0}}{T})).
    \end{align*}

    If $\taunormdown \le T$ occurs, then $\calE^{(s)}$ occurs for some $0\le s \le T$.
    For example, if $s\le \taunormdown$ is the round such that $\alphanorm_s = \max_{0\le t \le \taunormdown} \alphanorm_t$, then $\calE^{(s)}$ holds.
    Therefore, we have
    \begin{align*}
        \Pr\qty[ \taunormdown \le T ] &\le \Pr\qty[ \bigvee_{0\le s \le T} \calE^{(s)} ] \\
        &\le \sum_{0\le s \le T} \Pr\qty[ \calE^{(s)}] \\
        &\le T\exp\qty( - \Omega\qty( \frac{n\sqrt{\alphanorm_0}}{T} )).
    \end{align*}    

    Now, we prove the ``specifically'' part.
    For 3-Majority, substituting $T=\frac{C\log n}{\alphanorm_0}$ yields the claim since
    \begin{align*}
    \frac{n\sqrt{\alphanorm_0}}{T}
    =\frac{n\alphanorm_0^{1.5}}{C\log n}
    \geq  n^{1/4}\sqrt{C \log n}.
    \end{align*}
    For 2-Choices, substituting  $T=\frac{C\log n}{\alphanorm_0}$ yields the claim since
    \begin{align*}
    \frac{n}{T+\alphanorm_0^{-1/2}}
    = \frac{n}{\frac{C\log n}{\alphanorm_0}+\frac{1}{\sqrt{\alphanorm_0}}}
    \geq \frac{n}{\frac{n}{C\log n}+\frac{\sqrt{n}}{C\log n}}
    \geq \frac{C\log n}{2}.
    \end{align*}
\end{proof}

\section{Proof for 3-Majority and 2-Choices} \label{sec:proof}
We now have the tools developed in \cref{sec:drift analysis}, and thus present the proof of the main results for 3‑Majority and 2‑Choices. 
The main component of our proof is \cref{thm:consensus time large alphanorm}. 
We outline its proof (with \cref{thm:plurality} as a byproduct) in \cref{fig:overview}, which is followed by \cref{sec:weak opinion vanishes,sec:multiplicative drift of bias,sec:additive drift of bias,sec:bias amplification}.

In \cref{sec:weak opinion vanishes}, we show that any weak opinion vanishes (\cref{lem:weakvanish}). 
In \cref{sec:multiplicative drift of bias}, we demonstrate that the bias between two non‑weak opinions exhibits multiplicative drift (\cref{lem:deltaupweak}) and use this to prove that a sufficiently large initial bias leads to the emergence of weak opinions (\cref{lem:initial bias weak}). 
In \cref{sec:additive drift of bias}, we show that the squared bias between two non‑weak opinions exhibits an additive drift (\cref{lem:additive drift}).
In \cref{sec:bias amplification}, we show that the bias between two non‑weak opinions grows sufficiently large (\cref{lem:gap amplification}).
Finally, after proving the norm growth (\cref{lem:taunormplus}) in \cref{sec:growth of norm}, we conclude the proof of \cref{thm:main theorem} in \cref{sec:proof of main theorem}.

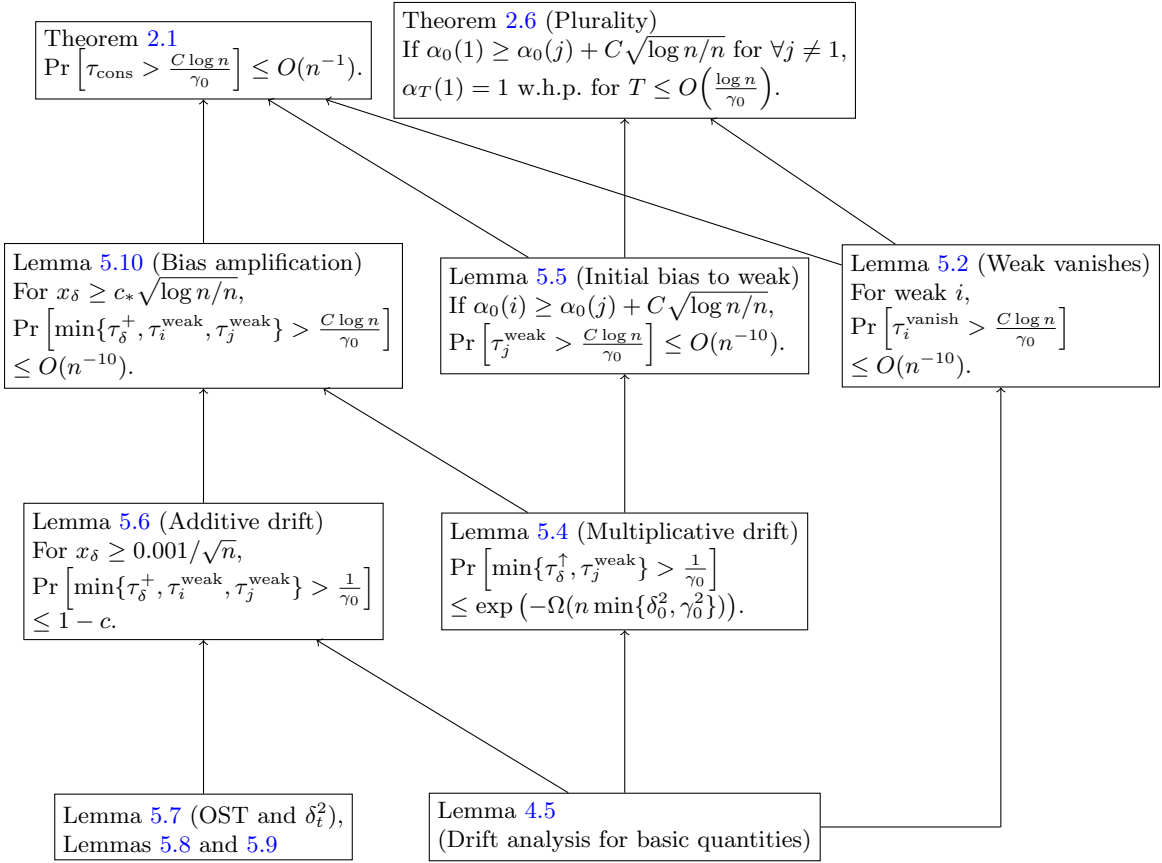
\begin{figure}[t]
\centering
    \begin{tikzpicture}
    \node[align=left,draw,rectangle] at (0,0) (main) 
    {\Cref{thm:consensus time large alphanorm} \\$\Pr\qty[\taucons>\frac{C\log n}{\alphanorm_0}]\leq O(n^{-1})$.};

    \node[align=left,draw,rectangle] at (5.6,0) (plurality) 
    {\cref{thm:plurality} (Plurality)\\
    If $\alpha_0(1)\geq \alpha_0(j)+C\sqrt{\log n/n}$ for $\forall j\neq 1$, \\
    $\alpha_T(1)=1$ w.h.p.~for $T\leq O\qty(\frac{\log n}{\alphanorm_0})$.};

    \node[align=left,draw,rectangle] at (0,-3.4) (gapamp) 
    {\cref{lem:gap amplification} (Bias amplification)\\
    For $\xdelta\geq c_*\sqrt{\log n/n}$, \\
    $\Pr\qty[\min\{\taudeltaplus,\tauiweak,\taujweak\}>\frac{C\log n}{\alphanorm_0}]$\\
    $\leq O(n^{-10})$.};

    \node[align=left,draw,rectangle] at (5.6,-3.4) (biasweak) 
    {\Cref{lem:initial bias weak} (Initial bias to weak)\\
    If $\alpha_0(i)\geq \alpha_0(j)+C\sqrt{\log n/n}$, \\
    $\Pr\qty[\taujweak>\frac{C\log n}{\alphanorm_0}]\leq O(n^{-10})$.};

    \node[align=left,draw,rectangle] at (10.6,-3.4) (vanish) 
    {\Cref{lem:weakvanish} (Weak vanishes)\\
    For weak $i$, \\
    $\Pr\qty[\tauivanish>\frac{C\log n}{\alphanorm_0}]$\\
    $\leq O(n^{-10})$.};
    
    \node[align=left,draw,rectangle] at (0,-6.8) (additive drift) 
    {\Cref{lem:additive drift} (Additive drift)\\
    For $\xdelta\geq 0.001/\sqrt{n}$, \\
    $\Pr\qty[\min\{\taudeltaplus,\tauiweak,\taujweak\}>\frac{1}{\alphanorm_0}]$\\
    $\leq 1-c$.};

    \node[align=left,draw,rectangle] at (5.6,-6.8) (deltaupweak) 
    {\Cref{lem:deltaupweak} (Multiplicative drift)\\
    $\Pr\qty[\min\{\taudeltaup, \taujweak\}>\frac{1}{\alphanorm_0}]$\\$
    \leq \exp\qty(-\Omega(n\min\{\delta_0^2,\alphanorm_0^2\}))$.};


    \node[align=left,draw,rectangle] at (0,-10.2) (OST) 
    {\Cref{lem:taudeltaplus OST} (OST and $\delta_t^2$), \\
    \Cref{lem:taudeltaplus general,lem:taudeltaplus one step sync}};


    \node[align=left,draw,rectangle] at (5.6,-10.2) (drift) 
    {\Cref{lem:drift analysis for basic} \\ (Drift analysis for basic quantities)};

    \draw[->] (gapamp) -- (main);
    \draw[->] (biasweak) -- (main);
    \draw[->] (biasweak) -- (plurality);
    \draw[->] (vanish) -- (main);
    \draw[->] (vanish) -- (plurality);
    \draw[->] (deltaupweak) -- (gapamp);
    \draw[->] (deltaupweak) -- (biasweak);
    \draw[->] (additive drift) -- (gapamp);
    \draw[->] (OST) -- (additive drift);
    \draw[->] (drift) -- (additive drift);
    \draw[->] (drift) -- (deltaupweak);
    \draw[-] (drift) -- (10.6,-10.2);
    \draw[->] (10.6,-10.2) -- (vanish);
\end{tikzpicture}
    \caption{Proof outline for 3-Majority in the case where $\alphanorm_0\geq C(\log n)/\sqrt{n}$. 
    Here, $C>0$ denotes a sufficiently large constant, $c\in (0,1)$ denotes a sufficiently small constant, and $c_*>0$ denotes an arbitrary constant.
    Throughout this proof outline, we use \cref{lem:taunormdown is large} to ensure $\alphanorm_t\geq C(1-\cnormdown)(\log n)/\sqrt{n}$ in a sufficiently long period. 
    The proof for 2-Choices follows a similar outline.}
    \label{fig:overview}
\end{figure}

\subsection{Weak Opinion Vanishes} \label{sec:weak opinion vanishes}
The first key tool is to show that any weak opinion vanishes within $ O((\log n)/\alphanorm_0) $ rounds for both 3-Majority and 2-Choices.
The key idea is to show that any weak opinion exhibits a multiplicative decreasing drift until the time corresponding to one of $\tauiactive, \taunormdown$ and $\tauivanish$ arrives.
We conclude the proof by demonstrating both $\tauiactive$ and $\taunormdown$ are sufficiently large, as established in \cref{item:tauiactive} of \cref{lem:drift analysis for basic} and \cref{lem:taunormdown is large}.
%
\begin{definition}[Vanishing time] \label{def:tauivanish}
  For an opinion $ i $, define $ \tauivanish $ as the first time when $ i $ vanishes, i.e.,
  \begin{align*} 
     \tauivanish = \inf\qty{ t \ge 0 \colon \alpha_t(i) = 0 }.
  \end{align*}
\end{definition}

\begin{lemma}[Weak opinion vanishes]\label{lem:weakvanish}
  Consider the stopping time $ \tauivanish $ defined in \cref{def:tauivanish}.
  Let $ i $ be an arbitrary weak opinion. 
  Suppose that, in 3-Majority, $ \alphanorm_0 \ge \frac{C\log n}{\sqrt{n}} $ and in 2-Choices, $ \alphanorm_0 \ge \frac{C(\log n)^2}{n} $, where $C>0$ is a sufficiently large constant. 
  Then, we have
  \begin{align*} 
     \Pr\qty[ \tauivanish \le \frac{C\log n}{\alphanorm_0} ] = 1-O(n^{-10}).
  \end{align*}
\end{lemma}
\begin{proof}
Let $ \tau = \min\{\tauiactive,\taunormdown,\tauivanish\} $ and 
  $r = 1 - (\cactive - \cnormdown)\alphanorm_0$.
  Note that $1>\cactive > \cnormdown>0$ from \cref{def:stopping times}.
  Suppose $t-1 < \tau$.
  For both models, we have
  \begin{align*} 
    \E_{t-1}\sbra*{ \alpha_t(i) } &= \alpha_{t-1}(i)(1+\alpha_{t-1}(i) - \alphanorm_{t-1}) \\
    &\le \alpha_{t-1}(i)(1+(1-\cactive)\alphanorm_0 - (1-\cnormdown)\alphanorm_0) \\
    &= r\alpha_{t-1}(i).
  \end{align*}

  Let $ X_t = r^{-t}\cdot \alpha_{t}(i) $
  and $ Y_t = X_{t\land \tau}$.
  Then, $(Y_t)$ is a supermartingale for both models since
  \begin{align*} 
     \E_{t-1}\qty[ Y_t - Y_{t-1} ] &=\indicator_{t-1<\tau}\cdot \E_{t-1}\qty[ X_t - X_{t-1} ] =\indicator_{t-1<\tau}r^{-t}\qty(\E_{t-1}\qty[ \alpha_t(i) ] - r\alpha_{t-1}(i)) \le 0.
  \end{align*}
  Furthermore, for any $T\ge 0$,
  \begin{align*}
    \E\sbra*{Y_T}
    &\geq \E\sbra*{Y_T\mid \tau>T}\Pr\sbra*{\tau>T}\\
    &=\E\sbra*{X_T\mid \tau>T}\Pr\sbra*{\tau>T} & & \because\text{$X_T=Y_T$ if $\tau>T$} \\
    &=r^{-T}\E\sbra*{\alpha_T(i)\mid \tau>T}\Pr\sbra*{\tau>T}\\
    &\geq r^{-T}n^{-1}\Pr\sbra*{\tau>T}. & & \because\text{$\alpha_T(i) \ge 1/n$ if $T>\taudisappeari$}
  \end{align*}
  and, since $ (Y_t) $ is a supermartingale, we obtain
  \begin{align}
    \Pr\sbra*{\tau>T}
    &\leq nr^T\E\sbra*{Y_T} \nonumber \\
    &\leq nr^T\E\sbra*{Y_0} \nonumber \\
    &\leq 
      n\exp\rbra*{-(\cactive-\cnormdown)\alphanorm_0T}. \label{eq:weak multi}
    \end{align}

  From \cref{lem:taunormdown is large} and \cref{item:tauiactive} of \cref{lem:drift analysis for basic} for $ T=\frac{C\log n}{\alphanorm_0}$ for a sufficiently large constant $C$, we have
  \begin{align} 
    \Pr[\taunormdown \le T] &\le 
     \begin{cases}
      T\exp\qty(-\Omega\qty(\frac{n\alphanorm_0^{1.5}}{\log n}))	& \text{for 3-Majority},\\[10pt]
      T\exp\qty(-\Omega\qty( \frac{n\alphanorm_0}{\log n} )) & \text{for 2-Choices}.
     \end{cases} \nonumber\\
     &\le n^{-10} \label{eq:taunormdown bound}
  \end{align}
  and
  \begin{align*}
    \Pr\qty[\tauiactive\leq \min\{T,\taunormdown\}]
      &\leq 
      \begin{cases}
          \exp\qty(-\Omega\qty(\frac{n\alphanorm_0^2}{\log n})) & \text{for 3-Majority}, \\[10pt]
          \exp\qty(-\Omega\qty(\frac{n\alphanorm_0}{\log n})) & \text{for 2-Choices}.\\
      \end{cases} \\
      &\le n^{-10}.
  \end{align*}
  Here, we used the assumption of $ \alphanorm_0 $.
  Thus, we have
  \begin{align} 
     \Pr\qty[ \tauiactive \le T ] &= \Pr\qty[ \tauiactive\le T \text{ and }\tauiactive \le \taunormdown ] + \Pr\qty[ \tauiactive\le T \text{ and }\tauiactive > \taunormdown ]  \nonumber\\
     &\le \Pr\qty[ \tauiactive \le \min\{T,\taunormdown\} ] + \Pr\qty[ \taunormdown \le T ] \nonumber\\
     &\le O(n^{-10}). \label{eq:tauiactive bound}
  \end{align}

  Therefore, 
  we have
  \begin{align*} 
     1-O(n^{-10}) &\le \Pr\qty[ \min\{\tauivanish,\taunormdown,\tauiactive\}\le T ] & & (\because\text{\cref{eq:weak multi}}) \\
     &=\Pr\qty[\tauivanish\le T\text{ or }\taunormdown\le T \text{ or }\tauiactive\le T] \\
     &\le \Pr\qty[\tauivanish\le T] + \Pr\qty[\taunormdown\le T] + \Pr\qty[\tauiactive\le T] \\
     &\le \Pr\qty[\tauivanish\le T] + O(n^{-10}). & & (\because\text{\cref{eq:tauiactive bound,eq:taunormdown bound}})
  \end{align*}
  That is, $ \Pr\qty[\tauivanish \le T] \ge 1-O(n^{-10}) $.
\end{proof}

\subsection{Multiplicative Drift of Bias}\label{sec:multiplicative drift of bias}
In this section, we first demonstrate that the bias between two non-weak opinions increases with a multiplicative factor (see \cref{lem:deltaupweak}). 
This result is derived by combining \cref{item:tauiup,item:tauidown,item:taudeltadown,item:taudeltaup} from \cref{lem:drift analysis for basic} with appropriately chosen constants. 
Subsequently, we show that a sufficiently large initial bias leads to the emergence of a weak opinion (see \cref{lem:initial bias weak}).

First, we introduce a quantity specific to 2-Choices that measures the bias between two opinions.
\begin{definition}[Scaled bias for 2-Choices] \label{def:tauetaup}
  Consider 2-Choices and let $ i,j $ be distinct opinions.
  Let
  \begin{align*} 
    \deltaratio_t(i,j) = \frac{\delta_t(i,j)}{\sqrt{\max\qty{\alpha_t(i),\alpha_t(j)}}}.
  \end{align*}
  For a constant $\cetaup>0$ and a parameter $ \xeta=\xeta(n)\in[1/n,1] $, let
  \begin{align*} 
    &\tauetaup = \inf\qty{ t\ge 0 \colon \deltaratio_t(i,j) \ge (1+\cetaup)\deltaratio_0 }, \\
    &\tauetaplus = \inf\qty{t\geq 0: \abs{\eta_t(i,j)}\geq \xeta}.
  \end{align*}  
  The constant $\cetaup$ is a universal constant and we will set $\cetaup=1/1000$.
\end{definition}
\begin{lemma}[Multiplicative drift of bias]\label{lem:deltaupweak}
  Let $ i,j $ be distinct opinions that are not weak at round $ 0 $ and $ \alpha_0(i) \ge \alpha_0(j) $.
  Consider the stopping times $\taujweak,\taudeltaup, \tauetaup$ defined in \cref{def:tauetaup,def:stopping times}
  and let $\cweak=1/10, \cdeltaup=1/20, \cetaup=1/1000$.
  Then, we have the following:
\begin{enumerate}
  \renewcommand{\labelenumi}{(\roman{enumi})}
  \item For 3-Majority,
  \begin{align*} 
     \Pr\qty[ \min\qty{ \taudeltaup,\taujweak } >  \frac{1}{\alphanorm_0} ] \le \exp\qty(-\Omega\qty(n\alphanorm_0^2)) + \exp\qty(-\Omega\qty(n\delta_0(i,j)^2)). 
  \end{align*}
  \item For 2-Choices,
  \begin{align*} 
    \Pr\qty[ \min\qty{ \tauetaup,\taujweak } >  \frac{1}{\alphanorm_0} ] \le  \exp\qty(-\Omega\qty(n\alphanorm_0)) + \exp\qty(-\Omega\qty( n\eta_0(i,j)^2 )).
  \end{align*}
\end{enumerate}
\end{lemma}
\begin{proof}
  Let
  \begin{align*} 
     P &= \begin{cases}
      \exp\qty(-\Omega\qty(n\alpha_0(i)^2)) & \text{for 3-Majority},\\
      \exp\qty(-\Omega\qty(n\alpha_0(i)))   & \text{for 2-Choices},
     \end{cases}\\
     Q &= \begin{cases}
      \exp\qty( -\Omega(n\delta_0^2) )	& \text{for 3-Majority},\\
      \exp\qty( -\Omega\qty( n\eta_0^2)) & \text{for 2-Choices}.
     \end{cases}
  \end{align*}
  Set $\calphaup=\calphadown=\cweak=\varepsilon=1/10$, and $\cdeltaup=\cdeltadown=1/20$.
  Then, constants appearing in \cref{lem:drift analysis for basic} become
  \begin{align*}
    \constrefs{lem:drift analysis for basic}{item:tauiup} &= \frac{ (1 - \varepsilon) \ciup }{(1+\ciup)^2}=\frac{9}{121}>0.073,\\
    \constrefs{lem:drift analysis for basic}{item:tauidown} &= \frac{(1-\cweak)(1-\varepsilon)\calphadown}{\cweak (1+\ciup)^2}=\frac{81}{121}>0.073,\\
    \constrefs{lem:drift analysis for basic}{item:taudeltaup} &= \frac{(1-\cweak)(1+\varepsilon)\cdeltaup}{(1-2\cweak)(1-\calphadown)(1-\cdeltadown)}=\frac{11}{152}<0.073. 
  \end{align*}
  In other words, letting $T=\frac{0.073}{\alpha_0(i)}$, we have 
  \begin{align*} 
      \frac{\constrefs{lem:drift analysis for basic}{item:taudeltaup}}{\alpha_0(i)} < T < \frac{\min\qty{\constrefs{lem:drift analysis for basic}{item:tauiup}, \constrefs{lem:drift analysis for basic}{item:tauidown}}}{\alpha_0(i)}.
  \end{align*}
  Note that we have $1/\alphanorm_0>T=0.073/\alpha_0(i)$ since $i$ is not weak at round $0$.
  
  First, we present a partial proof that is common for both 3-Majority and 2-Choices.
  Consider the stopping time
  $ \tau_0 = \min\qty{ \taujweak,\taudeltadown,\tauiup,\tauidown, T } $.
  Note that $ \tau_0 $ takes one of the values of $ \taujweak,\taudeltadown,\tauiup,\tauidown,T$.
  The cases are divided based on which value it takes.
  \begin{enumerate}
  \item Suppose $ \tau_0 = \tauiup$, which implies $ \tauiup \le T $. From \cref{item:tauiup} of \cref{lem:drift analysis for basic}, this occurs with probability $P$. \label{item:multi tauiup}
  \item Suppose $ \tau_0 = \taudeltadown$, which implies $ \taudeltadown \le \min\qty{ \taujweak, \tauiup, T }$. This occurs with probability $ Q $ from \cref{item:taudeltadown} of \cref{lem:drift analysis for basic}. \label{item:multi taudeltadown}
  \item Suppose $ \tau_0 = \tauidown$, which implies $ \tauidown\le \min\qty{ \taudeltadown, \taujweak, T }$.
    Observe that, for any $ 0\le t < \min\qty{\taudeltadown,\taujweak} $, we have $ \alpha_t(i)\ge \alpha_t(j) $; thus, $ i $ cannot become weak during these rounds. Therefore, $ \tauidown\le \min\qty{ \taudeltadown, \taujweak, T } \le \min\qty{ \taudeltadown, \tauiweak, T }  $.  This occurs with probability $ P $ from \cref{item:tauidown} of \cref{lem:drift analysis for basic}. 
    \label{item:multi tauidown}
  \end{enumerate}
  From above, we have
  \begin{align} 
     \Pr\qty[ \tau_0 = \min\qty{\taujweak,T} ] \ge 1-2P-Q.  \label{eq:tau0 value}
  \end{align}
  Note that, since the opinion $ i $ is not weak at round $ 0 $, we can substitute $ \alpha_0(i)=\Omega(\alphanorm_0) $ to $ P $.

  \begin{proof}[Proof for 3-Majority.]
  From \cref{eq:tau0 value}, we have
  \begin{align*}
  &\Pr\qty[\min\qty{\taudeltaup,\taujweak}>\frac{1}{\alphanorm_0}]\\
  &\leq \Pr\qty[\min\qty{\taudeltaup,\taujweak}>T] &&(\text{$\because$ $\alpha_0(i)\geq (1-\cweak)\alphanorm_0$})\\
  &= \Pr\qty[\min\qty{\taudeltaup,\taujweak,T}>T \textrm{ and } \tau_0=\min\qty{\taujweak,T}]\\
  &\; + \Pr\qty[\min\qty{\taudeltaup,\taujweak,T}>T \textrm{ and } \tau_0\neq \min\qty{\taujweak,T}]\\
  &\leq \Pr\qty[\min\qty{\taudeltaup,\tau_0}>T] + 2P+Q &&(\text{$\because$ \cref{eq:tau0 value}})\\
  &\leq 2P+2Q. &&(\text{$\because$ \cref{item:taudeltaup} of \cref{lem:drift analysis for basic}})
  \end{align*}
  This proves the claim for 3-Majority.
  \end{proof}

  \begin{proof}[Proof for 2-Choices.]
  The proof is similar to the proof for 3-Majority.
  The key difference is to consider $ \tauetaup $ in place of $ \taudeltaup $.
  From \cref{eq:tau0 value}, we have
  \begin{align*}
  &\Pr\qty[\min\qty{\tauetaup,\taujweak}>\frac{1}{\alphanorm_0}]\\
  &\leq \Pr\qty[\min\qty{\tauetaup,\taujweak}>T] &&(\text{$\because$ $\alpha_0(i)\geq (1-\cweak)\alphanorm_0$})\\
  &= \Pr\qty[\min\qty{\tauetaup,\taujweak,T}>T \textrm{ and } \tau_0=\min\qty{\taujweak,T}]\\
  &\; + \Pr\qty[\min\qty{\tauetaup,\taujweak,T}>T \textrm{ and } \tau_0\neq \min\qty{\taujweak,T}]\\
  &\leq \Pr\qty[\min\qty{\tauetaup,\tau_0}>T] + 2P+Q. &&(\text{$\because$ \cref{eq:tau0 value}})\\
  \end{align*}

  Now, recall $\calphaup=1/10,\cdeltaup=1/20$, and $\cetaup=1/1000$.
  Then, we have $\frac{1+\cdeltaup}{\sqrt{1+\calphaup}}=\frac{21\sqrt{110}}{220}>1+\cetaup$.
  Noting $ \tau_0 = \min\qty{ \taujweak,\taudeltadown,\tauiup,\tauidown, T } $, 
  we have
    \begin{align*}
        &\Pr\qty[\min\qty{\tauetaup,\tau_0}>T]\\
        &=\Pr\qty[\min\qty{\tauetaup,\tau_0}>T\textrm{ and }
        \tinyforall t\leq T, \frac{\delta_t}{\sqrt{\alpha_t(i)}}< (1+\cetaup)\frac{\delta_0}{\sqrt{\alpha_0(i)}}]&&(\because \tauetaup>T)\\
        &=\Pr\qty[\min\qty{\tauetaup,\tau_0}>T\textrm{ and }
        \tinyforall t\leq T, \delta_t< \underbrace{(1+\cetaup)\sqrt{1+\calphaup}}_{\le 1+\cdeltaup}\delta_0] &&(\because \tauiup>T)\\
        &\leq \Pr\qty[\min\qty{\tauetaup,\tau_0}>T\textrm{ and }\tinyforall t\leq T, \delta_t<(1+\cdeltaup)\delta_0] \\
        &\leq \Pr\qty[\min\qty{\taudeltaup,\tau_0}>T] \\
        &\le Q. & & (\because \text{\cref{item:taudeltaup} of \cref{lem:drift analysis for basic}})
    \end{align*}
    Combining the above, we obtain the claim for 2-Choices.
    \end{proof}
\end{proof}

\begin{lemma}[Initial bias leads a weak opinion]
    \label{lem:initial bias weak}
    Let $ i,j $ be distinct opinions that are not weak at round $ 0 $.
    Consider the stopping time $\taujweak $ defined in \cref{def:stopping times} and let $\cweak=1/10$.
    Then, we have the following:
    We have the following:
\begin{enumerate}
  \renewcommand{\labelenumi}{(\roman{enumi})}
        \item 
        Consider 3-Majority. 
        Suppose $\alpha_0(i)- \alpha_0(j)\geq C\sqrt{\frac{\log n}{n}}$ and 
        $\alphanorm_0\geq \frac{C\log n}{\sqrt{n}}$ for a sufficiently large constant $C>0$.
        Then,
        \[
        \Pr\qty[\taujweak>\frac{C\log n}{\alphanorm_0}]\leq O(n^{-10}).
        \]
        \item 
        Consider 2-Choices. 
        Suppose $\alpha_0(i)- \alpha_0(j)\geq C\sqrt{\frac{\alpha_0(i)\log n}{n}}$ and 
        $\alphanorm_0\geq \frac{(C\log n)^2}{n}$ for a sufficiently large constant $C>0.$
        Then,
        \[
        \Pr\qty[\taujweak>\frac{C\log n}{\alphanorm_0}]\leq O(n^{-10}).
        \]
\end{enumerate}
\end{lemma}
\begin{proof}[Proof for 3-Majority]
  First, from \cref{lem:taunormdown is large}, we may assume that $ \alphanorm_t \ge (1-\cnormdown)\alphanorm_0 $ for all $ 0\le t \le \frac{C\log n}{\alphanorm_0} $ with a probability larger than $1-O(n^{-11})$. 
  From \cref{lem:deltaupweak}, for some $ T_1:=O(1/\alphanorm_0) $, we have $ \delta_{T_1} \ge (1+\cdeltaup)\cdot \delta_{0} $ or $ \taujweak \le T_1 $ with probability $ 1-O(n^{-11}) $.
  By repeating this argument for $ \log_{1+\cdeltaup} n $ times, it must hold that $ \taujweak \le O(\log n/\alphanorm_0) $ with probability $ 1-O(n^{-11}/\log n) $.
\end{proof}
\begin{proof}[Proof for 2-Choices]
  First, from \cref{lem:taunormdown is large}, we may assume that $ \alphanorm_t \ge (1-\cnormdown)\alphanorm_0 $ for all $ 0\le t \le \frac{C\log n}{\alphanorm_0} $ with a probability larger than $1-O(n^{-11})$. 
  From \cref{lem:deltaupweak}, for some $ T_1:=O(1/\alphanorm_0) $, we have $ \eta_{T_1} \ge (1+\cetaup)\cdot \eta_{0} $ or $ \taujweak \le T_1 $ with probability $ 1-O(n^{-11}) $.
  By repeating this argument for $ \log_{1+\cetaup} n $ times, it must hold that $ \taujweak \le O(\log n/\alphanorm_0) $ with probability $ 1-O(n^{-11}/\log n) $.
\end{proof}

\subsection{Additive Drift of Bias}\label{sec:additive drift of bias}
In this section, we show that the bias between two non-weak opinions increases additively even when it is small (\cref{lem:additive drift}).
Fundamentally, our approach hinges on the observation that the square of the bias exhibits an additive drift.

\begin{lemma}[Additive drift of bias] \label{lem:additive drift}
  Let $ i,j $ be distinct opinions that are not weak at round $ 0 $.
  Consider the stopping times $\tauiweak,\taujweak, \taudeltaplus, \tauetaplus $ defined in \cref{def:stopping times,def:tauetaup} and let $\cweak=1/10$.
  Then, we have the following:
\begin{enumerate}
  \renewcommand{\labelenumi}{(\roman{enumi})}
  \item For 3-Majority, 
  let $\xdelta=\frac{1}{1000\sqrt{n}}$ and 
  suppose $\alphanorm_0=\Omega\qty(\sqrt{\frac{\log n}{n}})$.
  Then, there is a positive constant $c \in (0,1)$ such that
  \begin{align*} 
    \Pr\qty[ \min\qty{ \taudeltaplus,\tauiweak,\taujweak } > \frac{1}{\alphanorm_0}] \le 1-c.
  \end{align*}
  \item For 2-Choices, 
  let 
  $\xeta=\frac{1}{2000\sqrt{\e n}}$ and 
  suppose $\alphanorm_0\geq \frac{C(\log n)^2}{n}$ for a sufficiently large constant $C>0$.
  Then, there is a  positive constant $c \in (0,1)$ such that
  \begin{align*} 
    \Pr\qty[ \min\qty{ \tauetaplus,\tauiweak,\taujweak} > \frac{1}{\alphanorm_0} ] \le 1-c.
  \end{align*}
  \end{enumerate}
\end{lemma}
Now, we introduce the following key lemmas \cref{lem:taudeltaplus OST,lem:taudeltaplus general,lem:taudeltaplus one step sync}.
The first one, \cref{lem:taudeltaplus OST}, can be deduced from a natural consequence of the optimal stopping theorem.
\begin{lemma}[Optimal stopping theorem and $\delta_t^2$]
    \label{lem:taudeltaplus OST}
    Let $ i,j $ be distinct opinions.
    Consider the stopping times $ \tauiweak,\taujweak,\tauidown,\taujdown $ defined in \cref{def:stopping times}.
    Let $\tau\defeq \min\{\tauiweak,\taujweak,\tauidown,\taujdown\}$.
    Let $\constref{lem:two strong opinions}= 1-\frac{1}{\sqrt{2(1-\cweak)}}$ be a positive constant defined in \cref{lem:two strong opinions}.
    Then, we have
    \[
    \E[\tau]\leq \frac{\E\qty[\delta_\tau(i,j)^2]}{\varianceref{lem:taudeltaplus OST}},
    \]
    where
    \begin{align*}
    \varianceref{lem:taudeltaplus OST}=
    \begin{cases*}
         \constref{lem:two strong opinions}^3(1-\calphadown)\frac{\max\{\alpha_0(i),\alpha_0(j)\}}{n}& \text{for 3-Majority},\\
         \constref{lem:two strong opinions}^2(1-\calphadown)^2\frac{\max\{\alpha_0(i),\alpha_0(j)\}^2}{n}& \text{for 2-Choices}
    \end{cases*}.
    \end{align*}
\end{lemma}
\begin{proof}[Proof of \cref{lem:taudeltaplus OST}]
    Suppose $\tau>t-1$.
    First, we show that $\Var_{t-1}\geq \varianceref{lem:taudeltaplus OST}$ for 3-Majority and 2-Choices.
    Indeed, for 3-Majority, 
    \begin{align*}
        \Var_{t-1}\qty[\delta_t]
        &\geq \constref{lem:two strong opinions}^3\frac{\alpha_{t-1}(i)+\alpha_{t-1}(j)}{n} &&(\because \textrm{\cref{item:lower bound for variance} of \cref{lem:two strong opinions}}) \nonumber\\
        &\geq \constref{lem:two strong opinions}^3(1-\calphadown)\frac{\alpha_{0}(i)+\alpha_{0}(j)}{n} &&(\because \tauidown,\taujdown>t-1) \nonumber\\
        &\geq \varianceref{lem:taudeltaplus OST}
    \end{align*}
    holds and for 2-Choices, 
    \begin{align*}
        \Var_{t-1}\qty[\delta_t]
        &\geq \constref{lem:two strong opinions}^2\frac{\alpha_{t-1}(i)^2+\alpha_{t-1}(j)^2}{n} &&(\because \textrm{\cref{item:lower bound for variance} of \cref{lem:two strong opinions}}) \nonumber\\
        &\geq \constref{lem:two strong opinions}^2(1-\calphadown)^2\frac{\alpha_{0}(i)^2+\alpha_{0}(j)^2}{n} &&(\because \tauidown,\taujdown>t-1) \nonumber\\
        &\geq \varianceref{lem:taudeltaplus OST}
    \end{align*}
    holds.
    Hence, for both models, we have
    \begin{align*}
        \E_{t-1}[\delta_t^2]
        &=\E_{t-1}[\delta_t]^2+\Var_{t-1}[\delta_t]\\
        &=\delta_{t-1}^2\qty(1+\alpha_{t-1}(i)+\alpha_{t-1}(j)-\alphanorm_{t-1})^2+\Var_{t-1}[\delta_t]\\
        &\geq \delta_{t-1}^2\qty(1+\frac{1-2\cweak}{1-\cweak}\max\{\alpha_{t-1}(i),\alpha_{t-1}(j)\})^2+\varianceref{lem:taudeltaplus OST} &&(\text{$\because$ \cref{item:lower bound for multipricative term} of \cref{lem:two strong opinions}})\\
        &\geq \delta_{t-1}^2+\varianceref{lem:taudeltaplus OST}.
    \end{align*}
    Let $X_t=\varianceref{lem:taudeltaplus OST} \cdot t-\delta_t^2$
    and $Y_t=X_{t\wedge \tau}$.
    Then,
    \begin{align*}
    \E_{t-1}\qty[Y_t-Y_{t-1}]
    =\indicator_{\tau>t-1}\E_{t-1}\qty[X_t-X_{t-1}]
    =\indicator_{\tau>t-1}\qty(\varianceref{lem:taudeltaplus OST}\cdot t-\E_{t-1}[\delta_t^2]-\varianceref{lem:taudeltaplus OST}\cdot (t-1)+\delta_{t-1}^2)
    \leq 0,
    \end{align*}
    i.e., $(Y_t)_{t\in \Nat_0}$ is a supermartingale.
    Hence, applying the optimal stopping theorem (\cref{thm:OST}),
    $\E[Y_\tau]\leq \E[Y_0]=-\delta_0^2\leq 0$.
    Furthermore, $\E[Y_\tau]=\E[X_\tau]=\varianceref{lem:taudeltaplus OST}\E[\tau]-\E[\delta_\tau^2]$ holds from definition.
    Thus, 
    \begin{align*}
        \varianceref{lem:taudeltaplus OST}\E[\tau]-\E[\delta_\tau^2]=\E[Y_\tau]\leq \E[Y_0]\leq 0
    \end{align*}
    holds and we obtain the claim.
\end{proof}
According to \cref{lem:taudeltaplus OST}, a crucial step in obtaining an upper bound on $\E[\tau]$ is to bound $\E[\delta_\tau^2]$.
We establish this bound in a special case via \cref{lem:taudeltaplus general}, while \cref{lem:taudeltaplus one step sync} covers the remaining cases in the proof of \cref{lem:additive drift}. 
In particular, bounding the jump in bias at the stopping time in the synchronous process is one of the most complicated parts of this paper. 
We have deferred the proofs of \cref{lem:taudeltaplus general,lem:taudeltaplus one step sync} to \cref{sec:deferred gap grows}.
\begin{lemma}[Bound on the bias at a stopping time]
    \label{lem:taudeltaplus general}
    Let $ i,j $ be distinct opinions.
    Consider the stopping times 
    defined in \cref{def:stopping times} and let $\tau=\min\{\taudeltaplus,\tauiweak,\taujweak,\tauiup,\taujup,\tauidown,\taujdown\}$.
    Let $\varianceref{lem:taudeltaplus OST}$ be a parameter defined in \cref{lem:taudeltaplus OST} and
    let $\constref{lem:two strong opinions}= 1-\frac{1}{\sqrt{2(1-\cweak)}}$ be a positive constant defined in \cref{lem:two strong opinions}.
    Let $\Cdeltavarplus$ be a positive constant defined by
    \begin{align*}
        \Cdeltavarplus=
        \begin{cases*}
            \frac{2(1+\calphaup)}{\constref{lem:two strong opinions}^3(1-\calphadown)}& \text{for 3-Majority}\\
            \frac{2(1+\calphaup)^2(3-2\cweak)}{\constref{lem:two strong opinions}^2(1-\calphadown)^2(1-\cweak)}& \text{for 2-Choices}.
        \end{cases*}
    \end{align*}
    Let $\constref{lem:taudeltaplus general}\geq \Cdeltavarplus/2>0$ be a sufficiently large constant such that $x\exp(-\frac{2}{\Cdeltavarplus}x)\leq \frac{1}{100}$ holds for all $x\geq \constref{lem:taudeltaplus general}$.
    Suppose that $\xdelta\geq \frac{2\log n}{n}$ and $\frac{\xdelta^2}{\varianceref{lem:taudeltaplus OST}}\geq \constref{lem:taudeltaplus general}$ hold.
    Then, 
    \begin{align*}
    \E[\delta_\tau(i,j)^2]\leq 16\xdelta^2+\frac{\varianceref{lem:taudeltaplus OST}}{2}\E[\tau].
    \end{align*}
\end{lemma}
\begin{lemma}
    \label{lem:taudeltaplus one step sync}
    Let $ i,j $ be distinct opinions.
    Consider the stopping times $\taudeltaplus,\tauiweak,\taujweak$ defined in \cref{def:stopping times}.
    Let $\tau=\min\{\taudeltaplus,\tauiweak,\taujweak\}$
    and $\varianceref{lem:taudeltaplus OST}$ be a positive parameter defined in \cref{lem:taudeltaplus OST}.
    Suppose $\frac{\xdelta^2}{\varianceref{lem:taudeltaplus OST}}\leq C$ for some positive constant $C>0$.
    Then, $\Pr\qty[\tau>1]\leq 1-c$ holds for some positive constant $c\in (0,1)$.
\end{lemma}

\begin{proof}[Proof of \cref{lem:additive drift}]
  We set the parameter values to $\calphaup,\calphadown,\cweak,\epsilon=1/10$.
  Then, the constant factors $\constrefs{lem:drift analysis for basic}{item:tauiup}$ and $\constrefs{lem:drift analysis for basic}{item:tauidown}$ appearing in \cref{lem:drift analysis for basic} become 
  $\constrefs{lem:drift analysis for basic}{item:tauiup}=\frac{ (1 - \varepsilon) \ciup }{(1+\ciup)^2}=\frac{9}{121}>\frac{1}{20}$ and 
  $\constrefs{lem:drift analysis for basic}{item:tauidown}=\frac{(1-\cweak)(1-\varepsilon)\calphadown}{\cweak (1+\ciup)^2}=\frac{81}{121}>\frac{1}{20}$.
  Hence, for $T\defeq \frac{1}{20\max\{\alpha_0(i),\alpha_0(j)\}}$, both $T<\frac{\constrefs{lem:drift analysis for basic}{item:tauiup}}{\alpha_0(i)}$ and $T<\frac{\constrefs{lem:drift analysis for basic}{item:tauidown}}{\alpha_0(j)}$ hold.
  
  Now, we consider the stopping time $\tau_0=\min\{\tauiweak,\taujweak,\tauiup,\taujup,\tauidown,\taujdown,T\}$.
  Then, we have the following:
  \begin{itemize}
    \item Suppose $ \tau_0 = \min\{\tauiup,\taujup\}$, which implies $ \tauiup \le T$ or $ \taujup \le T $. 
    From \cref{item:tauiup} of \cref{lem:drift analysis for basic}, this occurs with probability
          \begin{align*}
          \begin{cases}
              \exp(-\Omega(n\alpha_0(i)^2))+\exp(-\Omega(n\alpha_0(j)^2))\leq \exp(-\Omega(n\alphanorm_0^2))=n^{-\Omega(1)} & \text{for 3-Majority},\\
              \exp(-\Omega(n\alpha_0(i)))+\exp(-\Omega(n\alpha_0(j)))\leq \exp(-\Omega(n\alphanorm_0))=n^{-\Omega(1)} & \text{for 2-Choices}.
          \end{cases}
      \end{align*}
    \item Suppose $ \tau_0 = \min\{\tauidown,\taujdown\}$, which implies $ \tauidown\le \min\qty{\tauiup, \tauiweak, T }$ or $ \taujdown\le \min\qty{\taujup, \taujweak, T }$.
      This occurs with probability
      \begin{align*}
          \begin{cases}
              \exp(-\Omega(n\alpha_0(i)^2))+\exp(-\Omega(n\alpha_0(j)^2))\leq \exp(-\Omega(n\alphanorm_0^2))=n^{-\Omega(1)} & \text{for 3-Majority},\\
              \exp(-\Omega(n\alpha_0(i)))+\exp(-\Omega(n\alpha_0(j)))\leq \exp(-\Omega(n\alphanorm_0))=n^{-\Omega(1)} & \text{for 2-Choices}.
          \end{cases}
      \end{align*}
    \end{itemize}
    Consequently, the following holds for both models:
    \begin{align}
      \Pr\qty[\tau_0=\min\{\tauiweak,\taujweak,T\}]\geq 1- n^{-\Omega(1)}.
      \label{eq:add tauijupdown}
    \end{align}
    In the following, write $\cdeltaplus=1/1000$ for simplicity.

  \paragraph*{Proof for 3-Majority.}
%
  Let $\tau=\min\{\taudeltaplus,\tauiweak,\taujweak,\tauiup,\tauidown,\taujup,\taujdown\}$.
  In the following, we apply \cref{lem:taudeltaplus OST,lem:taudeltaplus general} for the case where 
  $\frac{\xdelta^2}{\varianceref{lem:taudeltaplus OST}}\geq \constref{lem:taudeltaplus general}$
  and \cref{lem:taudeltaplus one step sync} for the other case.
  Note that $\xdelta=\frac{\cdeltaplus}{\sqrt{n}}\geq \frac{2\log n}{n}$ holds for a sufficiently large $n$.
  For the first case, applying \cref{lem:taudeltaplus OST,lem:taudeltaplus general}, we have
  \begin{align*}
    \E[\tau]\leq \frac{\E[\delta_\tau^2]}{\varianceref{lem:taudeltaplus OST}}\leq \frac{16(\cdeltaplus)^2}{\constref{lem:two strong opinions}^3(1-\calphadown)}\cdot \frac{1}{\max\{\alpha_0(i),\alpha_0(j)\}}+\frac{\E[\tau]}{2},
  \end{align*}
  i.e., $\E[\tau]\leq \frac{32(\cdeltaplus)^2}{\constref{lem:two strong opinions}^3(1-\calphadown)}\cdot \frac{1}{\max\{\alpha_0(i),\alpha_0(j)\}}$. 
  Here, $\constref{lem:two strong opinions}>0$ is a positive constant defined in \cref{lem:two strong opinions}.
  Hence, from $\frac{64(\cdeltaplus)^2}{\constref{lem:two strong opinions}^3(1-\calphadown)}=\frac{27+12\sqrt{5}}{12500} <\frac{1}{20}$ and Markov inequality,
    \begin{align}
      \Pr\qty[\tau> T]
      \leq 
      \Pr\qty[\frac{64(\cdeltaplus)^2}{\constref{lem:two strong opinions}^3(1-\calphadown)\max\{\alpha_0(i),\alpha_0(j)\}}]
      \leq \frac{1}{2} \label{eq:taudeltaplus case1 3maj}
    \end{align}
    holds (recall $T=\frac{1}{20\max\{\alpha_0(i),\alpha_0(j)\}}$). 
    
    Recall $\tau_0=\min\{\tauiweak,\taujweak,\tauiup,\taujup,\tauidown,\taujdown,T\}$.
    From $\frac{1}{\alphanorm_0}\geq T$, \cref{eq:add tauijupdown,eq:taudeltaplus case1 3maj}
    \begin{align*}
      &\Pr\qty[\min\{\taudeltaplus,\tauiweak,\taujweak\}>\frac{1}{\alphanorm_0}] \\
      &\leq \Pr\qty[\min\{\taudeltaplus,\tauiweak,\taujweak\}>T]\\
      &=\Pr\qty[\min\{\taudeltaplus,\tauiweak,\taujweak,T\}>T\textrm{ and }\tau_0=\min\{\tauiweak,\taujweak,T\}]+n^{-\Omega(1)}\\
      &\leq \Pr\qty[\min\{\taudeltaplus,\tauiweak,\taujweak,\tauiup,\tauidown,\taujup,\taujdown,T\}>T]+n^{-\Omega(1)}\\
      &\leq \frac{1}{2}+n^{-\Omega(1)}
    \end{align*}
    holds for this case.

Next, consider the other case where $\frac{\xdelta^2}{\varianceref{lem:taudeltaplus OST}}\leq \constref{lem:taudeltaplus general}$.
  Then, from \cref{lem:taudeltaplus one step sync}, we have
  \begin{align*}
    \Pr\qty[\min\{\taudeltaplus,\tauiweak,\taujweak\}>\frac{1}{\alphanorm_0}]
    \leq \Pr\qty[\min\{\taudeltaplus,\tauiweak,\taujweak\}>1]\leq 1-c
  \end{align*}
  for some positive constant $c\in (0,1)$.
  Thus, we obtain the claim.

\paragraph*{Proof for 2-Choices.}
  Recall $\cdeltaplus=1/1000$.
  Let $\xdelta=\cdeltaplus\sqrt{\frac{\max\{\alpha_0(i),\alpha_0(j)\}}{n}}$.
  Then, we have $\xeta=\frac{\xdelta}{2\sqrt{\e \max\{\alpha_0(i),\alpha_0(j)\}}}$.
  Let $\tau=\min\{\taudeltaplus,\tauiweak,\taujweak,\tauiup,\tauidown,\taujup,\taujdown\}$.
  First, we assume $\frac{\xdelta^2}{\varianceref{lem:taudeltaplus OST}}\geq \constref{lem:taudeltaplus general}$.
  Suppose $\tau>t-1$.
  From $\alphanorm_0\geq \frac{C(\log n)^2}{n}$ for a sufficiently large constant $C>0$, we have
  \begin{align*}
    \xdelta\geq \cdeltaplus \sqrt{\frac{(1-\cweak)\alphanorm_0}{n}}
    \geq \cdeltaplus \sqrt{(1-\cweak)C}\frac{\log n}{n}
    \geq \frac{2\log n}{n}.
  \end{align*}
    Hence, applying \cref{lem:taudeltaplus OST,lem:taudeltaplus general}, we have
    \begin{align*}
      \E[\tau]\leq \frac{\E[\delta_\tau^2]}{\varianceref{lem:taudeltaplus OST}}\leq \frac{16(\cdeltaplus)^2}{\constref{lem:two strong opinions}^2(1-\calphadown)^2}\cdot \frac{1}{\max\{\alpha_0(i),\alpha_0(j)\}}+\frac{\E[\tau]}{2},
    \end{align*}
    i.e., $\E[\tau]\leq \frac{32(\cdeltaplus)^2}{\constref{lem:two strong opinions}^2(1-\calphadown)^2}\cdot \frac{1}{\max\{\alpha_0(i),\alpha_0(j)\}}$. 
    Here, $\constref{lem:two strong opinions}>0$ is a positive constant defined in \cref{lem:two strong opinions}.
    Hence, from $\frac{64(\cdeltaplus)^2}{\constref{lem:two strong opinions}^2(1-\calphadown)^2}=\frac{7+3\sqrt{5}}{11250} <\frac{1}{20}$ and Markov inequality,
    \begin{align}
      \Pr\qty[\tau>T]\leq 
      \Pr\qty[\frac{64(\cdeltaplus)^2}{\constref{lem:two strong opinions}^2(1-\calphadown)^2\max\{\alpha_0(i),\alpha_0(j)\}}]
      \leq \frac{1}{2} \label{eq:taudeltaplus case1 2choices}
    \end{align}
    holds (recall $T=\frac{1}{20\max\{\alpha_0(i),\alpha_0(j)\}}$).   
   
    Let $\tau^*=\min\qty{\tauetaplus,\tauiweak,\taujweak,\tauiup,\tauidown,\taujup,\taujdown,T}$ and write $\zeta_t=\max\{\alpha_t(i),\alpha_t(j)\}$ for convenience.
    Then, since $\sqrt{(1+\calphaup)\zeta_0}\xeta \leq 2\sqrt{\e \zeta_0}\xeta =\xdelta$,
    we have
      \begin{align}
          \Pr\qty[\tau^*>T]
          &=\Pr\qty[\tau^*>T\textrm{ and }
          \tinyforall t\leq T, \frac{\abs{\delta_t}}{\sqrt{\zeta_t}}< \xeta ]&&(\because \tauetaup>T) \nonumber\\
          &=\Pr\qty[\tau^*>T\textrm{ and }
          \tinyforall t\leq T, \abs{\delta_t}< \sqrt{(1+\calphaup)\zeta_0}\xeta] &&(\because \tauiup>T) \nonumber\\
          &\leq \Pr\qty[\tau^*>T\textrm{ and }\tinyforall t\leq T, \abs{\delta_t}<\xdelta] \nonumber\\
          &\leq \Pr\qty[\tau>T] \nonumber\\
          &\le 1/2. & & \because \text{\cref{eq:taudeltaplus case1 2choices}} \label{eq:tauetaplus case1 2choice}
      \end{align}
    Recall $\tau_0=\min\{\tauiweak,\taujweak,\tauiup,\taujup,\tauidown,\taujdown,T\}$.
    From $\frac{1}{\alphanorm_0}\geq T$, \cref{eq:add tauijupdown,eq:tauetaplus case1 2choice}, we obtain 
    \begin{align*}
      &\Pr\qty[\min\{\tauetaplus,\tauiweak,\taujweak\}>\frac{1}{\alphanorm_0}]\\
      &\leq \Pr\qty[\min\{\tauetaplus,\tauiweak,\taujweak\}>T]\\
      &=\Pr\qty[\min\{\tauetaplus,\tauiweak,\taujweak,T\}>T\textrm{ and }\tau_0=\min\{\tauiweak,\taujweak,T\}]+n^{-\Omega(1)}\\
      &\leq \Pr\qty[\min\{\tauetaplus,\tauiweak,\taujweak,\tauiup,\tauidown,\taujup,\taujdown,T\}>T]+n^{-\Omega(1)}\\
      &\leq \frac{1}{2}+n^{-\Omega(1)}
    \end{align*}
    holds for this case.
  
    Second, consider the other case where $\frac{\xdelta^2}{\varianceref{lem:taudeltaplus OST}}\leq \constref{lem:taudeltaplus general}$.
  From \cref{thm:Chernoff}, $\E[\alpha_1(i)]\leq 2\alpha_0(i)$, and $\alphanorm_0=\Omega(\log n/n)$, we have $\Pr\qty[\alpha_1(i)\geq 4\e\alpha_0(i)]\leq 2^{-4\e n\alpha_0(i)}\leq 2^{-4\e n\alphanorm_0/(1-\cweak)}=o(1)$.
  Then, from \cref{lem:taudeltaplus one step sync}, we have
  \begin{align*}
    &\Pr\qty[\min\{\tauetaplus,\tauiweak,\taujweak\}>\frac{1}{\alphanorm_0}]\\
    &\leq \Pr\qty[\min\{\tauetaplus,\tauiweak,\taujweak\}>1]\\
    &=\Pr\qty[\min\{\tauetaplus,\tauiweak,\taujweak\}>1\textrm{ and }
    \frac{\abs{\delta_1}}{\sqrt{\zeta_1}}< \xeta \textrm{ and } \zeta_1\leq 4\e\zeta_0]+o(1) &&(\because \tauetaplus>1)\nonumber\\
    &=\Pr\qty[\min\{\tauetaplus,\tauiweak,\taujweak\}>1\textrm{ and }
    \abs{\delta_1}< \sqrt{4\e\zeta_0}\xeta]+o(1) \nonumber\\
    &\leq \Pr\qty[\min\{\tauetaplus,\tauiweak,\taujweak\}>1\textrm{ and }\tinyforall t\leq T, \abs{\delta_t}<\xdelta] +o(1)\nonumber\\
    &\leq \Pr\qty[\min\{\taudeltaplus,\tauiweak,\taujweak\}>1]+o(1)\nonumber \\
    &\leq 1-c
  \end{align*}
  for some positive constant $c\in (0,1)$.
  Thus, we obtain the claim.
\end{proof}
  
\subsection{Bias Amplification} \label{sec:bias amplification}
Combining the additive and multiplicative drift components of the bias (see \cref{lem:deltaupweak,lem:additive drift}), 
we prove that for any two non‑weak opinions, the bias increases to \(\Omega(\sqrt{\log n / n})\) within \(O(\log n / \gamma_0)\) rounds.

\begin{lemma}[Bias amplification] \label{lem:gap amplification}
  Let $ i,j $ be distinct opinions. 
  Consider the stopping times $\tauiweak,\taujweak,\taudeltaplus,\tauetaplus$ defined in \cref{def:stopping times,def:tauetaup} and let $\cweak=1/10$.
  Suppose that
  \begin{align*}
    \alphanorm_0 \ge \begin{cases}
      c_0 \sqrt{\frac{\log n}{n}}	& \text{for 3-Majority},\\
      \frac{c_0(\log n)^2}{n} & \text{for 2-Choices}
    \end{cases}
  \end{align*}
  for a sufficiently large constant $c_0>0$.
  Then, for any constant $ c_{*}>0 $, there exists a large constant $ C>0 $ such that the following holds
  for $ \xdelta = \xeta = c_{*}\sqrt{\frac{\log n}{n}}$.
  \begin{enumerate}
            \renewcommand{\labelenumi}{(\roman{enumi})}
  \item For 3-Majority, we have
  \begin{align*} 
    \Pr\qty[ \min\qty{ \taudeltaplus, \tauiweak,\taujweak } > \frac{C\log n}{\alphanorm_0} ] \le O(n^{-10}).
  \end{align*}
  \item For 2-Choices, we have
  \begin{align*}
    \Pr\qty[ \min\qty{ \tauetaplus, \tauiweak,\taujweak } > \frac{C\log n}{\alphanorm_0} ] \le O(n^{-10}).
  \end{align*}
  \end{enumerate}
\end{lemma}

To this end, we use the drift analysis result due to \cite{Doerr11} that addresses both additive and multiplicative drift simultaneously.
In order to apply their results in our setting, we use the following general version.
\begin{lemma} \label{lem:nazo lemma}
Let $ (Z_t)_{t\ge 0} $ be a Markov chain over a state space $ \Omega $ associated with natural filtration $\calF=(\calF_t)_{t\ge 0}$ and let $\tau$ be any stopping time with respect to $\calF$.
Let $ \varphi \colon \Omega \to \Real_{\ge 0} $ be a function.
For a parameter $ x\in\Real_{\ge 0} $, let $ \tauphiplus(x) = \inf\qty{t\ge 0 \colon \varphi(Z_t) \ge x} $.
Let $ T,x_0,\cphiup>0 $ be parameters and suppose that the following holds:
\begin{enumerate}
            \renewcommand{\labelenumi}{(\roman{enumi})}
\item There exists $ C_1>0 $ such that for any $ z \in \Omega $, \[ \Pr\qty[\min\qty{\tauphiplus(x_0), \tau}\le T\condition Z_0=z] \ge C_1. \] 
\item Define $ \tauphiup = \inf\qty{ t\ge 0 \colon \varphi(Z_t) \ge (1+\cphiup) \cdot \varphi(Z_0) } $.
Then, there exists $ C_2>0$ such that for any $ z \in \Omega $, \[ \Pr\qty[ \min\qty{\tauphiup, \tau} \le T \condition Z_0 = z ] \ge 1-\exp\qty( - C_2 \varphi(z)^2 ).\] 
\end{enumerate}
Then, there exists $ C = C(C_1,C_2,\cphiup,x_0) >0 $ such that,
for any $ x^* > x_0$, any $ z \in \Omega $ and any $ \varepsilon>0 $, we have
\begin{align*} 
  \Pr\qty[ \min\qty{\tauphiplus(x^*), \tau} \le C\cdot T\cdot \qty(\log(1/\varepsilon) +  \log(x^*/x_0)) \condition Z_0 = z] \ge 1-\varepsilon.
\end{align*}
\end{lemma}
Readers are encouraged to think of $\Omega$ as the set of all configurations $[k]^V$,
$Z_t \in [k]^V$ is the configuration at the $t$-th round, $\tau=\min\qty{\tauiweak,\taujweak,\taunormdown}$ (thus we can set $T=O(1/\alphanorm_0)$ for both models),
and $\varphi$ is a function that maps a configuration to bias between two specific opinions:
Specifically, for 3-Majority we consider $\varphi(Z_t)=\sqrt{n}\cdot \abs{\delta_t}$, and for 2-Choices we consider $\varphi(Z_t)=\sqrt{n}\cdot \abs{\eta_t}$ (the factor $ \sqrt{n} $ is because we can set $ x_0 $ as a constant).

Intuitively speaking, the first condition of \cref{lem:nazo lemma} refers to the additive drift of $\varphi(Z_t)$, which means that $\varphi(Z_t)$ becomes at least $x_0$ with probability $C_1=\Omega(1)$ even if we start with $\varphi(Z_0)=0$.
The second condition asserts the multiplicative drift of $\varphi(Z_t)$ since it means that $\varphi(Z_t)$ becomes at least $(1+\cphiup)\cdot \varphi(Z_0)$ within $T$ rounds with probability $1-\exp(\Omega(\varphi(Z_0)^2))$.
The proof of \cref{lem:nazo lemma} is presented in \cref{sec:proof of nazo lemma}

\begin{proof}[Proof of \cref{lem:gap amplification}.]
  For simplicity, we prove for 3-Majority.
  The proof for 2-Choices can be obtained by the same way.
  We apply \cref{lem:nazo lemma}, where each $ Z_t \in [k]^V $ is an opinion configuration,
  $ \tau = \min\{\tauiweak,\taujweak,\taunormdown\} $, $ \varphi(Z_t) = \sqrt{n}\cdot \abs{\delta_t} $.
  From \cref{lem:deltaupweak,lem:additive drift}, for a sufficiently large constant $ C_0>0 $ and for $ T_0 = \frac{C_0}{\max\{\alpha_0(i),\alpha_0(j)\}} $, we have
  \begin{equation} \label{eq:inequalities of nazolemma}
    \begin{aligned}
    &\Pr\qty[ \min\qty{ \taudeltaplus,\tau } > T_0 ] \le 1-\Omega(1),\\
    &\Pr\qty[ \min\qty{ \taudeltaup, \tau } > T_0 ] \le \exp\qty( -\Omega(n\alphanorm_0^2) ) + \exp\qty( -\Omega(n\delta_0^2) ) \le \exp\qty( -\Omega(\varphi(Z_0)^2) ).
    \end{aligned}
  \end{equation}
  Here, note that $ \alphanorm_0 = \omega(\sqrt{\log n/n}) $ and thus the term $ \exp\qty( -\Omega(n\alphanorm_0^2) ) = n^{-\omega(1)}$ is negligible.
  By definition of $ \tau $, for any $ t < \tau $, we have $ \max\qty{\alpha_t(i),\alpha_t(j)} \ge (1-\cweak)\alphanorm_0 $.
  Therefore, \cref{eq:inequalities of nazolemma} holds even if we replace $ T_0 $ by $ T\defeq \frac{C_0}{(1-\cweak)\alphanorm_0} $.

  Let $ c >0$ be an arbitrary large constant (as considered in \cref{lem:gap amplification}).
  From \cref{lem:nazo lemma} for $ \varepsilon = n^{-10} $, $ x_0 = \cdeltaplus $ (the constant of \cref{lem:additive drift}) and $ x^*=c\sqrt{\log n} $, with probability $ 1-O(n^{-10}) $,
  within $ O(T\cdot \log n) $ rounds, we have either $ t= \tau = \min\{\tauiweak,\taujweak,\taunormdown\}$ or $ \varphi(Z_t) \ge c\sqrt{\log n} $, i.e., $ \abs{\delta_t} \ge c\sqrt{\log n / n} $.
  From \cref{lem:taunormdown is large}, the event $ t=\taunormdown $ does not occur during the consecutive $ O(T\cdot \log n) $ rounds with probability $ 1-O(n^{-10}) $.
  Therefore, we obtain the claim.
\end{proof}

\subsection{Growth of the Norm} \label{sec:growth of norm}
This section gives the proof of the following lemma, which is a generalized version of \cref{thm:growth of alphanorm}.
\begin{lemma}[Growth of $ \alphanorm_t $] \label{lem:taunormplus}
  Consider the stopping time $ \taunormplus $ defined in \cref{def:stopping times}.
  Let $C>\mathrm{e}^{-1}$ and $\epsilon\in(0,1)$ be arbitrary positive constants and suppose $\frac{C^2\lg^2n}{n}\leq \xnorm\leq 1-\varepsilon$.
  Then,
  \begin{align*} 
     \Pr\qty[ \taunormplus \ge T ] \le 
     \begin{cases}
      \frac{64\e^2}{\epsilon}\cdot  \frac{\xnorm n}{T}	& \text{for 3-Majority},\\[10pt]
      \frac{192\e^2}{\epsilon^2}\cdot  \frac{\xnorm n^2}{T}	& \text{for 2-Choices}.
     \end{cases}
  \end{align*}
\end{lemma}
The proof of \cref{lem:taunormplus} is obtained by applying the following two lemmas.
First, we present the following lemma, which is a natural consequence of the optimal stopping theorem.
\begin{lemma}[Optimal stopping theorem and $\alphanorm_t$]
    \label{lem:taunormplus OST}
    Consider 3-Majority or 2-Choices.
    Suppose $\xnorm\leq 1-\varepsilon$ for a positive constant $\varepsilon\in (0,1)$.
Let $\driftnorm$ be a positive parameter defined by
\begin{align*}
    \driftnorm=
    \begin{cases*}
        \frac{\varepsilon}{n} & \text{for 3-Majority},\\
        \frac{\varepsilon^2}{3n^2} & \text{for 2-Choices}.
    \end{cases*}
\end{align*}
Then, 
\begin{align*}
    \E\qty[\taunormplus]
    \leq \frac{\E\qty[\alphanorm_{\taunormplus}]}{\driftnorm}.
\end{align*}
\end{lemma}
\begin{proof}[Proof of \cref{lem:taunormplus OST}]
    Let $\tau=\taunormplus$, $X_t=\driftnorm t-\alphanorm_t$ and $Y_t=X_{t\wedge \tau}$.
    Then, from the definition of $\driftnorm$ and \cref{item:expectation of alphanorm} of \cref{lem:basic inequality},
    \begin{align*}
    \E_{t-1}\qty[Y_t-Y_{t-1}]
    &=\indicator_{\tau>t-1}\qty(\driftnorm t-\E_{t-1}[\alphanorm_t]-\driftnorm (t-1)+\alphanorm_{t-1})
    =\indicator_{\tau>t-1}\qty(\alphanorm_{t-1}+\driftnorm-\E_{t-1}[\alphanorm_t])
    \leq 0,
    \end{align*}
    i.e., $(Y_t)_{t\in \Nat_0}$ is a supermartingale.
    Note that, for 2-Choices, 
    $\E_{t-1}\qty[\alphanorm_t]-\alphanorm_{t-1}\geq \frac{(1-\sqrt{1-\epsilon})\epsilon}{n^2}\geq \frac{\epsilon^2}{3n}$ from $1-\sqrt{1-\varepsilon}\geq 1-\exp(-\varepsilon/2) \geq 1-\frac{1}{1+\varepsilon/2}\geq \frac{\varepsilon}{3}$.
    
    From the optimal stopping theorem (\cref{thm:OST}), we have $\E[Y_\tau]\leq \E[Y_0]=-\alphanorm_0\leq 0$ 
    Furthermore, $\E[Y_\tau]=\E[X_\tau]=\driftnorm \E[\tau]-\E[\alphanorm_\tau]$ holds.
    Thus, 
    \begin{align*}
        \E[\tau]&=\frac{\E[Y_\tau]+\E[\alphanorm_\tau]}{\driftnorm}
        \leq \frac{\E[\alphanorm_\tau]}{\driftnorm}
    \end{align*}
    holds and we obtain the claim.
\end{proof}
The key tool for synchronous processes for \cref{lem:taunormplus} is the following lemma, which provides an appropriate upper bound for $\E[\alphanorm_{\taunormplus}]$.
We put the proof in \cref{sec:deferred norm grows}.
\begin{lemma}[Bound on the norm at a stopping time]
    \label{lem:taunormplus expectation of norm}
    Consider the stopping time $\taunormplus$ defined in \cref{def:stopping times}.
    For any positive parameter $C$, 
    \begin{align*}
        \E\qty[\alphanorm_{\taunormplus}]\leq 
        16\e^2\qty(\xnorm+\frac{C^2\lg^2n}{n})+2n^{-4\e C+1}\E\qty[\taunormplus].
    \end{align*}
\end{lemma}
\begin{proof}[Proof of \cref{lem:taunormplus}]
    In the following, write $\tau=\taunormplus$ for convenience.

    \paragraph{3-Majority.}
    From \cref{lem:taunormplus OST,lem:taunormplus expectation of norm}, the following holds for a sufficiently large $n$:
    \begin{align*}
        \E[\tau]
        &\leq \frac{n}{\epsilon}\cdot \qty(16\e^2\qty(\xnorm+\frac{C^2\lg^2n}{n})+2n^{-4\e C+1}\E[\tau])
        \leq \frac{32\e^2}{\epsilon}\xnorm n +\frac{1}{2}\E[\tau]
    \end{align*}
    Hence, $\E[\tau]\leq \frac{64\e^2}{\epsilon}\xnorm n$ and we obtain the claim from the Markov inequality.

    \paragraph{2-Choices.}
    From \cref{lem:taunormplus OST,lem:taunormplus expectation of norm}, the following holds for a sufficiently large $n$:
    \begin{align*}
        \E[\tau]
        &\leq \frac{3n^2}{\epsilon^2}\cdot \qty(16\e^2\qty(\xnorm+\frac{C^2\lg^2n}{n})+2n^{-4\e C+1}\E[\tau])
        \leq \frac{96\e^2}{\epsilon^2}\xnorm n^2 +\frac{1}{2}\E[\tau].
    \end{align*}
    Hence, $\E[\tau]\leq \frac{192\e^2}{\epsilon^2}\xnorm n^2$ and we obtain the claim from the Markov inequality.
  \end{proof}


\subsection{Putting All Together} \label{sec:proof of main theorem}
We are now ready to prove the main theorem.

\begin{proof}[Proof of \cref{thm:consensus time large alphanorm}.]
  Consider 3-Majority (the proof for 2-Choices is similar).
  Suppose that the initial configuration satisfies $ \alphanorm_0 \ge  \frac{C\log n}{\sqrt{n}}$, where $C>0$ is a sufficiently large constant.
  Fix any two distinct opinions $ i,j $.
  From \cref{lem:taunormdown is large}, we may assume that $\alphanorm_t =\Omega(\alphanorm_0)$ during the process.
  From \cref{lem:gap amplification,lem:initial bias weak}, either $i$ or $j$ becomes weak within $O((\log n)/\alphanorm_0)$ rounds with probability $1-O(n^{-10})$.
  Moreover, from \cref{lem:weakvanish} with probability $ 1-O(n^{-10}) $, every weak opinion vanishes within $ O((\log n)/\alphanorm_0) $ rounds.
  Therefore, by the union bound over $ i\ne j $, with probability $ 1-O(n^{-8}) $,
    for any pair of distinct opinions, either $ i $ or $ j $ vanishes within $ O((\log n)/\alphanorm_0) $ rounds.
  This completes the proof.
\end{proof}

\begin{proof}[Proof of \cref{thm:growth of alphanorm}.]
  Consider 3-Majority (the proof for 2-Choices is similar).
  Let $C>0$ be a sufficiently large constant.
  From \cref{lem:taunormplus} for $\xnorm = C\log n / \sqrt{n}$, we have $\alphanorm_t \ge C\log n/\sqrt{n}$ with probability $1/2$ within $t=O(\sqrt{n}\log n)$ rounds.
  Repeating this argument for $O(\log n)$ times,
  with high probability, we have $\alphanorm_T \ge C\log n / \sqrt{n}$ for some $T=O(\sqrt{n}(\log n)^{2})$.  
\end{proof}

\begin{proof}[Proof of \cref{thm:plurality}]
  Consider 3-Majority (the proof for 2-Choices is similar).
    Consider an arbitrary opinion $j\neq 1$.
    From \cref{lem:initial bias weak}, $\Pr[\taujweak>T_1]\leq O(n^{-10})$ holds for some $T_1\leq \frac{C_1\log n}{\alphanorm_0}$.
    Furthermore, from \cref{lem:weakvanish}, $\Pr[\alpha_{T_1+T_2}(j)>0]\leq O(n^{-10})$ holds for some $T_2\leq \frac{C_2\log n}{\alphanorm_0}$.
    Note that $\Pr[\alphanorm_{T_1}\leq  (1-\cnormdown)\alphanorm_0]\leq O(n^{-10})$ from \cref{lem:taunormdown is large}.
    Hence, from the union bound, $\Pr[\bigvee_{j\neq 1}\qty{\alpha_{T_1+T_2}(j)>0}]\leq O(n^{-1})$ holds and we obtain the claim.
\end{proof}

\begin{proof}[Proof of \cref{thm:lowerbound}]
  Suppose $\alpha_0(i)=1/k$ for all $i\in [k]$.
  Then, from \cref{item:tauiup} of \cref{lem:drift analysis for basic}, we obtain
  \begin{align*}
    \Pr\qty[\taucons \leq  \constrefs{lem:drift analysis for basic}{item:tauiup} k]
    &\leq \Pr\qty[\exists i\in [k], \tauiup \leq  \frac{\constrefs{lem:drift analysis for basic}{item:tauiup}}{\alpha_0(i)} ]\\
    &\leq k \cdot
    \begin{cases*}
      \exp(-\Omega(n\alpha_0(i)^2)) &\text{for 3-Majority} \\
      \exp(-\Omega(n\alpha_0(i))) &\text{for 2-Choices} \\
    \end{cases*}\\
    &\leq n^{-1}
  \end{align*}
  holds for a sufficiently small constant $c>0$.
\end{proof}

\begin{proof}[Proof of \cref{thm:main theorem}]
  The consensus time bound follows from \cref{thm:consensus time large alphanorm,thm:growth of alphanorm}, and the plurality consensus follows from \cref{thm:plurality}.
  The lower bound follows from \cref{thm:lowerbound}:
  If $k\le c\sqrt{\frac{n}{\log n}}$, \cref{thm:lowerbound} ensures that the consensus time is $\Omega(k)$ with high probability.
  Otherwise, we can consider the balanced configuration with $c\sqrt{n/\log n}$ opinions as the initial configuration, which requires $\Omega\qty(\sqrt{\frac{n}{\log n}})$ to reach consensus with high probability from \cref{thm:lowerbound}.
\end{proof}

\section*{Acknowledgements}

The authors are grateful to Colin Cooper, Frederik Mallmann-Trenn, and Tomasz Radzik for helpful discussions during our visit to King's College London.
Nobutaka Shimizu is supported by JSPS KAKENHI Grant Number 23K16837.
Takeharu Shiraga is supported by JSPS KAKENHI Grant Number 23K16840.

\printbibliography
\appendix

\section{Tools}
\begin{theorem}[\mbox{\cite[Corollary 1.10.4]{Doe18}}]
  \label{thm:Chernoff}
  Let $X_1, \ldots, X_N$ be independent $[0,1]$-valued random variables and let $X=\sum_{i=1}^NX_i$.
  Then, for any $z\geq 2\e\E[X]$, we have
  \begin{align*}
    \Pr\sbra*{X\geq z}\leq 2^{-z}.
  \end{align*}
\end{theorem}
\begin{theorem}[Bernstein inequality; \mbox{\cite[Theorem 2.8.4]{Ver18}}]
  \label{thm:Bernstein}
  Let $X_1, \ldots, X_N$ be independent mean-zero random variables such that $\abs*{X_i}\leq \bounded$ for all $i$.
  Let $X=\sum_{i=1}^NX_i$.
  Then, for any $z\geq 0$, we have
  \begin{align*}
    \Pr\sbra*{\abs*{X}\geq z}\leq 2\exp\left(-\frac{z^2/2}{\Var[X]+\bounded z/3}\right).
  \end{align*}
\end{theorem}
We shall use the following results concerning (super)martingales.
\begin{theorem}[Optimal stopping theorem. See, e.g.,~Theorem 4.8.5 of \cite{Dur19}]\label{thm:OST}
        Let $(X_t)_{t\in \mathbb{N}_0}$ be a submartingale (resp.\ supermartingale) such that $\E_{t-1}\sbra*{|X_t-X_{t-1}|}<\infty$ a.s.\
        and let $\tau$ be a stopping time such that $\E[\tau]<\infty$.
        Then, $\E[X_\tau]\geq \E[X_0]$ (resp.\ $\E[X_\tau]\leq \E[X_0]$).
\end{theorem}

\begin{definition}[Negative association]
    \label{def:NA}
    Random variables $X_1,\ldots,X_n$ are \emph{negatively associated} if 
    for every two disjoint index sets $I,J\subseteq [n]$,
    \[
    \E\qty[f(X_i,i\in I)g(X_j,j\in J)]\leq \E\qty[f(X_i,i\in I)]\E\qty[g(X_j,j\in J)]
    \] 
    for all functions $f:\mathbb{R}^I\to \mathbb{R}$ and $g:\mathbb{R}^J\to \mathbb{R}$
    that are both non-decreasing.
\end{definition}
\begin{lemma}[Lemma 2 of \cite{Dubhashi1998-cf}]
    \label{lem:prodict of NA}
    Let $X_1,\ldots,X_n$ be a sequence of negatively associated random variables.
    Then for any non-decreasing functions $f_i$, $i\in [n]$,
    \[
    \E\qty[\prod_{i\in [n]}f_i(X_i)]\leq \prod_{i\in [n]}\E\qty[f_i(X_i)].
    \]
\end{lemma}
\begin{lemma}[Lemma 8 of \cite{Dubhashi1998-cf}]
    \label{lem:zero one lemma}
    Let $X_1,\ldots,X_n$ be random variables taking values in $\binset$ such that $\sum_{i\in[n]}X_i=1$.
    Then $X_1,\ldots,X_n$ are negatively associated.
\end{lemma}
\begin{proposition}[Proposition 7 of \cite{Dubhashi1998-cf}]
    \label{lem:concatination of NA}
    We have the following:
    \begin{enumerate}
        \item Let $X_1,\ldots, X_n$ and $Y_1,\ldots, Y_n$ be two sequences of negatively associated random variables that are mutually independent.
        Then $X_1,\ldots, X_n, Y_1,\ldots, Y_n$ are negatively associated.
        \item Let $X_1,\ldots, X_n$ be a sequence of negatively associated random variables.
        Let $I_1,\ldots, I_k$ be disjoint index sets for some $k$.
        For $j\in [k]$, let $h_j:\mathbb{R}^{I_j}\to \mathbb{R}$ be functions that are all non-decreasing or all non-increasing, and define $Y_j\defeq h_j(X_i,i\in I_j)$.
        Then, $Y_1,\ldots,Y_k$ are negatively associated.
        That is, non-decreasing (or non-increasing) functions of disjoint subsets of negatively associated random variables are also negatively associated.
    \end{enumerate}
\end{proposition}

\begin{definition}[Stochastic domination]
    \label{def:Stochastic domination}
    For two random variables $X$ and $Y$, we say that $Y$ stochastically dominates $X$, written as $X\preceq Y$, if for all $\lambda \in \mathbb{R}$ we have $\Pr\qty[X\leq \lambda]\leq \Pr\qty[Y\leq \lambda]$.
\end{definition}

\begin{lemma}[Lemmas 1.8.2 and 1.8.5 of \cite{Doe18}]
    \label{lem:Stochastic domination}
    We have the following:
    \begin{enumerate}
        \item If $X\preceq Y$, then $\E[f(X)]\leq \E[f(Y)]$ for any non-decreasing function $f:\mathbb{R}\to \mathbb{R}$.
        \item If $X\leq Y$, then $X\preceq Y$.
        \item If $X$ and $Y$ are identically distributed, then $X\preceq Y$.
    \end{enumerate}
\end{lemma}

\begin{lemma}[Lemma 1.8.9 of \cite{Doe18}]
    \label{lem:binomial dominance}
    We have the following:
    \begin{enumerate}
        \item If $X\sim \mathrm{Bin}(n,p)$ and $Y\sim \mathrm{Bin}(n,q)$ for $p\leq q$, then $X\preceq Y$.
        \item If $X\sim \mathrm{Bin}(n,p)$ and $Y\sim \mathrm{Bin}(m,p)$ for $n\leq m$, then $X\preceq Y$.
    \end{enumerate}
\end{lemma}

\begin{lemma}[See, e.g., p.39 of \cite{Ver18}]
\label{lem:Bernstein lemma}
For any $|z|<3$, $\mathrm{e}^z\leq 1+z+\frac{z^2/2}{1-|z|/3}$ holds.
\end{lemma}

\begin{lemma}[\mbox{\cite[Lemma 7]{Doerr11}}] \label{lem:concentration exponentially decay}
    Let $ X^{(1)},\dots,X^{(m)} $ be i.i.d.\ $ \Int_{\ge 0} $-valued random variables such that for some $ a,b>0 $ and any $ \ell\in \Int_{\ge 0} $,
    \begin{align*}
        \Pr\qty[ X^{(1)} = \ell ] \le a\cdot (1-b)^{\ell}.
    \end{align*}
    Let $ X = X^{(1)} + \dots + X^{(m)} $ and $ \mu = \E[X] $.
    Then, for some $ C=C(a,b) > 0$ and for any $ \gamma > 0$, we have
    \begin{align*}
        \Pr\qty[ X \ge (1+\gamma) \mu + Cm ] \le \exp\qty( -\frac{\gamma^2 m}{2(1+\gamma)} ).
    \end{align*}
\end{lemma}


\section{Proof of Basic Inequalities} \label{sec:proof of basic inequalities}
In this section, we show basic inequalities for the 3-Majority and 2-Choices.
To begin with, we list the basic facts for both models.

Observe that $n\alpha_t(i)=\sum_{v\in V}\indicator_{\opinion_t(v)=i}$ conditioned on the round $t-1$ is the sum of $n$ independent Bernoulli random variables $(\indicator_{\opinion_t(v)=i})_{v\in V}$.
Hence, we have
\begin{align}
    \E_{t-1}\sbra*{\alpha_t(i)}&=\frac{1}{n}\sum_{v\in V}\Pr_{t-1}\sbra*{\opinion_t(v)=i}, \label{eq:expectation of alpha sync}\\
    \Var_{t-1}\sbra*{\alpha_t(i)}&=\frac{1}{n^2}\sum_{v\in V}\Pr_{t-1}\sbra*{\opinion_t(v)=i}\Pr_{t-1}\sbra*{\opinion_t(v)\neq i}.\label{eq:variance of alpha sync}
\end{align}

For $\delta_t(i,j)$, we again observe that $n\delta_t=n\alpha_t(i)-n\alpha_t(j)=\sum_{v\in V}\rbra*{\indicator_{\opinion_t(v)=i}-\indicator_{\opinion_t(v)=j}}$ conditioned on the round $t-1$ is the sum of $n$ independent random variables $(\indicator_{\opinion_t(v)=i}-\indicator_{\opinion_t(v)=j})_{v\in V}$.
We have
\begin{align}
    \Var_{t-1}\sbra*{\delta_t(i,j)}
    &=\frac{1}{n^2}\sum_{v\in V}\Var_{t-1}\sbra*{\indicator_{\opinion_t(v)=i}-\indicator_{\opinion_t(v)=j}} \nonumber \\
    &=\frac{1}{n^2}\sum_{v\in V}\rbra*{\Var_{t-1}\sbra*{\indicator_{\opinion_t(v)=i}}+\Var_{t-1}\sbra*{\indicator_{\opinion_t(v)=j}} +2\E_{t-1}\sbra*{\indicator_{\opinion_t(v)=i}}\E_{t-1}\sbra*{\indicator_{\opinion_t(v)=j}}}\nonumber \\
    &=\Var_{t-1}\sbra*{\alpha_t(i)}+\Var_{t-1}\sbra*{\alpha_t(j)}+\frac{2}{n^2}\sum_{v\in V}\Pr_{t-1}\sbra*{\opinion_t(v)=i}\Pr_{t-1}\sbra*{\opinion_t(v)=j}. \label{eq:variance of delta}
\end{align}
Note that $\mathbf{Cov}_{t-1}\sbra*{\indicator_{\opinion_t(v)=i},\indicator_{\opinion_t(v)=j}}=-\E_{t-1}\sbra*{\indicator_{\opinion_t(v)=i}}\E_{t-1}\sbra*{\indicator_{\opinion_t(v)=j}}$ holds.

For $\alphanorm_t=\sum_{i\in [k]}\alpha_t(i)^2$, we use the following equality:
\begin{align}
    \E_{t-1}\sbra*{\alphanorm_t}
    &=\sum_{i\in [k]}\E_{t-1}\sbra*{\alpha_t(i)^2}
    =\sum_{i\in [k]}\rbra*{\E_{t-1}\sbra*{\alpha_t(i)}^2+\Var_{t-1}\sbra*{\alpha_t(i)}}. \label{eq:expectation of alpha syncnorm}
\end{align}

\subsection{\texorpdfstring{Proof of \cref{lem:basic inequality}}{Proof of basic inequalities}}
\begin{proof}[Proof of \cref{item:alpha} of \cref{lem:basic inequality}]
    Combining \cref{eq:expectation of alpha sync,eq:updating probability for 3Majority},
    \[
    \E_{t-1}\sbra*{\alpha_{t-1}(i)}=\frac{1}{n}\cdot n\cdot \alpha_{t-1}(i)\rbra*{1+\alpha_{t-1}(i)-\alphanorm_{t-1}}=\alpha_{t-1}(i)\rbra*{1+\alpha_{t-1}(i)-\alphanorm_{t-1}}
    \]
    holds for 3-Majority. 
    Combining \cref{eq:expectation of alpha sync,eq:updating probability for 2 choices},
    \[
    \E_{t-1}\sbra*{\alpha_{t-1}(i)}
    =\alpha_{t-1}(i)\rbra*{1-\alphanorm_{t-1}+\alpha_{t-1}(i)^2}
    +\rbra*{1-\alpha_{t-1}(i)}\alpha_{t-1}(i)^2
    =\alpha_{t-1}(i)\rbra*{1+\alpha_{t-1}(i)-\alphanorm_{t-1}}
    \]
    holds for 2-Choices. 

    For variance, combining \cref{eq:variance of alpha sync,eq:updating probability for 3Majority}, 
    \begin{align*}
        \Var_{t-1}[\alpha_t(i)]
        &=\frac{\alpha_{t-1}(i)\rbra*{1+\alpha_{t-1}(i)-\alphanorm_{t-1}}\qty(1-\alpha_{t-1}(i)\rbra*{1+\alpha_{t-1}(i)-\alphanorm_{t-1}})}{n}\\
        &\leq \frac{\alpha_{t-1}(i)\rbra*{1+\alpha_{t-1}(i)-\alphanorm_{t-1}}\qty(1-\alpha_{t-1}(i)+\alphanorm_{t-1})}{n}\\
        &=\frac{\alpha_{t-1}(i)\rbra*{1-\qty(\alpha_{t-1}(i)-\alphanorm_{t-1})^2}}{n}\\
        &\leq \frac{\alpha_{t-1}(i)}{n}
    \end{align*}
    holds for 3-Majority.
    Combining \cref{eq:variance of alpha sync,eq:updating probability for 3Majority}, 
    \begin{align}
        &\Var_{t-1}[\alpha_t(i)] \nonumber \\
        &= \frac{\alpha_{t-1}(i)\rbra*{1-\alphanorm_{t-1}+\alpha_{t-1}(i)^2}\rbra*{\alphanorm_{t-1}-\alpha_{t-1}(i)^2}}{n}
        +\frac{\rbra*{1-\alpha_{t-1}(i)}\alpha_{t-1}(i)^2\rbra*{1-\alpha_{t-1}(i)^2}}{n} \label{eq:var alpha 2}\\
        &\leq \frac{\alpha_{t-1}(i)\alphanorm_{t-1}}{n}
        +\frac{\alpha_{t-1}(i)^2}{n} \nonumber 
    \end{align}
    holds for 2-Choices.
\end{proof}

\begin{proof}[Proof of \cref{item:delta} of \cref{lem:basic inequality}]
   From \cref{item:alpha} of \cref{lem:basic inequality}, 
    \begin{align*}
        \E_{t-1}[\delta_{t}]&=\E_{t-1}[\alpha_t(i)]-\E_{t-1}[\alpha_t(j)]
        =(\alpha_{t-1}(i)-\alpha_{t-1}(j))\rbra*{1+\alpha_{t-1}(i)+\alpha_{t-1}(j)+\alphanorm_{t-1}}
    \end{align*}
    holds for both models.

    For variance, recall \cref{eq:variance of delta}.
    \paragraph*{3-Majority.}
    Write $f_i=\Pr_{t-1}\sbra*{\opinion_t(v)=i}=\alpha_{t-1}(i)(1+\alpha_{t-1}(i)-\alphanorm_{t-1})$ for convenience.
    From \cref{eq:variance of alpha sync}, we have $\Var_{t-1}[\alpha_t(i)]=\frac{f_i(1-f_i)}{n}$ and $\Var_{t-1}[\alpha_t(j)]=\frac{f_j(1-f_j)}{n}$. 
    Hence, from \cref{eq:variance of delta}, we obtain
    \begin{align*}
        \Var_{t-1}[\delta_t]
        &=\frac{f_i(1-f_i)}{n}+\frac{f_j(1-f_j)}{n}+\frac{2f_if_j}{n}\\
        &=\frac{f_i+f_j-(f_i-f_j)^2}{n}\\
        &\leq \frac{\alpha_{t-1}(i)(1+\alpha_{t-1}(i)-\alphanorm_{t-1})+\alpha_{t-1}(j)(1+\alpha_{t-1}(j)-\alphanorm_{t-1})}{n}\\
        &\leq \frac{2\alpha_{t-1}(i)+2\alpha_{t-1}(j)}{n}.
    \end{align*}
    \paragraph*{2-Choices.}
    First, from \cref{eq:updating probability for 2 choices}, we have
    \begin{align*}
        \frac{1}{n}\sum_{v\in V}\Pr_{t-1}\sbra*{\opinion_t(v)=i}\Pr_{t-1}\sbra*{\opinion_t(v)=j}
        &=\alpha_{t-1}(i)\rbra*{1-\alphanorm_{t-1}+\alpha_{t-1}(i)^2}\alpha_{t-1}(j)^2\\
        &\hspace{1em}+\alpha_{t-1}(j)\alpha_{t-1}(i)^2\rbra*{1-\alphanorm_{t-1}+\alpha_{t-1}(j)^2}\\
        &\hspace{1em}+\rbra*{1-\alpha_{t-1}(i)-\alpha_{t-1}(j)}\alpha_{t-1}(i)^2\alpha_{t-1}(j)^2\\
        &\leq \alpha_{t-1}(i)\alpha_{t-1}(j)^2+\alpha_{t-1}(i)^2\alpha_{t-1}(j)+\alpha_{t-1}(i)^2\alpha_{t-1}(j)^2.
    \end{align*}
    Hence, from \cref{eq:variance of delta} and \cref{item:alpha} of \cref{lem:basic inequality}, we have
    \begin{align*}
        \Var_{t-1}[\delta_t]
        &\leq \frac{\alpha_{t-1}(i)(\alpha_{t-1}(i)+\alphanorm_{t-1})}{n}
        +\frac{\alpha_{t-1}(j)(\alpha_{t-1}(j)+\alphanorm_{t-1})}{n}\\
        &\hspace{1em}+\frac{\alpha_{t-1}(i)\alpha_{t-1}(j)\qty(\alpha_{t-1}(i)+\alpha_{t-1}(j)+\alpha_{t-1}(i)\alpha_{t-1}(j))}{n}\\
        &= \frac{\alphanorm_{t-1}\qty(\alpha_{t-1}(i)+\alpha_{t-1}(j))+\qty(\alpha_{t-1}(i)+\alpha_{t-1}(j))^2-2\alpha_{t-1}(i)\alpha_{t-1}(j)}{n}\\
        &\hspace{1em}+\frac{\alpha_{t-1}(i)\alpha_{t-1}(j)\qty(\alpha_{t-1}(i)+\alpha_{t-1}(j)+\alpha_{t-1}(i)\alpha_{t-1}(j))}{n}\\
        &\leq \frac{(\alpha_{t-1}(i)+\alpha_{t-1}(j))(\alpha_{t-1}(i)+\alpha_{t-1}(j)+\alphanorm_{t-1})}{n}.
    \end{align*}
\end{proof}

\begin{proof}[Proof of \cref{item:expectation of alphanorm} of \cref{lem:basic inequality}]
    Recall \cref{eq:expectation of alpha syncnorm}.
    \paragraph*{3-Majority.}
    Write $f_i=\Pr_{t-1}\sbra*{\opinion_t(v)=i}=\alpha_{t-1}(i)(1+\alpha_{t-1}(i)-\alphanorm_{t-1})$ for convenience.
    From 
    \cref{lem:basic inequality} and 
    \cref{eq:variance of alpha sync}, we have 
    $\E_{t-1}[\alpha_t(i)]=f_i$ and 
    $\Var_{t-1}[\alpha_t(i)]=\frac{f_i(1-f_i)}{n}$.
    From \cref{eq:expectation of alpha syncnorm}, we have 
    \begin{align*}
        \E_{t-1}[\alphanorm_t]
        = \sum_{i\in [k]}\rbra*{\E_{t-1}\sbra*{\alpha_t(i)}^2 + \Var_{t-1}\sbra*{\alpha_t(i)}}
        =\sum_{i\in [k]}\rbra*{f_i^2+\frac{f_i(1-f_i)}{n}}
        =\rbra*{1-\frac{1}{n}}\sum_{i\in [k]}f_i^2+\frac{1}{n}.
    \end{align*}
    Furthermore, we have
    \begin{align}
        \sum_{i\in [k]}\alpha_{t-1}(i)^2\rbra*{1+\alpha_{t-1}(i)-\alphanorm_{t-1}}^2
        &\geq \sum_{i\in [k]}\alpha_{t-1}(i)^2+2\sum_{i\in [k]}\alpha_{t-1}(i)^2\rbra*{\alpha_{t-1}(i)-\alphanorm_{t-1}} \nonumber\\
        &=\alphanorm_{t-1}+2\rbra*{\|\alpha_{t-1}\|_3^3-\alphanorm_{t-1}^2}\nonumber \\
        &\geq \alphanorm_{t-1}. \label{eq:squared sum of alpha}
    \end{align}
        Note that $\left(\sum_{i\in [k]}\alpha_{t-1}(i)^2\right)^2\leq \sum_{i\in [k]}(\alpha_{t-1}(i)^{1/2})^2(\alpha_{t-1}(i)^{3/2})^2=\|\alpha_{t-1}\|_3^3$ holds from the Cauchy-Schwartz inequality.
   Thus, we obtain
    \begin{align}
         \E_{t-1}[\alphanorm_t]\geq \rbra*{1-\frac{1}{n}}\alphanorm_{t-1}+\frac{1}{n}
         =\alphanorm_{t-1}+\frac{1-\alphanorm_{t-1}}{n}
         \geq
         \alphanorm_{t-1}
    \end{align}
    holds and we obtain the claim.
    \paragraph*{2-Choices.}
    From \cref{item:alpha} of \cref{lem:basic inequality} and
    \cref{eq:squared sum of alpha}, we have 
    \begin{align}
       \sum_{i\in [k]}\E_{t-1}\sbra*{\alpha_t(i)}^2
        =\sum_{i\in [k]}\alpha_{t-1}(i)^2\rbra*{1+\alpha_{t-1}(i)-\alphanorm_{t-1}}^2
        \geq \alphanorm_{t-1}. \label{eq:alphanorm 2choices 1}
    \end{align}
    Furthermore, from \cref{eq:var alpha 2}, 
    \begin{align*}
    \sum_{i\in [k]}\Var_{t-1}[\alpha_t(i)]
    &\geq\sum_{i\in [k]}\frac{\rbra*{1-\alpha_{t-1}(i)}\alpha_{t-1}(i)^2\rbra*{1-\alpha_{t-1}(i)^2}}{n}\\
    &\geq \frac{\qty(1-\sqrt{\alphanorm_{t-1}})\qty(1-\alphanorm_{t-1})\alphanorm_{t-1}}{n}\\
    &\geq 0.
     \end{align*}
     Hence, from \cref{eq:expectation of alpha syncnorm}, we obtain the claim.
\end{proof}

\subsection{\texorpdfstring{Proof of \cref{lem:two strong opinions}}{Proof of Special cases}}\label{sec:special inequalities}
First, we observe the following holds for both models:
For any distinct $i,j\in [k]$ and $t-1<\min\{\tauiweak,\taujweak\}$, 
\begin{align}
    \|\alpha_{t-1}\|_\infty^2\leq \alphanorm_{t-1} \leq \frac{\min\qty{\alpha_{t-1}(i),\alpha_{t-1}(j)}}{1-\cweak} \leq \frac{1}{2(1-\cweak)}.
    \label{eq:upper bound for norm} 
\end{align}
    The first inequality is obvious from the definition of norms.
    The last inequality follows from $\min\qty{\alpha_{t-1}(i),\alpha_{t-1}(j)}\leq 1/2$.
    Furthermore, since $\alpha_{t-1}(i),\alpha_{t-1}(j)\geq (1-\cweak)\alphanorm_{t-1}$ holds, we have 
    $\min\qty{\alpha_{t-1}(i),\alpha_{t-1}(j)}\geq (1-\cweak)\alphanorm_{t-1}$ and we obtain \cref{eq:upper bound for norm}.
\begin{proof}[Proof of \cref{item:lower bound for multipricative term} of \cref{lem:two strong opinions}]
    From \cref{eq:upper bound for norm}, we obtain
    \begin{align*}
        &\alpha_{t-1}(i)+\alpha_{t-1}(j)-\alphanorm_{t-1}\\
        &\geq \max\qty{\alpha_{t-1}(i),\alpha_{t-1}(j)}+\min\qty{\alpha_{t-1}(i),\alpha_{t-1}(j)}-\frac{\min\qty{\alpha_{t-1}(i),\alpha_{t-1}(j)}}{1-\cweak}\\
        &= \max\qty{\alpha_{t-1}(i),\alpha_{t-1}(j)}-\frac{\cweak}{1-\cweak}\min\qty{\alpha_{t-1}(i),\alpha_{t-1}(j)}\\
        &\geq \frac{1-2\cweak}{1-\cweak}\max\qty{\alpha_{t-1}(i),\alpha_{t-1}(j)}.
    \end{align*}
\end{proof}
\begin{proof}[Proof of \cref{item:lower bound for variance} of \cref{lem:two strong opinions}]
    Recall \cref{eq:variance of delta}.
    \paragraph*{3-Majority.}
    From \cref{eq:variance of alpha sync,eq:updating probability for 3Majority,eq:upper bound for norm}, we have 
    \begin{align}
        \Var_{t-1}[\alpha_t(i)]
        &=\frac{\alpha_{t-1}(i)(1+\alpha_{t-1}(i)-\alphanorm_{t-1})\qty(1-\alpha_{t-1}(i)(1+\alpha_{t-1}(i)-\alphanorm_{t-1}))}{n} \nonumber\\
        &\geq \frac{\alpha_{t-1}(i)(1-\alphanorm_{t-1})\qty(1-\alpha_{t-1}(i)-\alpha_{t-1}(i)^2+\alpha_{t-1}(i)\alphanorm_{t-1})}{n} \nonumber\\
        &\geq \frac{\alpha_{t-1}(i)(1-\|\alpha_{t-1}\|_\infty)\qty(1-\alpha_{t-1}(i)-\alpha_{t-1}(i)^2+\alpha_{t-1}(i)^3)}{n} \nonumber\\
        &\geq \frac{\alpha_{t-1}(i)(1-\|\alpha_{t-1}\|_\infty)\qty(1-\|\alpha_{t-1}\|_\infty-\|\alpha_{t-1}\|_\infty^2+\|\alpha_{t-1}\|_\infty^3)}{n} \nonumber\\
        &=\frac{\alpha_{t-1}(i)(1-\|\alpha_{t-1}\|_\infty)^3(1+\|\alpha_{t-1}\|_\infty)}{n} \nonumber\\
        &\geq \qty(1-\frac{1}{\sqrt{2(1-\cweak)}})^3\frac{\alpha_{t-1}(i)}{n}. \label{eq:varalphalower}
    \end{align}
    Note that the function $f(x)=1-x-x^2+x^3$ is decreasing in range $[0,1]$.
    Since $ \Var_{t-1}\qty[\delta_t]\geq \Var_{t-1}[\alpha_t(i)]+\Var_{t-1}[\alpha_t(j)]$ holds from \cref{eq:variance of delta}, we obtain the claim.
    \paragraph*{2-Choices.}
    From \cref{eq:var alpha 2,eq:upper bound for norm}, we have
   \begin{align*}
       \Var_{t-1}[\alpha_{t}(i)]
       \geq \frac{\rbra*{1-\alpha_{t-1}(i)}\alpha_{t-1}(i)^2\rbra*{1-\alpha_{t-1}(i)^2}}{n}
       \geq \qty(1-\frac{1}{\sqrt{2(1-\cweak)}})^2\frac{\alpha_{t-1}(i)^2}{n}.
   \end{align*}
   Since $ \Var_{t-1}\qty[\delta_t]\geq \Var_{t-1}[\alpha_t(i)]+\Var_{t-1}[\alpha_t(j)]$ holds from \cref{eq:variance of delta}, we obtain the claim.
\end{proof}

\section{Deferred Proof}
\subsection{Additive Drift of the Bias}\label{sec:deferred gap grows}
\begin{proof}[Proof of \cref{lem:taudeltaplus general}]
    Write $L=16\xdelta^2$ and $x\vee y\defeq \max\{x,y\}$ for convenience.
    Obviously, we have
    \begin{align*}
        \E[\delta_{\tau}^2]
        =\E[\delta_{\tau}^2\indicator_{\delta_\tau^2\leq L}]+\E[\delta_{\tau}^2\indicator_{\delta_\tau^2> L}]
        \leq L+\E[\delta_{\tau}^2\indicator_{\delta_\tau^2> L}].
    \end{align*}
    Furthermore, 
    \begin{align*}
        \E[\delta_{\tau}^2\indicator_{\delta_\tau^2> L}]
        &=\sum_{t=1}^\infty\E\qty[\indicator_{\tau=t}\delta_{t}^2\indicator_{\delta_t^2> L}]
        \leq \sum_{t=1}^\infty\E\qty[\indicator_{\tau>t-1}\delta_{t}^2\indicator_{\delta_t^2> L}]
        =\sum_{t=1}^\infty\E\qty[\indicator_{\tau>t-1}\E_{t-1}\qty[\delta_{t}^2\indicator_{\delta_t^2> L}]]
    \end{align*}
    holds. 

    Now, we claim that
    \begin{align}
        \indicator_{\tau>t-1}\E_{t-1}[\delta_{t}^2\indicator_{\delta_t^2> L}]
        &\leq \frac{\varianceref{lem:taudeltaplus OST}}{2} \label{eq:taudeltaplusc1}
    \end{align}
    holds for a sufficiently large $n$.
    Assuming \cref{eq:taudeltaplusc1}, we obtain
    \begin{align*}
        \E[\delta_{\tau}^2]
        \leq 16\xdelta^2+\frac{\varianceref{lem:taudeltaplus OST}}{2}\sum_{t=1}^\infty\E[\indicator_{\tau>t-1}]\leq 16\xdelta^2+\frac{\varianceref{lem:taudeltaplus OST}}{2}\E[\tau]
    \end{align*}
    

    From here, we give a proof of \cref{eq:taudeltaplusc1}.
    To begin with, we show $\Var_{t-1}[\delta_t]\leq \Cdeltavarplus \varianceref{lem:taudeltaplus OST}$ for $\tau>t-1$.
    Indeed, for 3-Majority, 
    \begin{align*}
        \Var_{t-1}\qty[\delta_t]
        &\leq \frac{2(\alpha_{t-1}(i)+\alpha_{t-1}(j))}{n} &&(\because \textrm{\cref{item:delta} of \cref{lem:basic inequality}})\\
        &\leq \frac{(1+\calphaup)(\alpha_{0}(i)+\alpha_{0}(j))}{n} &&(\because \tauiup,\taujup>t-1)\\
        &\leq 2(1+\calphaup)\frac{\max\{\alpha_0(i),\alpha_0(j)\}}{n}\\
        &=\Cdeltavarplus \varianceref{lem:taudeltaplus OST}
    \end{align*}
    holds, and for 2-Choices, 
    \begin{align*}
        \Var_{t-1}\qty[\delta_t]
        &\leq \frac{(\alpha_{t-1}(i)+\alpha_{t-1}(j))(\alpha_{t-1}(i)+\alpha_{t-1}(j)+\alphanorm_{t-1})}{n} &&(\because \textrm{\cref{item:delta} of \cref{lem:basic inequality}})\\
        &\leq \frac{(\alpha_{t-1}(i)+\alpha_{t-1}(j))\qty(\alpha_{t-1}(i)+\alpha_{t-1}(j)+\frac{\alpha_{t-1}(i)}{1-\cweak})}{n}&&(\because \tauiweak>t-1)\\
        &\leq \frac{(1+\calphaup)^2(\alpha_{0}(i)+\alpha_{0}(j))\qty(\frac{2-\cweak}{1-\cweak}\alpha_{0}(i)+\alpha_{0}(j))}{n} &&(\because \tauiup,\taujup>t-1)\\
        &\leq 2(1+\calphaup)^2\frac{3-2\cweak}{1-\cweak} \cdot \frac{\max\{\alpha_0(i),\alpha_0(j)\}^2}{n}\\
        &=\Cdeltavarplus\varianceref{lem:taudeltaplus OST}.
    \end{align*}
    
    Next, we have
    \begin{align*}
        \E_{t-1}\qty[\delta_{t}^2\indicator_{\delta_t^2> L}]
        =\int_{0}^1\Pr_{t-1}\qty[\delta_{t}^2\indicator_{\delta_t^2> L}>y]dy
        =\int_{0}^1\Pr_{t-1}\qty[\delta_{t}^2>(y\vee L)]dy
        =\int_{0}^1\Pr_{t-1}\qty[\abs{\delta_{t}}>\sqrt{y\vee L}]dy.
    \end{align*}
    We observe that
    $\abs{\delta_t}\leq \abs{\delta_t-\E_{t-1}\qty[\delta_t]}+\frac{\sqrt{y\vee L}}{2}$ holds for $\taudeltaplus>t-1$, since
    $\abs{\E_{t-1}\qty[\delta_t]}\leq 2\abs{\delta_{t-1}}\leq 2\xdelta$ and
    $\abs{\delta_t}\leq \abs{\delta_t-\E_{t-1}\qty[\delta_t]}+\abs{\E_{t-1}\qty[\delta_t]}\leq \abs{\delta_t-\E_{t-1}\qty[\delta_t]}+2\xdelta\leq \abs{\delta_t-\E_{t-1}\qty[\delta_t]}+\frac{\sqrt{y\vee L}}{2}$ hold.
    Applying the Bernstein inequality (\cref{thm:Bernstein}) with $n\delta_t-\E_{t-1}[n\delta_t]$, 
    \begin{align*}
        \Pr_{t-1}\qty[\abs{\delta_{t}}>\sqrt{y\vee L}]
        &\leq \Pr_{t-1}\qty[\abs{\delta_{t}-\E_{t-1}\qty[\delta_t]}>\frac{\sqrt{y\vee L}}{2}]\\
        &\leq \Pr_{t-1}\qty[\abs{n\delta_{t}-\E_{t-1}\qty[n\delta_t]}>\frac{n\sqrt{\max\{y\vee L\}}}{2}]\\
        &\leq 2\exp\qty(-\frac{(y\vee L)n^2/4}{\Var_{t-1}[n\delta_t]+\sqrt{\max\{y\vee L\}}n/6})
        && (\because \textrm{\cref{thm:Bernstein}})\\
        &\leq 2\exp\qty(-\frac{3n(y\vee L)/2}{6n\Cdeltavarplus\varianceref{lem:taudeltaplus OST}+\sqrt{y\vee L}})
        && (\because \Var_{t-1}\qty[\delta_t]\leq \Cdeltavarplus\varianceref{lem:taudeltaplus OST})\\
        &\leq 2\exp\qty(-\frac{3}{4}n\sqrt{y\vee L})
        +2\exp\qty(-\frac{y\vee L}{8\Cdeltavarplus\varianceref{lem:taudeltaplus OST}})
    \end{align*}
    holds for $\tau>t-1$.
    By integrating each term, we obtain
    \begin{align*}
    &\int_0^1\exp\qty(-\frac{3}{4}n\sqrt{y\vee L})dy\\
    &=L\exp\qty(-\frac{3}{4}n\sqrt{L})+\int_L^1\exp\qty(-\frac{3}{4}n\sqrt{y})dy\\
    &\leq L\exp\qty(-\frac{3}{4}n\sqrt{L})+2\frac{(3/4)n\sqrt{L}+1}{(3/4)^2n^2}\exp\qty(-\frac{3}{4}n\sqrt{L})\\
    &\leq 2L\exp\qty(-\frac{3}{4}n\sqrt{L}) &&(\because L=16\xdelta^2\geq 16/n^2)
    \end{align*}
    and
    \begin{align*}
        \int_0^1\exp\qty(-\frac{\max\{L,y\}}{8\Cdeltavarplus\varianceref{lem:taudeltaplus OST}})dy
        &=L\exp\qty(-\frac{L}{8\Cdeltavarplus\varianceref{lem:taudeltaplus OST}})+\int_L^1\exp\qty(-\frac{y}{8\Cdeltavarplus\varianceref{lem:taudeltaplus OST}})dy\\
        &\leq L\exp\qty(-\frac{L}{8\Cdeltavarplus\varianceref{lem:taudeltaplus OST}})+8\Cdeltavarplus\varianceref{lem:taudeltaplus OST}\exp\qty(-\frac{L}{8\Cdeltavarplus\varianceref{lem:taudeltaplus OST}})\\
        &\leq 2L\exp\qty(-\frac{L}{8\Cdeltavarplus\varianceref{lem:taudeltaplus OST}}). &&(\because \Cdeltavarplus\leq 2\xdelta^2/\varianceref{lem:taudeltaplus OST})
    \end{align*}
    Now, we claim $\frac{\xdelta^2}{\varianceref{lem:taudeltaplus OST}}\leq n^4$ holds for both models and for a sufficiently large $n$.
    For 3-Majority, $\frac{\xdelta^2}{\varianceref{lem:taudeltaplus OST}}\leq \frac{n}{\constref{lem:two strong opinions}^3(1-\calphadown)\max\{\alpha_0(i),\alpha_0(j)\}}\leq \frac{n^2}{\constref{lem:two strong opinions}^3(1-\calphadown)}$.
    Similarly, for 2-Choices,
    $\frac{\xdelta^2}{\varianceref{lem:taudeltaplus OST}}\leq \frac{n}{\constref{lem:two strong opinions}^2(1-\calphadown)^2\max\{\alpha_0(i),\alpha_0(j)\}^2}\leq \frac{n^3}{\constref{lem:two strong opinions}^2(1-\calphadown)^2}$.
    Consequently, for $\tau>t-1$, 
    \begin{align*}
        \E_{t-1}[\delta_{t}^2\indicator_{\delta_t^2> L}]
        &\leq 2\varianceref{lem:taudeltaplus OST}\qty(\underbrace{\frac{16\xdelta^2}{\varianceref{lem:taudeltaplus OST}}}_{\leq 16n^{4}}\exp\qty(-3\underbrace{n\xdelta}_{\geq 2 \log n})
        +\underbrace{\frac{16\xdelta^2}{\varianceref{lem:taudeltaplus OST}}\exp\qty(-\frac{2\xdelta^2}{\Cdeltavarplus\varianceref{lem:taudeltaplus OST}})}_{\leq 16/100})\\
        &\leq 32\varianceref{lem:taudeltaplus OST}\qty(n^{-2}+\frac{1}{100}) \\
        &\leq \varianceref{lem:taudeltaplus OST}/2
    \end{align*}
    holds for a sufficiently large $n$ and that concludes the claim.
\end{proof}

\begin{proof}[Proof of \cref{lem:taudeltaplus one step sync}]
    From definition, $\abs{\delta_1}\leq \abs{\delta_1-\E[\delta_1]}-\abs{\E[\delta_1]}\leq \abs{\delta_1-\E[\delta_1]}-2\xdelta$ holds.
    Note that $\abs{\E[\delta_1]}\leq 2\abs{\delta_0}\leq 2\xdelta$.
    Since $n\delta_1=\sum_{v\in V}(\indicator_{\opinion_1(v)=i}-\indicator_{\opinion_1(v)=j})$ is the sum of $n$ independent random variables, $\lim_{n\to \infty}\Pr\qty[\frac{n\delta_1-\E[n\delta_1]}{\sqrt{\Var[n\delta_1]}}\leq x]=\Phi(x)$ holds from the central limit theorem.
    Here, $\Phi(x)=\int_{-\infty}^x \frac{1}{\sqrt{2\pi}}\e^{-y^2/2}dy$ is the cumulative distribution function of the standard normal distribution.
    Hence, there exists some positive constant $0<c<1$ such that
\begin{align*}
\Pr\qty[\tau>1]
&= \Pr\qty[\tau>1\textrm{ and }\abs{\delta_1}<\xdelta]\\
&\leq \Pr\qty[\abs{\delta_1-\E[\delta_1]}<3\xdelta \textrm{ and }\tau>1]\\
&=\Pr\qty[\abs{\frac{\delta_1-\E[\delta_1]}{\sqrt{\Var[\delta_1]}}}<\frac{3\xdelta}{\sqrt{\Var[\delta_1]}}\textrm{ and }\tau>1]\\
&\leq \Pr\qty[\abs{\frac{n\delta_1-\E[n\delta_1]}{\sqrt{\Var[n\delta_1]}}}<3\sqrt{C}\textrm{ and }\tau>1]
&&(\because \frac{x_\delta}{\sqrt{\Var[\delta_1]}}\leq \sqrt{C})\\
&\leq \Phi(3\sqrt{C})-\Phi(-3\sqrt{C})+o(1)\\
&\leq 1-c
\end{align*}
holds as $n\to \infty$.
Note that $\Var[\delta_1]\geq \varianceref{lem:taudeltaplus OST}$ holds for both model.
Indeed, for 3-Majority, 
\begin{align*}
    \Var[\delta_1]
    &\geq \constref{lem:two strong opinions}^3\frac{\alpha_{0}(i)+\alpha_{0}(j)}{n} &&(\because \textrm{\cref{item:lower bound for variance} of \cref{lem:two strong opinions}}) \nonumber\\
    &\geq \varianceref{lem:taudeltaplus OST}
\end{align*}
holds and for 2-Choices, 
\begin{align*}
    \Var[\delta_1]
    &\geq \constref{lem:two strong opinions}^2\frac{\alpha_{0}(i)^2+\alpha_{0}(j)^2}{n} &&(\because \textrm{\cref{item:lower bound for variance} of \cref{lem:two strong opinions}}) \nonumber\\
    &\geq \varianceref{lem:taudeltaplus OST}
\end{align*}
holds.
\end{proof}


\subsection{Bound on the Norm at a Stopping time}\label{sec:deferred norm grows}
\begin{proof}[Proof of \cref{lem:taunormplus expectation of norm}]
    Write $\tau=\taunormplus$ and $A_t(i)=n\alpha_t(i)=\sum_{v\in V}\indicator_{\opinion_t(v)=i}$.
    First, we decompose $\E[\alphanorm_\tau]$ into the following four terms as given in \cref{eq:taunormplus 1,eq:taunormplus 2,eq:taunormplus 3,eq:taunormplus 4}:
    \begin{align}
    \E[\alphanorm_\tau]
    &=\frac{1}{n^2}\sum_{i\in [k]}\E\qty[A_\tau(i)^2]
    =\frac{1}{n^2}\sum_{i\in [k]}\sum_{t=1}^\infty\E\qty[A_t(i)^2\indicator_{\tau=t}] \nonumber \\
    &=\frac{1}{n^2}\sum_{i\in [k]}\sum_{t=1}^\infty\E\qty[A_t(i)^2\indicator_{\tau=t}\indicator_{A_{t-1}(i)\geq C\lg n}\indicator_{A_{t}(i)\geq 4\e A_{t-1}(i)}] \label{eq:taunormplus 1}\\
    &+\frac{1}{n^2}\sum_{i\in [k]}\sum_{t=1}^\infty\E\qty[A_t(i)^2\indicator_{\tau=t}\indicator_{A_{t-1}(i)\geq C\lg n}\indicator_{A_{t}(i)< 4\e A_{t-1}(i)}] \label{eq:taunormplus 2}\\
    &+\frac{1}{n^2}\sum_{i\in [k]}\sum_{t=1}^\infty\E\qty[A_t(i)^2\indicator_{\tau=t}\indicator_{A_{t-1}(i)< C\lg n}\indicator_{A_{t}(i)\geq 4\e C\lg n}] \label{eq:taunormplus 3}\\
    &+\frac{1}{n^2}\sum_{i\in [k]}\sum_{t=1}^\infty\E\qty[A_t(i)^2\indicator_{\tau=t}\indicator_{A_{t-1}(i)< C\lg n}\indicator_{A_{t}(i)<4\e C\lg n}]. \label{eq:taunormplus 4}
    \end{align}
    Regarding \cref{eq:taunormplus 2,eq:taunormplus 4}, we use the following bounds that are straightforwardly derived from the definitions:
    \begin{align*}
        &\frac{1}{n^2}\sum_{i\in [k]}\sum_{t=1}^\infty\E\qty[A_t(i)^2\indicator_{\tau=t}\indicator_{A_{t-1}(i)\geq C\lg n}\indicator_{A_{t}(i)< 4\e A_{t-1}(i)}] \\
        &< \frac{1}{n^2}\sum_{i\in [k]}\sum_{t=1}^\infty\E\qty[16\e^2A_{t-1}(i)^2\indicator_{\tau=t}] & &(\because A_{t}(i)< 4\e A_{t-1}(i))\\
        &=16\e^2 \sum_{t=1}^\infty\E\qty[\alphanorm_{t-1}\indicator_{\tau=t}]\\
        &\leq 16\e^2 \sum_{t=1}^\infty\E\qty[\xnorm\indicator_{\tau=t}] & &(\because \tau=\taunormplus>t-1)\\
        &\leq 16\e^2\xnorm, 
    \end{align*}
    \begin{align*}
        &\frac{1}{n^2}\sum_{i\in [k]}\sum_{t=1}^\infty\E\qty[A_t(i)^2\indicator_{\tau=t}\indicator_{A_{t-1}(i)< C\lg n}\indicator_{A_{t}(i)<4\e C\lg n}]\\
        &< \frac{1}{n^2}\sum_{i\in [k]}\sum_{t=1}^\infty\E\qty[16\e^2C^2\lg^2n\indicator_{\tau=t}] & &(\because A_{t-1}(i)< C\lg n)\\
        &\leq \frac{16\e^2C^2\lg^2n}{n}.
    \end{align*}
    Now, we estimate \cref{eq:taunormplus 1,eq:taunormplus 3}.
    \paragraph{Bound for \cref{eq:taunormplus 1}: The case when $A_{t-1}(i)\geq C\lg n$ and $A_{t}(i)\geq 4\e A_{t-1}(i)$.}
    For \cref{eq:taunormplus 1}, we observe that
    \begin{align*}
        &\E\qty[A_t(i)^2\indicator_{\tau=t}\indicator_{A_{t-1}(i)\geq C\lg n}\indicator_{A_{t}(i)\geq 4\e A_{t-1}(i)}]\\
        &\leq \E\qty[A_t(i)^2\indicator_{\tau>t-1}\indicator_{A_{t-1}(i)\geq C\lg n}\indicator_{A_{t}(i)\geq 4\e A_{t-1}(i)}]\\
        &\leq \E\qty[\indicator_{\tau>t-1}\indicator_{A_{t-1}(i)\geq C\lg n}\E_{t-1}\qty[A_t(i)^2\indicator_{A_{t}(i)\geq 4\e A_{t-1}(i)}]]
    \end{align*}
    holds.
    For $A_{t-1}(i)\geq C\lg n$, applying \cref{thm:Chernoff} for $z=4\e A_{t-1}(i)\geq 2\e\E_{t-1}[A_t(i)]$ yields
    \begin{align*}
        \E_{t-1}\qty[A_t(i)^2\indicator_{A_{t}(i)\geq 4\e A_{t-1}(i)}]
        &=\sum_{\ell=1}^{\infty}\Pr_{t-1}\qty[A_t(i)^2\indicator_{A_{t}(i)\geq 4\e A_{t-1}(i)}\geq \ell]\\
        &=\sum_{\ell=1}^{n^2}\Pr_{t-1}\qty[A_t(i)^2\geq \ell \textrm{ and } A_{t}(i)\geq 4\e A_{t-1}(i)]\\
        &\leq n^2\Pr_{t-1}\qty[A_{t}(i)\geq 4\e A_{t-1}(i)]\\
        &\leq n^2 2^{-4\e A_{t-1}(i)} &&(\because \mathrm{\cref{thm:Chernoff}})\\
        &\leq n^2 2^{-4\e C\lg n}&&(\because A_{t-1}(i)\geq C\lg n)\\
        &=n^{-4\e C+2}.
    \end{align*}
    Thus,
    \begin{align*}
        \frac{1}{n^2}\sum_{i\in [k]}\sum_{t=1}^\infty\E\qty[A_t(i)^2\indicator_{\tau=t}\indicator_{A_{t-1}(i)\geq C\lg n}\indicator_{A_{t}(i)\geq 4\e A_{t-1}(i)}]
        &\leq \frac{1}{n^2}\sum_{i\in [k]}\sum_{t=1}^\infty\E\qty[\indicator_{\tau>t-1}n^{-4\e C+2}]\\
        &\leq n^{-4\e C+1}\E[\tau].
    \end{align*}

    \paragraph{Bound for \cref{eq:taunormplus 3}: The case when $A_{t-1}(i)< C\lg n$ and $A_{t}(i)\geq 4\e C\lg n$.}
    Similarly, for \cref{eq:taunormplus 3}, we have
    \begin{align*}
        &\E\qty[A_t(i)^2\indicator_{\tau=t}\indicator_{A_{t-1}(i)< C\lg n}\indicator_{A_{t}(i)\geq 4\e C\lg n}] \\
        &\leq \E\qty[A_t(i)^2\indicator_{\tau>t-1}\indicator_{A_{t-1}(i)< C\lg n}\indicator_{A_{t}(i)\geq 4\e C\lg n}] \\
        &\leq \E\qty[\indicator_{\tau>t-1}\indicator_{A_{t-1}(i)< C\lg n}\E_{t-1}\qty[A_t(i)^2\indicator_{A_{t}(i)\geq 4\e C\lg n}]].
    \end{align*}
    Since $2\e\E_{t-1}[A_t(i)]\leq 4\e A_{t-1}(i) <4\e C\lg n$ for $A_{t-1}(i)< C\lg n$, 
    applying \cref{thm:Chernoff} yields
    \begin{align*}
        \E_{t-1}\qty[A_t(i)^2\indicator_{A_{t}(i)\geq 4\e C\lg n}]
        &=\sum_{\ell=1}^{\infty}\Pr_{t-1}\qty[A_t(i)^2\indicator_{A_{t}(i)\geq 4\e C\lg n}\geq \ell]\\
        &=\sum_{\ell=1}^{n^2}\Pr_{t-1}\qty[A_t(i)^2\geq \ell \textrm{ and } A_{t}(i)\geq 4\e C\lg n]\\
        &\leq n^2\Pr_{t-1}\qty[A_{t}(i)\geq 4\e C\lg n]\\
        &\leq n^2 2^{-4\e C\lg n} &&(\because \mathrm{\cref{thm:Chernoff}})\\
        &=n^{-4\e C+2}
    \end{align*}
    for $A_{t-1}(i)< C\lg n$.
    Thus,
    \begin{align*}
        \frac{1}{n^2}\sum_{i\in [k]}\sum_{t=1}^\infty\E\qty[A_t(i)^2\indicator_{\tau=t}\indicator_{A_{t-1}(i)< C\lg n}\indicator_{A_{t}(i)\geq 4\e C\lg n}]
        &\leq \frac{1}{n^2}\sum_{i\in [k]}\sum_{t=1}^\infty\E\qty[\indicator_{\tau>t-1}n^{-4\e C+2}]\\
        &\leq n^{-4\e C+1}\E[\tau].
    \end{align*}
    Consequently, we obtain
    \begin{align*}
        \E[\alphanorm_\tau]
        &\leq 16\e^2\qty(\xnorm+\frac{C^2\lg^2n}{n})+2n^{-4\e C+1}\E[\tau].
    \end{align*}
\end{proof}

\subsection{Additive and Multiplicative Drift} \label{sec:proof of nazo lemma}
Our proof basically follows the proof technique of \cite{Doerr11}.
\begin{proof}[Proof of \cref{lem:nazo lemma}]
    We divide time into phases each consists of consecutive rounds of several length.
    Formally, the phase $ s $ begins at round $ \tau(s) $ and ends at round $ \tau(s+1) $, where
    \begin{align*} 
        \tau(s)  = 
            \begin{cases}
                0	& \text{for $s=0$},\\
                \inf\qty{ t \ge \tau(s-1) \colon 
                \begin{aligned}
                &&&t\ge \tau(s-1)+ T  \text{ or} \\
                &&&t\ge \tau  \text{ or}\\
                &&&\varphi(Z_t) \ge \max\qty{  x_0, (1+\cphiup)\varphi\qty(Z_{\tau(s-1)}) }
                \end{aligned}
                } & \text{for $s\ge 1$}.
            \end{cases}
    \end{align*}
    We say that the phase $ s $ is \emph{good} if it ends
    due to either the second or third condition being satisfied, i.e.,
    $ \tau(s) \ge \tau $ or $ \varphi(Z_{\tau(s)}) \ge \max\qty{  x_0, (1+\cphiup)\varphi\qty(Z_{\tau(s-1)})}$.
    Otherwise, the phase $ s $ is \emph{bad}.
    For example, if the phase $ s $ ends with the condition $ t\ge\tau $, then $ \tau(s+c) = \tau(s)$ for all $ c\in\Nat_0 $ (the length of a phase can be zero); thus, all subsequent phases are good.

    We shall count the number of consecutive good phases starting from round $ 0 $.
    By the first assumption of \cref{lem:nazo lemma}, for any $ z \in \Omega $, conditioned on $ Z_0 = z $,
    the phase $ 0 $ is good with probability $ C_1 $; then either $ \tau(1)\ge \tau $ or $ \varphi(Z_{\tau(1)}) \ge x_0$ holds.
    Again, by the second assumption (and since $ (Z_t) $ is a Markov chain), conditioned on the event that the phase $ 0 $ is good, the phase $ 1 $ is good with probability $ 1-\exp(-C_2 x_0^2) $; then either $ \tau(2) \ge \tau $ or $ \varphi(Z_{\tau(2)}) \ge (1 + \cphiup)x_0$ holds.
    By repeating this argument, conditioned on the event that the phases $0,\dots,s-1$ are good (in which case either $ \tau(s-1) \ge \tau $ or $ \varphi(Z_{\tau(s-1)}) \ge (1+\cphiup)^{s-1}x_0 $), we have that the phase $ s $ is good with probability $ 1-\exp(-C_2 (1+\cphiup)^{2s-2}\cdot x_0^2) $.
    Let $ S $ be the number of consecutive good phases starting at round $ 0 $.
    Note that $ S $ can be $ \infty $ when a phase ends with the condition $ t\ge \tau $.
    For the target value $ x^* $, let $ K\in\Nat $ be the minimum integer such that $ x_0 \cdot (1+\cphiup)^K \ge x^* $.
    Note that $ S>K $ implies that either $ t\ge \tau $ or $ \varphi(Z_t) \ge x_0 \cdot (1+\cphiup)^K \ge x^*$ for the time round $ t $ soon after the $ K+1 $ consecutive success phases, meaning that $ \varphi(Z_t) $ reaches the target value $ x^* $.

    Throughout the proof, we assume that the big-O notation hides factors depending on $C_1, C_2,x_0 $ and $ \cphiup $.
    For any $ \ell \ge 0 $ and any $ Z_0\in\Omega $, we have
    \begin{align*} 
        \Pr\qty[ S > \ell ] &= \Pr\qty[ \text{phases $0,\dots,\ell$ are good} ]  \\
        &\ge C_1 \prod_{s \in [\ell]} \qty( 1-\exp(-C_2 (1+\cphiup)^{2s-2}\cdot x_0^2) ) \\
        &\ge C_1 \prod_{s \in [\ell]} (1-p^s) & & \text{for some $ p <1$}\\
        &\ge C_1\cdot \prod_{s\in[\ell]} \exp\qty( -\frac{p^s}{1-p^s} ) & & (\because\text{$1-x \ge \e^{-\frac{x}{1-x}}$ for all $x\in[0,1)$}) \\
        &\ge C_1\cdot \exp\qty(-\frac{p}{1-p}\sum_{s\ge 0} p^s) \\
        &\ge C_1\cdot \exp\qty(-\frac{p}{(1-p)^2}) \\
        &=\Omega(1).
    \end{align*}
    Note that the inequality above holds regardless of the initial state $ Z_0 $.

    Consider the sequence $ (Z_t)_{t\ge 0} $.
    Let the number of consecutive successful phases be denoted sequentially as $S^{(0)},S^{(1)},\dots$.
    We stop the sequence $ (Z_t) $ when the number of consecutive good phases exceeds $ K $.
    Therefore, the number of phases during this process is at most $ S^{(0)} + S^{(1)} + \dots + S^{(U)} $,
        where $ U \in \Nat $ is the smallest integer such that $ S^{(U)} > K $ (here, we set $ S^{(U)}=K+1 $).
    Since each $ S^{(i)} $ satisfies $ \Pr[S^{(i)} > K] =\Omega(1)$, we have $ U= O(\log(1/\varepsilon)) $ with probability $ 1-\varepsilon $ (over randomness of $ (Z_t)$).

    We obtain an upper tail of $ S^{(0)} + \dots + S^{(U-1)} $.
    If $ \Pr[S^{(i)} > K] = 1 $, then $ U=1 $ and we have $ S^{(0)} = K+1$.
    If not, the marginal distribution of each $ S^{(i)} $ for $ i<U $ is the distribution of $ S $ conditioned on $ S\le K$.
    Moreover, for any $ \ell \le K$,
    \begin{align*}
        \Pr\qty[ S = \ell \condition S\le K] &\le \frac{\Pr\qty[S=\ell]}{\Pr\qty[S\le K]} \\
        &\le O(1)\cdot \Pr\qty[ S=\ell \condition S\ge \ell ] \\
        &=\Pr\qty[\text{phase $\ell$ is bad} \condition \text{phases $0,\dots,\ell-1$ are good}] \\
        &\le \begin{cases}
            \exp(-C_2 (1+\cphiup)^{2\ell-2}\cdot x_0^2) & \text{if $\ell\ge 1$} \\
            1-C_1 & \text{if $\ell=0$}
        \end{cases} \\
        &\le p^\ell.  & & \text{for some $ p <1$}
    \end{align*}
    In particular, $ \E\qty[S \condition S<\infty] \le O(1) $.
    Therefore, conditioned on $ U $, for a sufficiently large constant $ C'>0 $,
    from \cref{lem:concentration exponentially decay} (for $ \mu = O(U),m=U,\gamma = C'\log(1/\varepsilon)/\mu $),
    we have
    \begin{align*}
        \Pr\qty[ S^{(0)} + \dots +S^{(U-1)} \ge C'\log(1/\varepsilon) ] \le \varepsilon^{-\Omega(1)}.
    \end{align*}

    Therefore, for any $ \varepsilon>0 $, we have $ S^{(0)} + \dots + S^{(U)} = O(\log(1/\varepsilon)) + K = O(\log(1/\varepsilon) + \log(x^*/x_0)) $ with probability $ 1-\varepsilon $.
\end{proof}

\end{document}